\documentclass[10pt]{amsart}

\usepackage[all,2cell]{xy} \UseAllTwocells \SilentMatrices
\usepackage{graphicx,amssymb,amsfonts,amsmath,amsthm,newlfont}
\usepackage{epsf}
\usepackage{epsfig}

\usepackage{axodraw4j}
\usepackage{pstricks}

\newtheorem{thm}{Theorem}
\newtheorem{lem}[thm]{Lemma}
\newtheorem{prop}[thm]{Proposition}
\newtheorem{rem}[thm]{Remark}

\newtheorem{cor}[thm]{Corollary}

\theoremstyle{definition}
\newtheorem{var}[thm]{Variant}
\newtheorem{defn}[thm]{Definition}
\newtheorem{example}[thm]{Example}

\newcommand{\A}{\mathbb A}
\newcommand{\C}{\mathbb C}
\newcommand{\Q}{\mathbb Q}
\newcommand{\Z}{\mathbb Z}
\newcommand{\R}{\mathbb R}

\newcommand{\Pro}{\mathbb P}
\newcommand{\Gbar}{\overline{G}}

\newcommand{\cH}{\mathfrak{h}}

\newcommand{\pomega}{\nu}

\newcommand{\To}{\longrightarrow}

\newcommand{\Iup}{I^{\ell}}

\newcommand{\Gdot}{G^{\bullet}}
\newcommand{\ren}{\mathrm{ren}}
\newcommand{\Res}{\mathrm{Res}}
\newcommand{\omegat}{\omega^{(2)}}
\newcommand{\axi}{(\mathbf{1})}
\newcommand{\axii}{(\mathbf{2})}
\newcommand{\axiii}{(\mathbf{3})}

\newcommand{\D}{\mathfrak{D}}

\newcommand{\DD}{\sigma}

\newcommand{\diverg}{ {\tiny \hbox{ div. }}}

\title{Angles, Scales and Parametric Renormalization}

\author{Francis Brown${}^\ast$
\and 
Dirk Kreimer${}^{\ast\ast}$}
\thanks{${}^\ast$Supported by CNRS and ERC grant 257638. ${}^{\ast\ast}$Alexander von Humboldt Chair in Mathematical Physics, supported by the Alexander von Humboldt Foundation and the BMBF}
\address{Inst.Math.Jussieu\\ 175 Rue de Chevaleret\\  75013 Paris\\ France \and
Humboldt U.\\ Unter den Linden 6\\ 10099 Berlin\\ Germany}
\begin{document}
\maketitle

\bibliographystyle{plain}
\bibliography{main}

\newtheorem{exa}[thm]{Example}

\theoremstyle{definition}

%Figures%%%%%%%%%%%%%%%%%%%%%%%%%%%%%%%%%%%
\def\overlap{\;\raisebox{-12mm}{\epsfysize=36mm\epsfbox{overlap.eps}}\;}
\def\decompwthree{\;\raisebox{-8mm}{\epsfysize=24mm\epsfbox{decompwthree.eps}}\;}

\def\oldc{\;\raisebox{-10mm}{\epsfysize=24mm\epsfbox{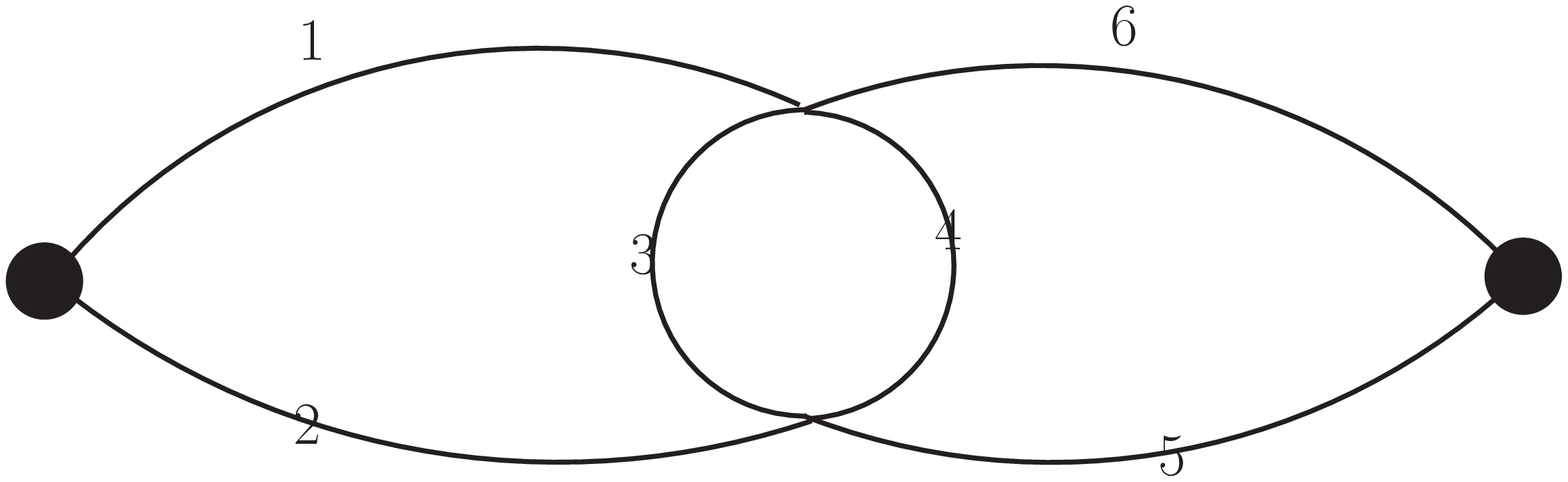}}\;}
\def\noldcl{\;\raisebox{-10mm}{\epsfysize=24mm\epsfbox{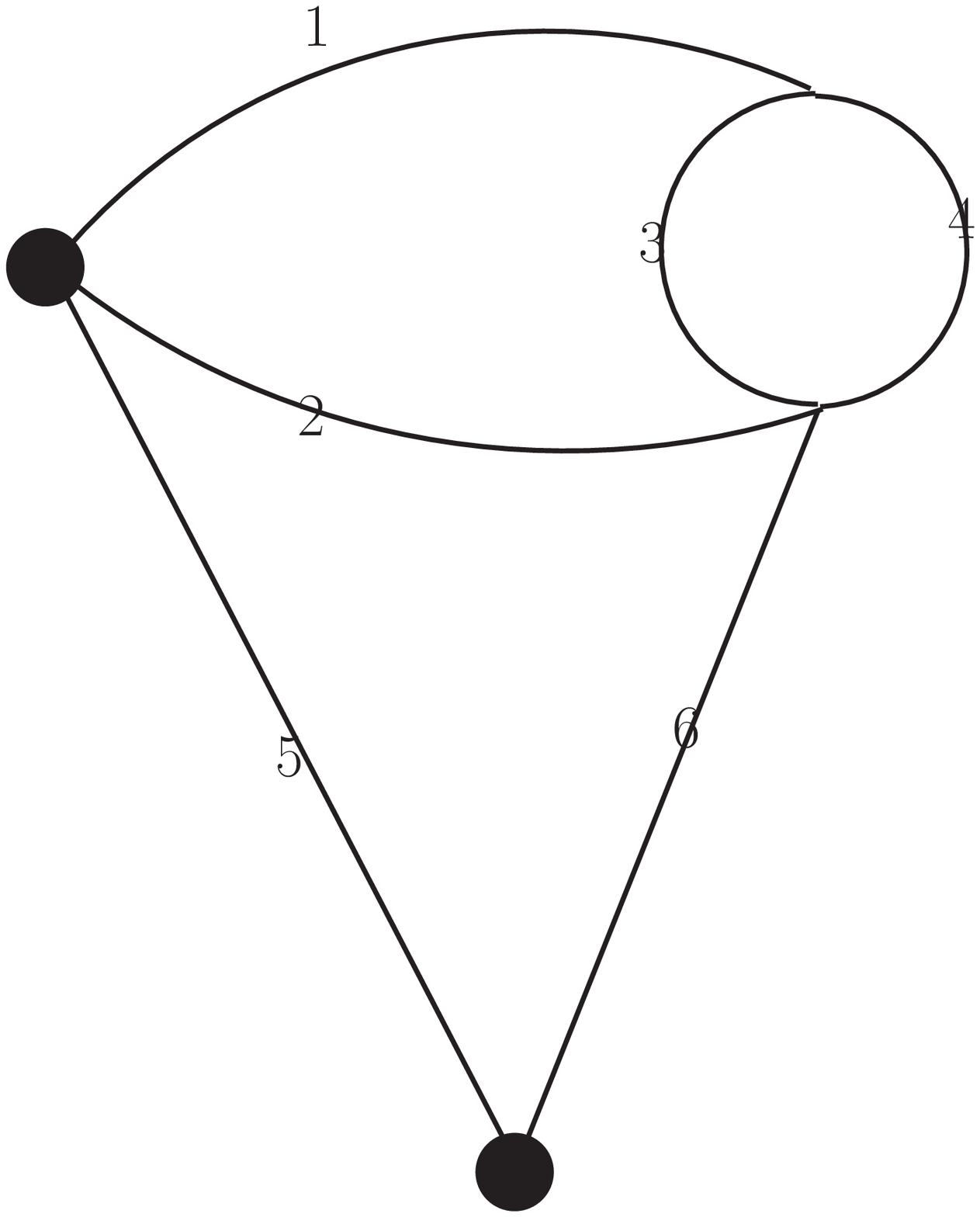}}\;}
\def\noldcr{\;\raisebox{-10mm}{\epsfysize=24mm\epsfbox{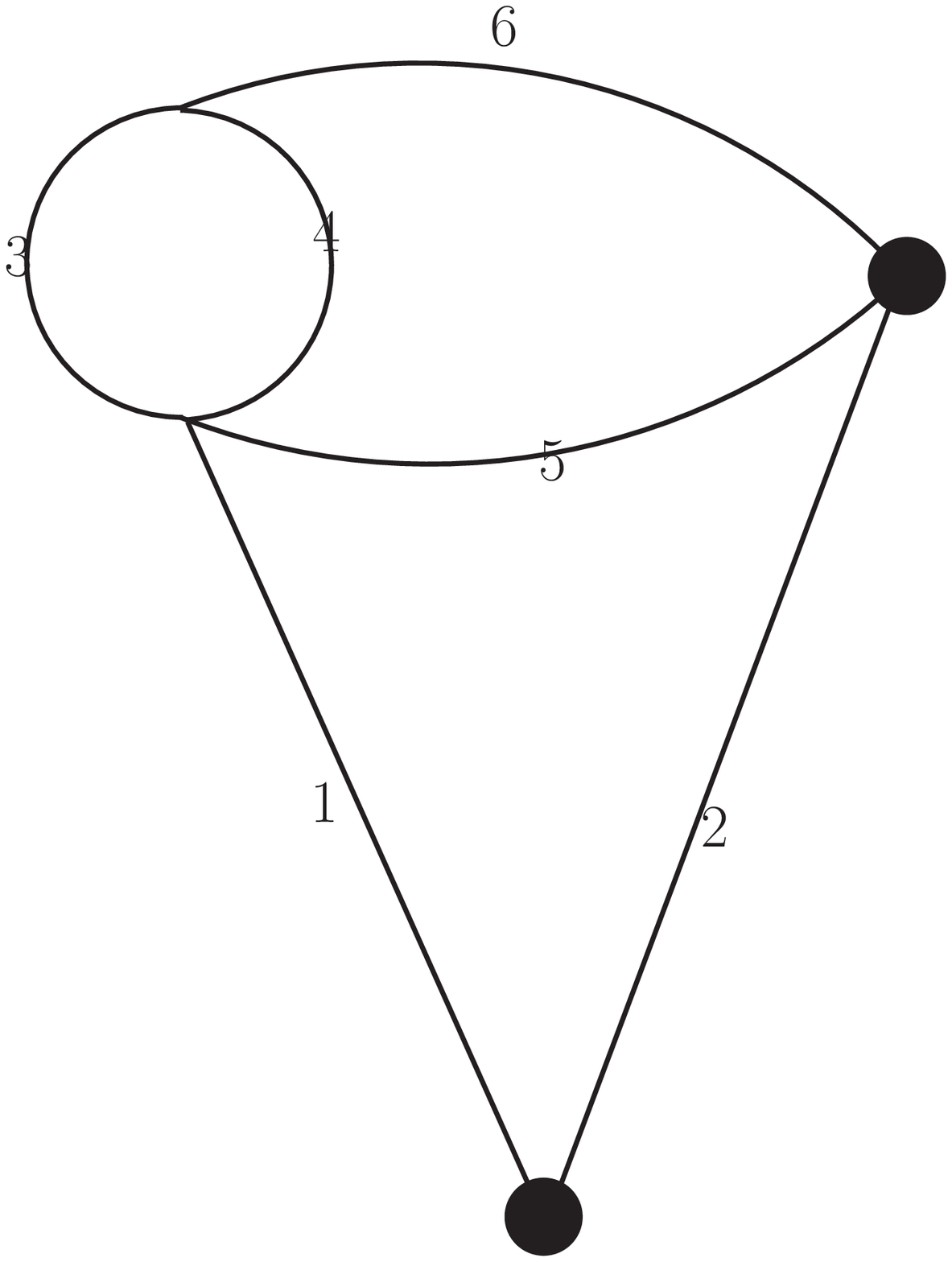}}\;}
\def\dcl{\;\raisebox{-10mm}{\epsfysize=24mm\epsfbox{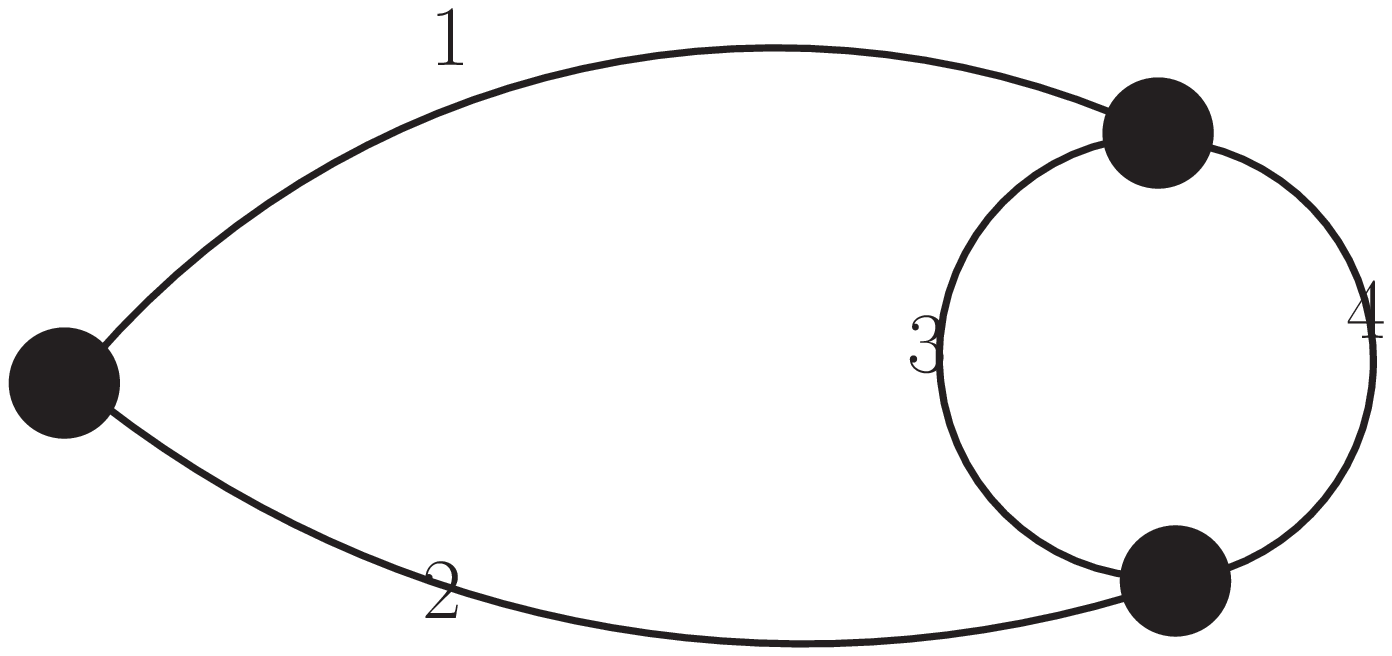}}\;}
\def\dcr{\;\raisebox{-10mm}{\epsfysize=24mm\epsfbox{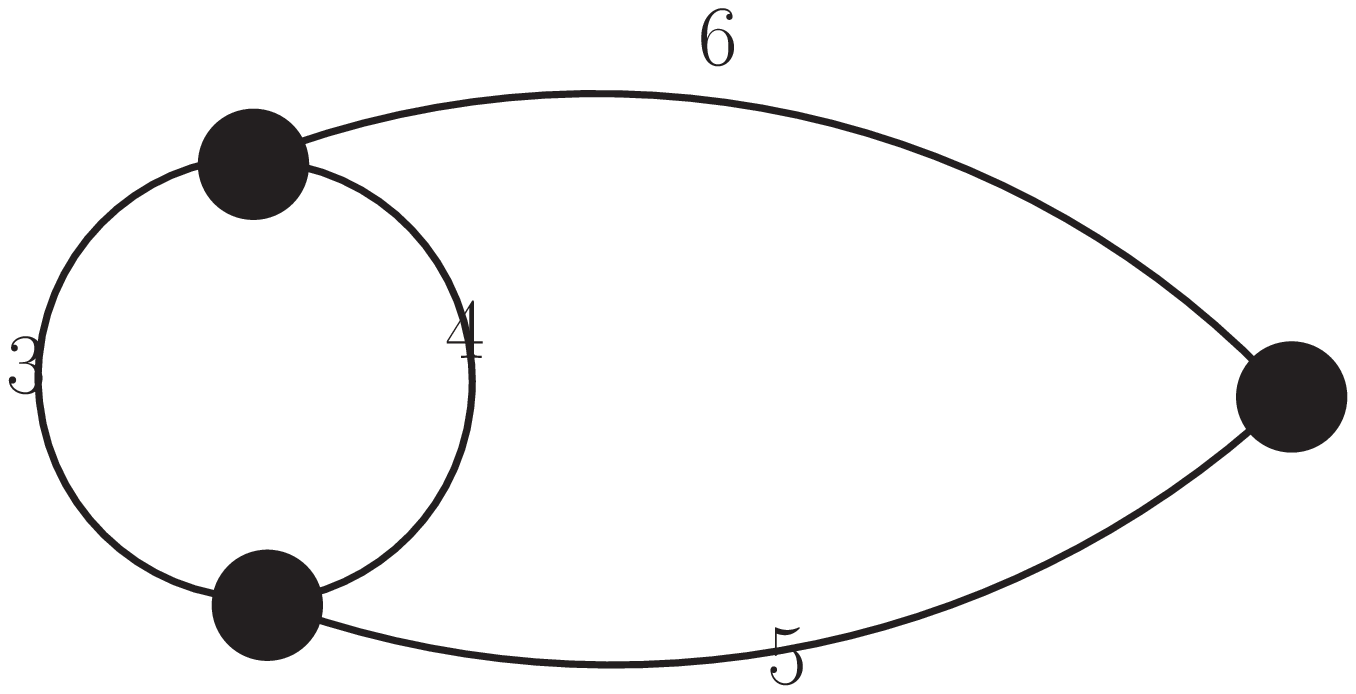}}\;}
\def\dclb{\;\raisebox{-10mm}{\epsfysize=24mm\epsfbox{dclb.eps}}\;}
\def\dcrb{\;\raisebox{-10mm}{\epsfysize=24mm\epsfbox{dcrb.eps}}\;}
\def\oot{\;\raisebox{-10mm}{\epsfysize=24mm\epsfbox{oot.eps}}\;}
\def\otf{\;\raisebox{-10mm}{\epsfysize=24mm\epsfbox{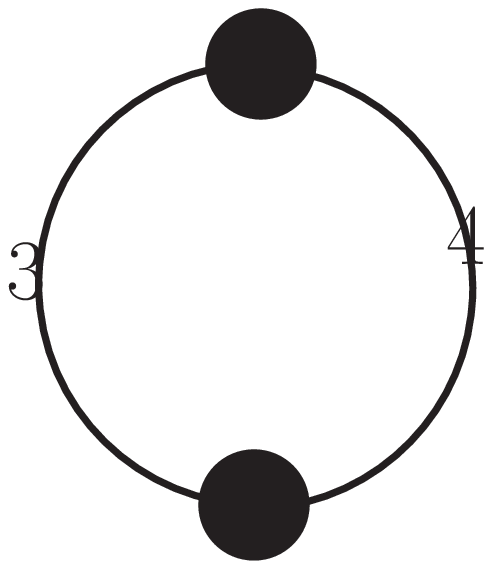}}\;}
\def\ofs{\;\raisebox{-10mm}{\epsfysize=24mm\epsfbox{ofs.eps}}\;}
\def\foroldcl{\;\raisebox{-20mm}{\epsfysize=48mm\epsfbox{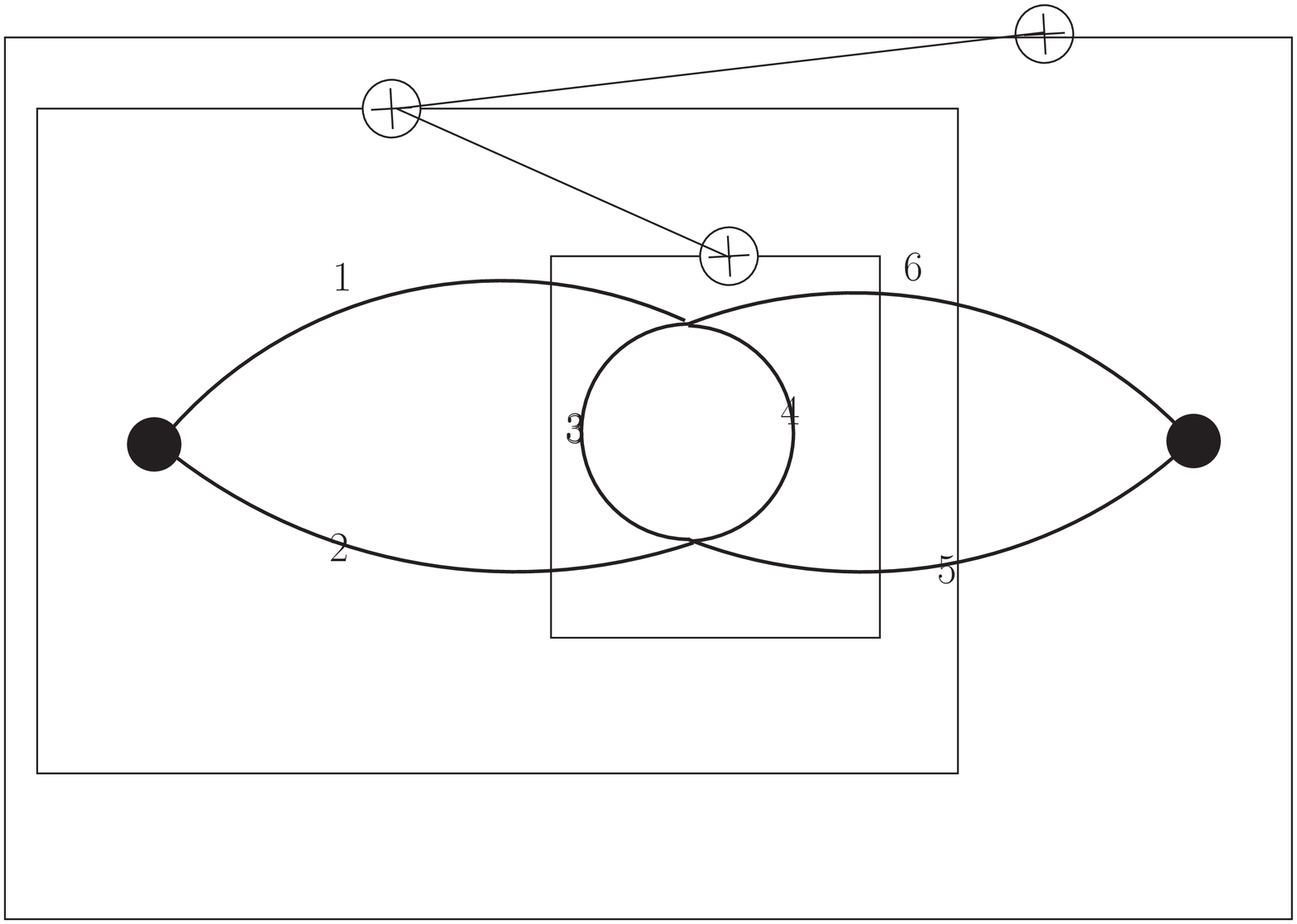}}\;}
\def\foroldcr{\;\raisebox{-20mm}{\epsfysize=48mm\epsfbox{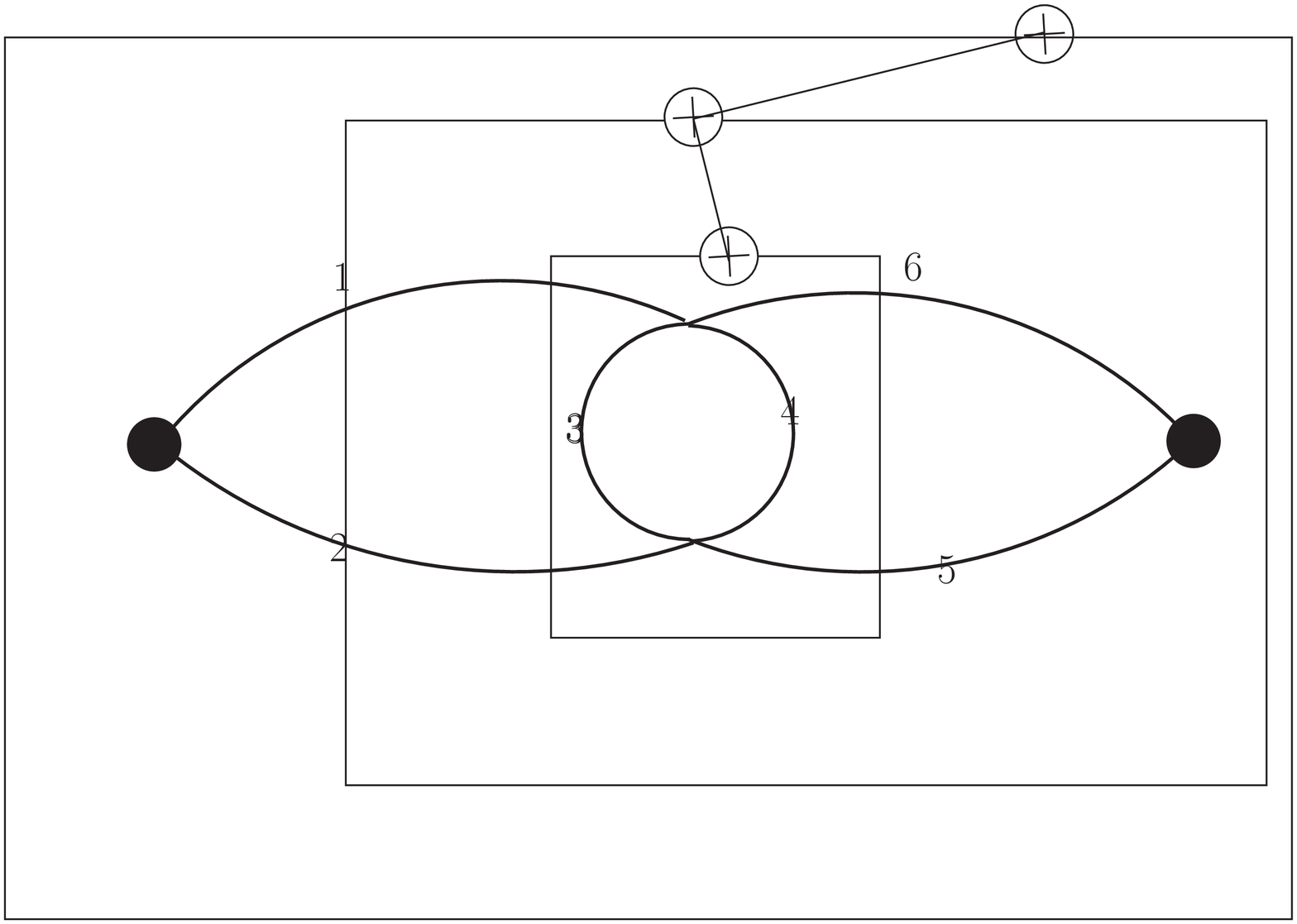}}\;}

\def\tfour{\;\raisebox{-20mm}{\epsfysize=40mm\epsfbox{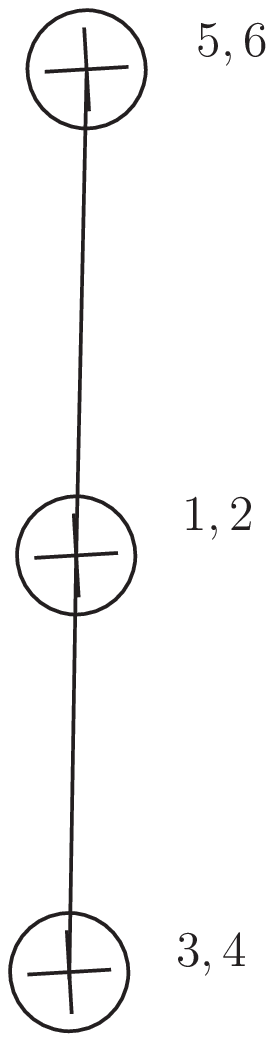}}\;}
\def\tthree{\;\raisebox{-20mm}{\epsfysize=40mm\epsfbox{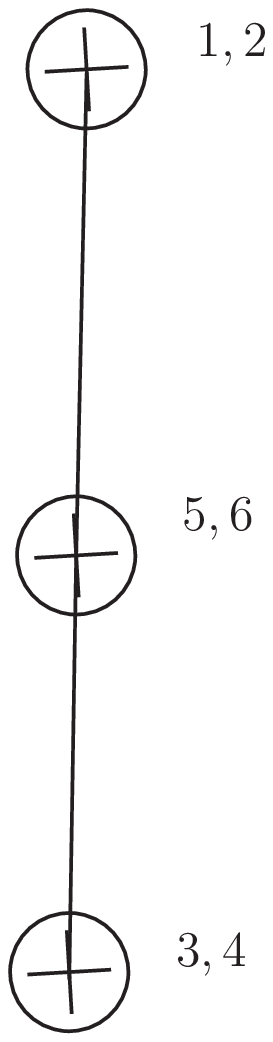}}\;}

\def\loglog{\;\raisebox{-12mm}{\epsfysize=24mm\epsfbox{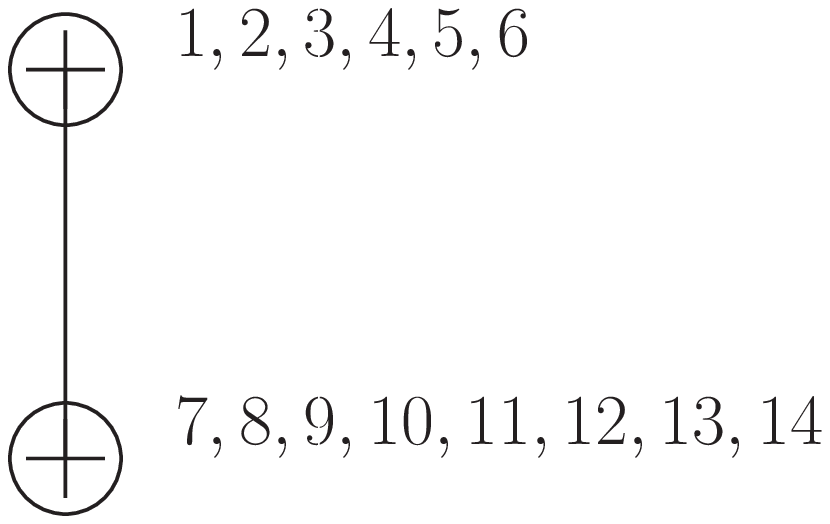}}\;}

\def\soldc{\;\raisebox{-5mm}{\epsfysize=10mm\epsfbox{oldc.eps}}\;}
\def\soldcbm{\;\raisebox{-5mm}{\epsfysize=10mm\epsfbox{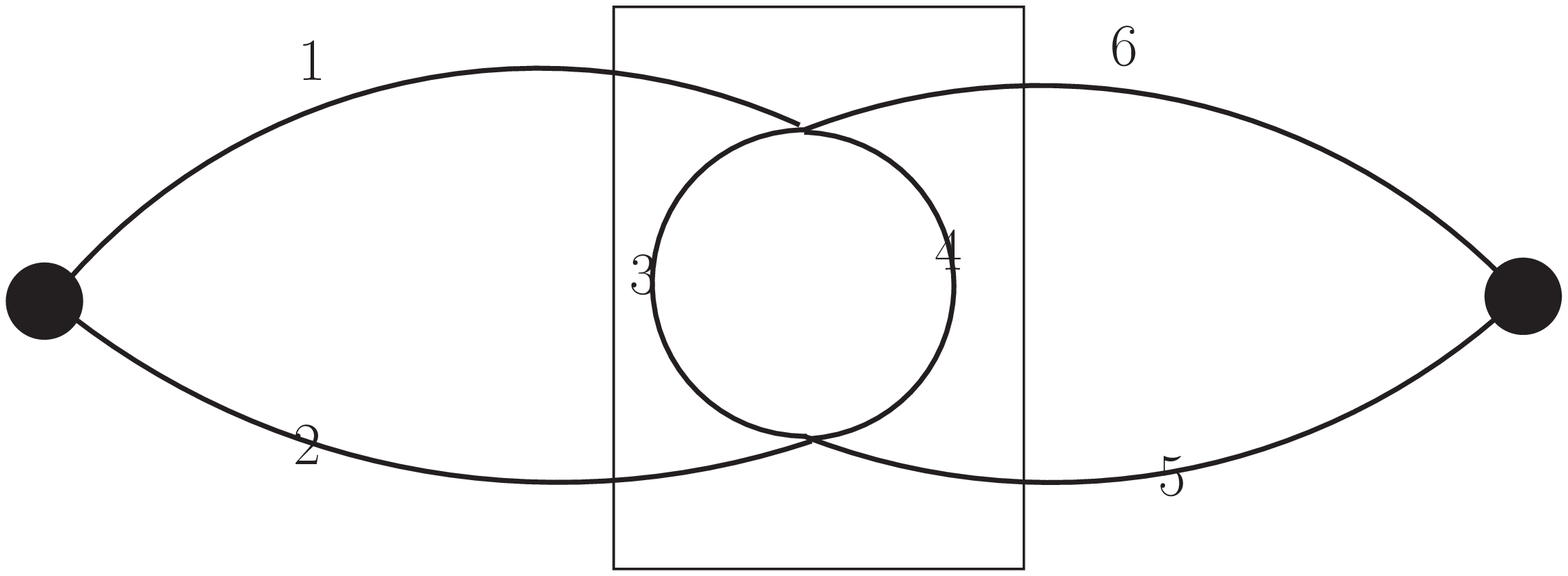}}\;}
\def\soldcbl{\;\raisebox{-5mm}{\epsfysize=10mm\epsfbox{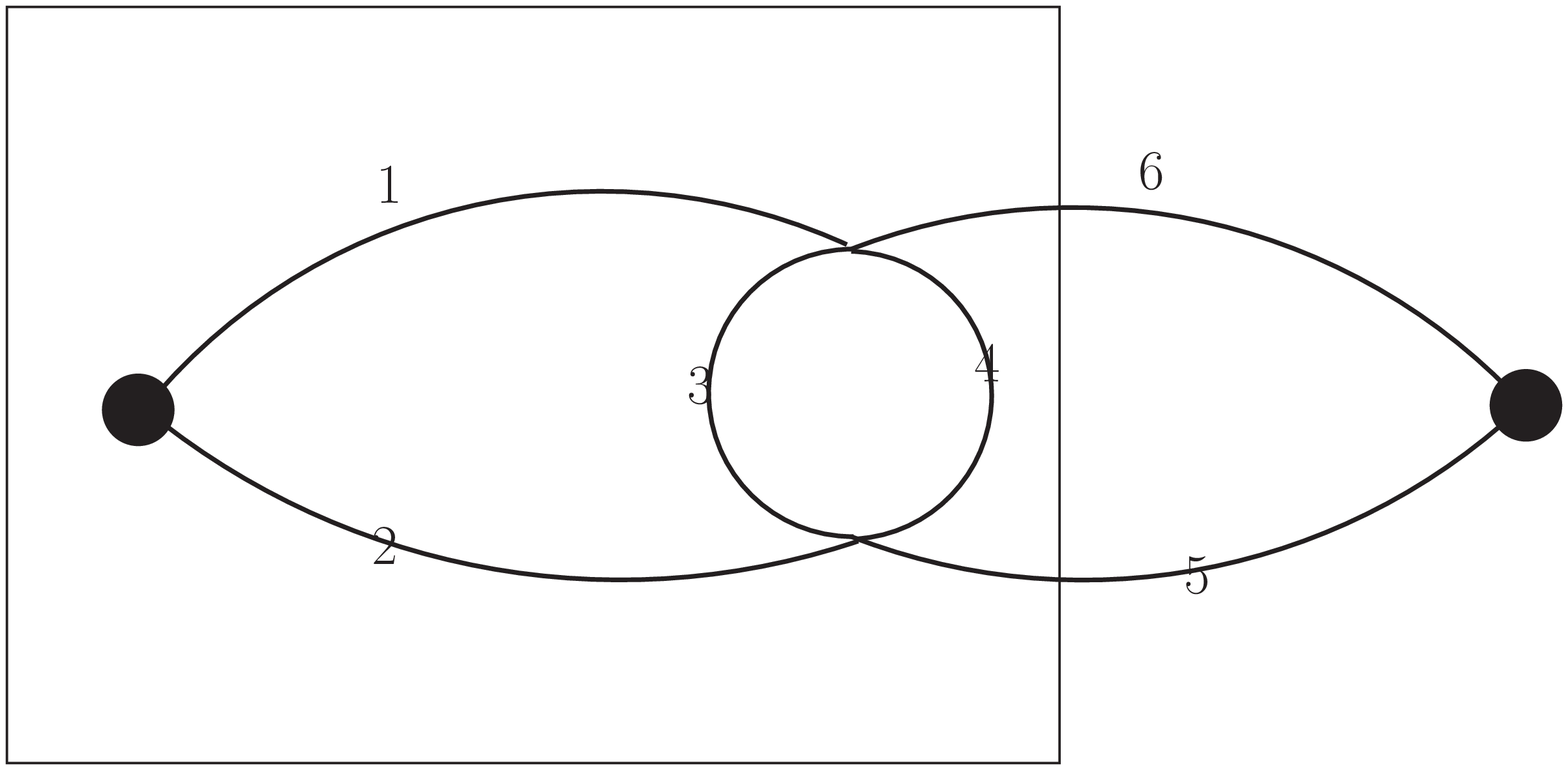}}\;}
\def\soldcbr{\;\raisebox{-5mm}{\epsfysize=10mm\epsfbox{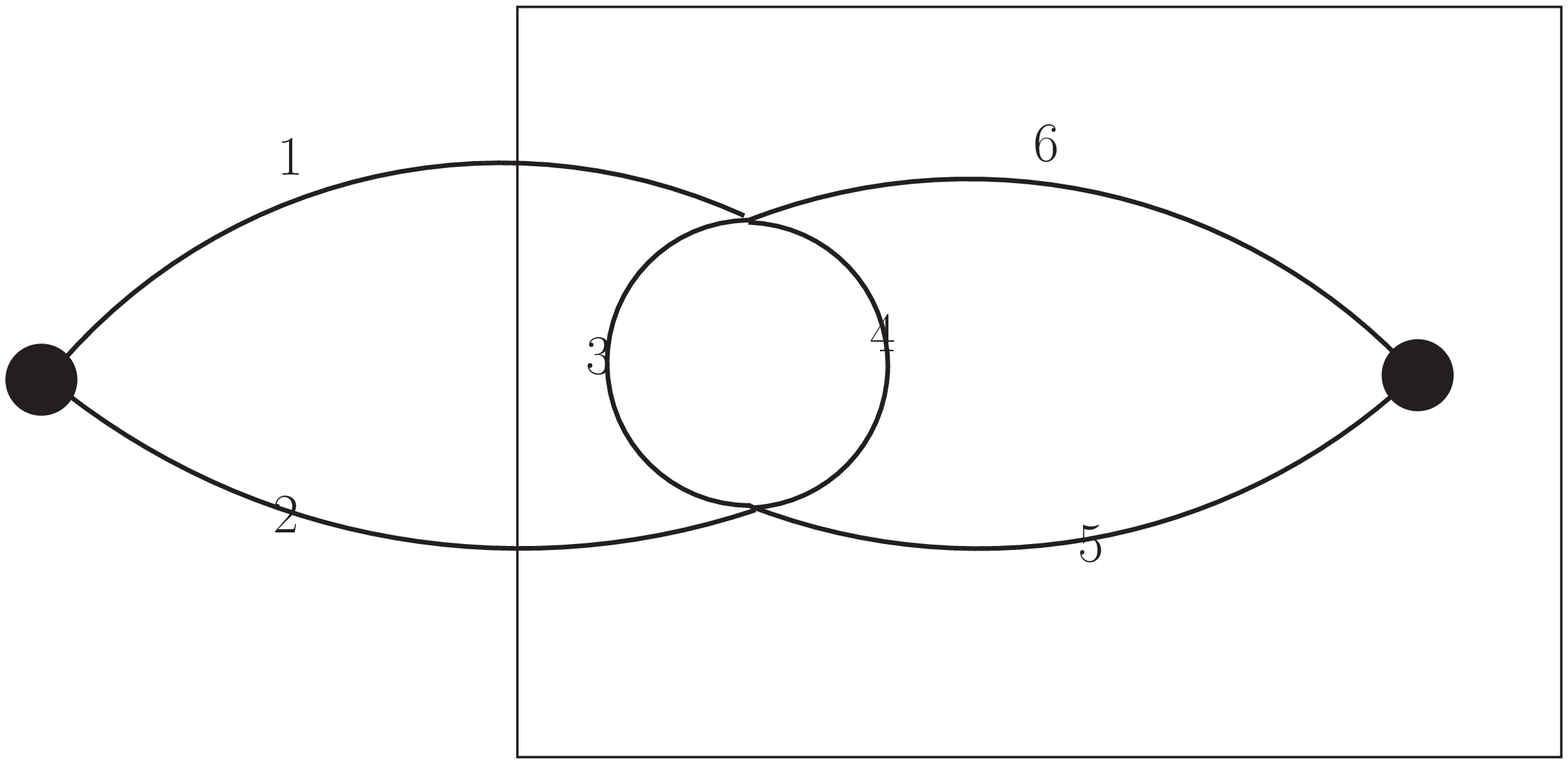}}\;}
\def\soldcbml{\;\raisebox{-5mm}{\epsfysize=10mm\epsfbox{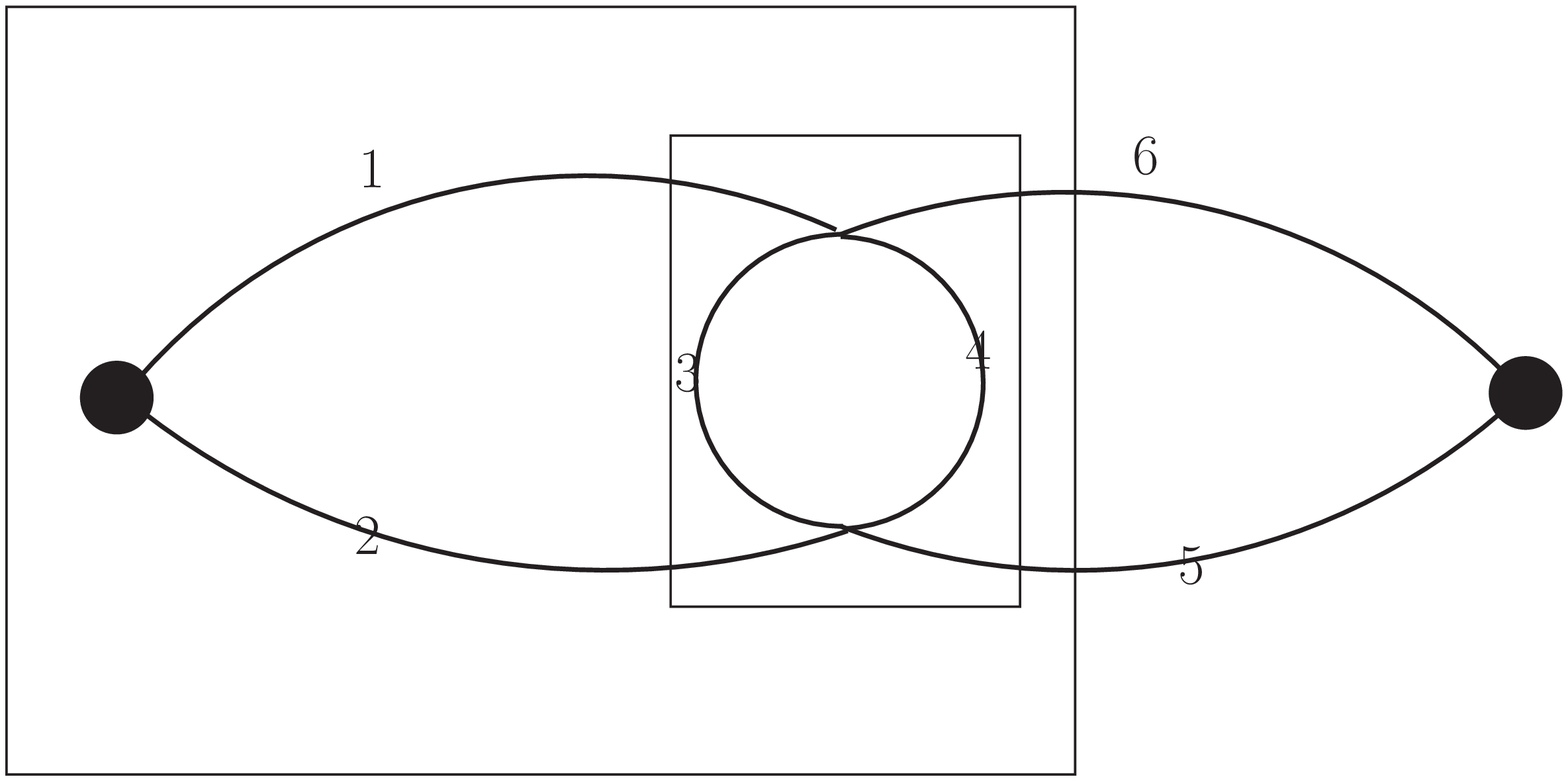}}\;}
\def\soldcbmr{\;\raisebox{-5mm}{\epsfysize=10mm\epsfbox{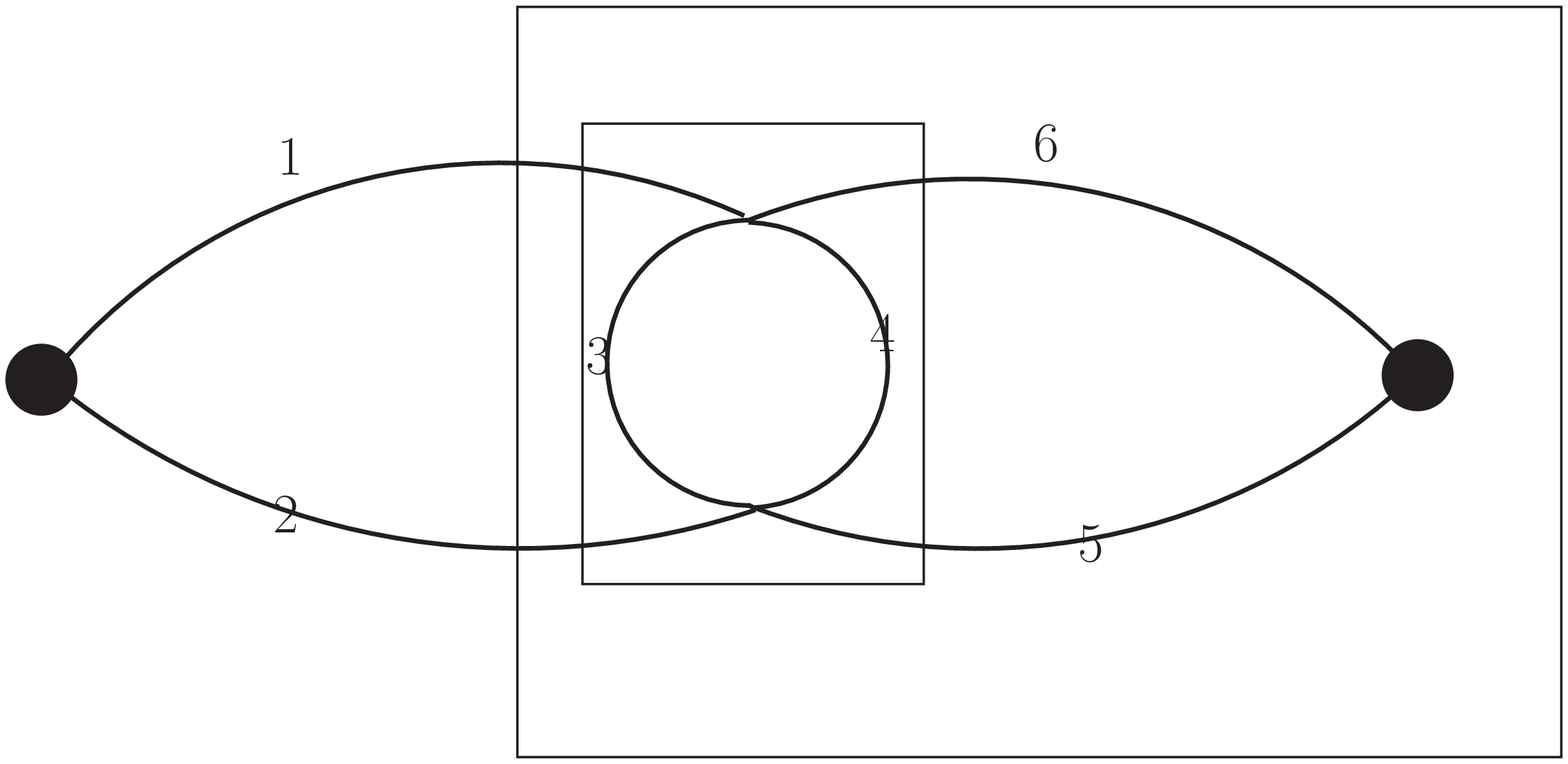}}\;}

\def\wttf{\;\raisebox{-10mm}{\epsfysize=20mm\epsfbox{w334.eps}}\;}
\def\wttft{\;\raisebox{-10mm}{\epsfysize=20mm\epsfbox{w3342.eps}}\;}
\def\wtttf{\;\raisebox{-10mm}{\epsfysize=20mm\epsfbox{w3324.eps}}\;}
\def\wttt{\;\raisebox{-10mm}{\epsfysize=20mm\epsfbox{w332.eps}}\;}

\def\wfconn{\;\raisebox{-10mm}{\epsfysize=20mm\epsfbox{wfconn.eps}}\;}
\def\wfconnos{\;\raisebox{-10mm}{\epsfysize=20mm\epsfbox{wfconnos.eps}}\;}

\def\wtwf{\;\raisebox{-15mm}{\epsfysize=36mm\epsfbox{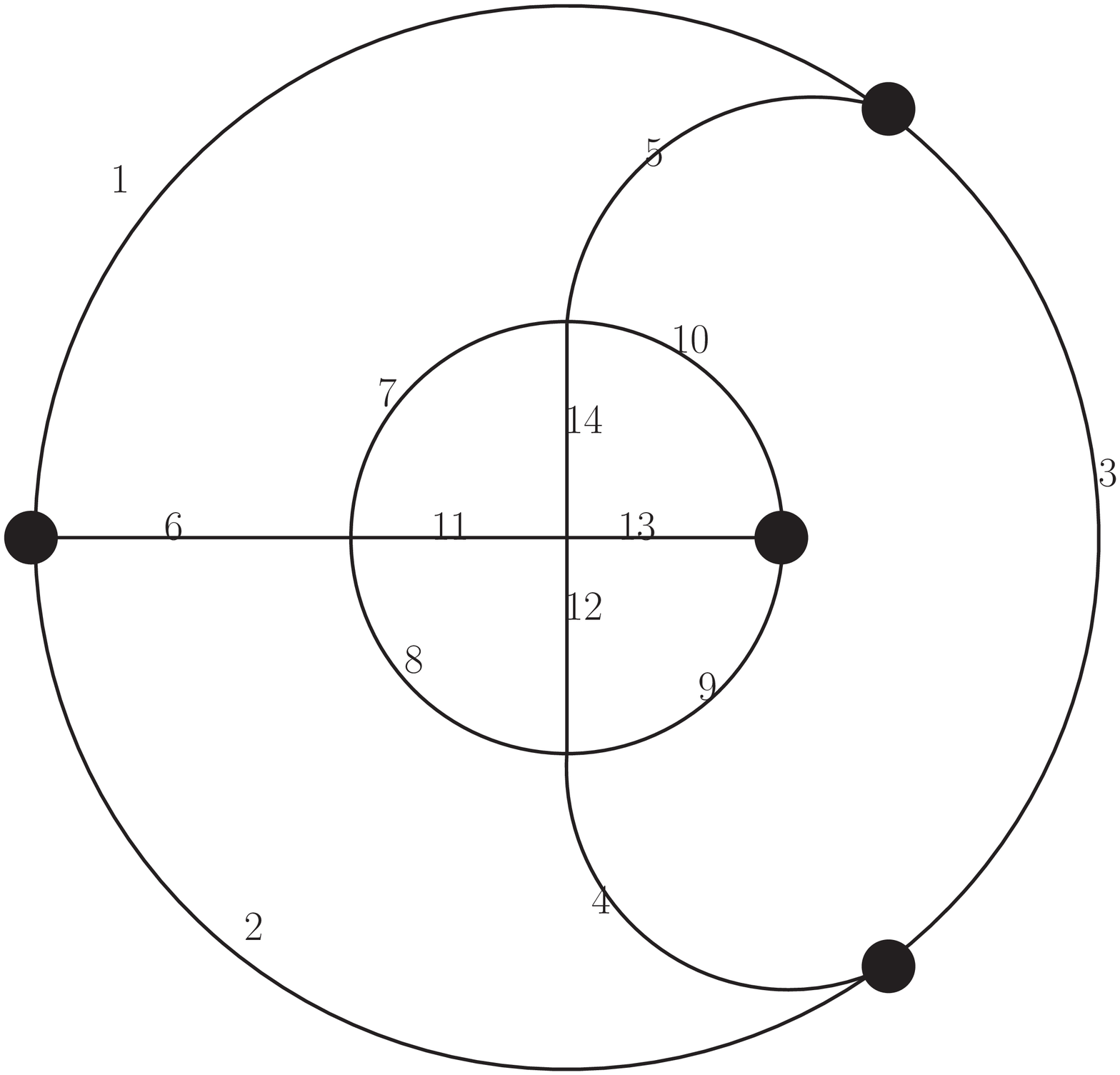}}\;}
\def\wtwft{\;\raisebox{-15mm}{\epsfysize=36mm\epsfbox{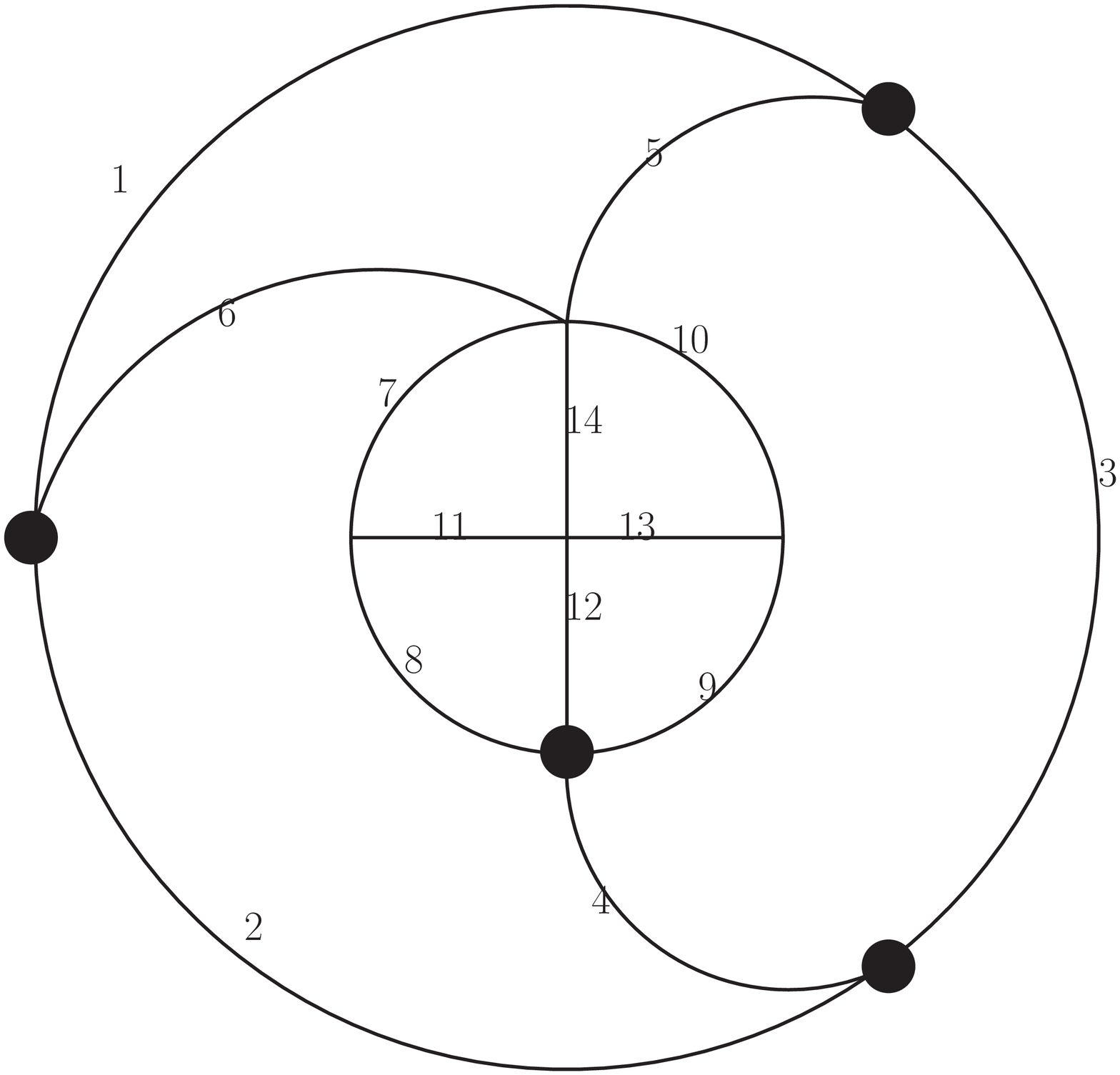}}\;}
\def\wtwfb{\;\raisebox{-15mm}{\epsfysize=36mm\epsfbox{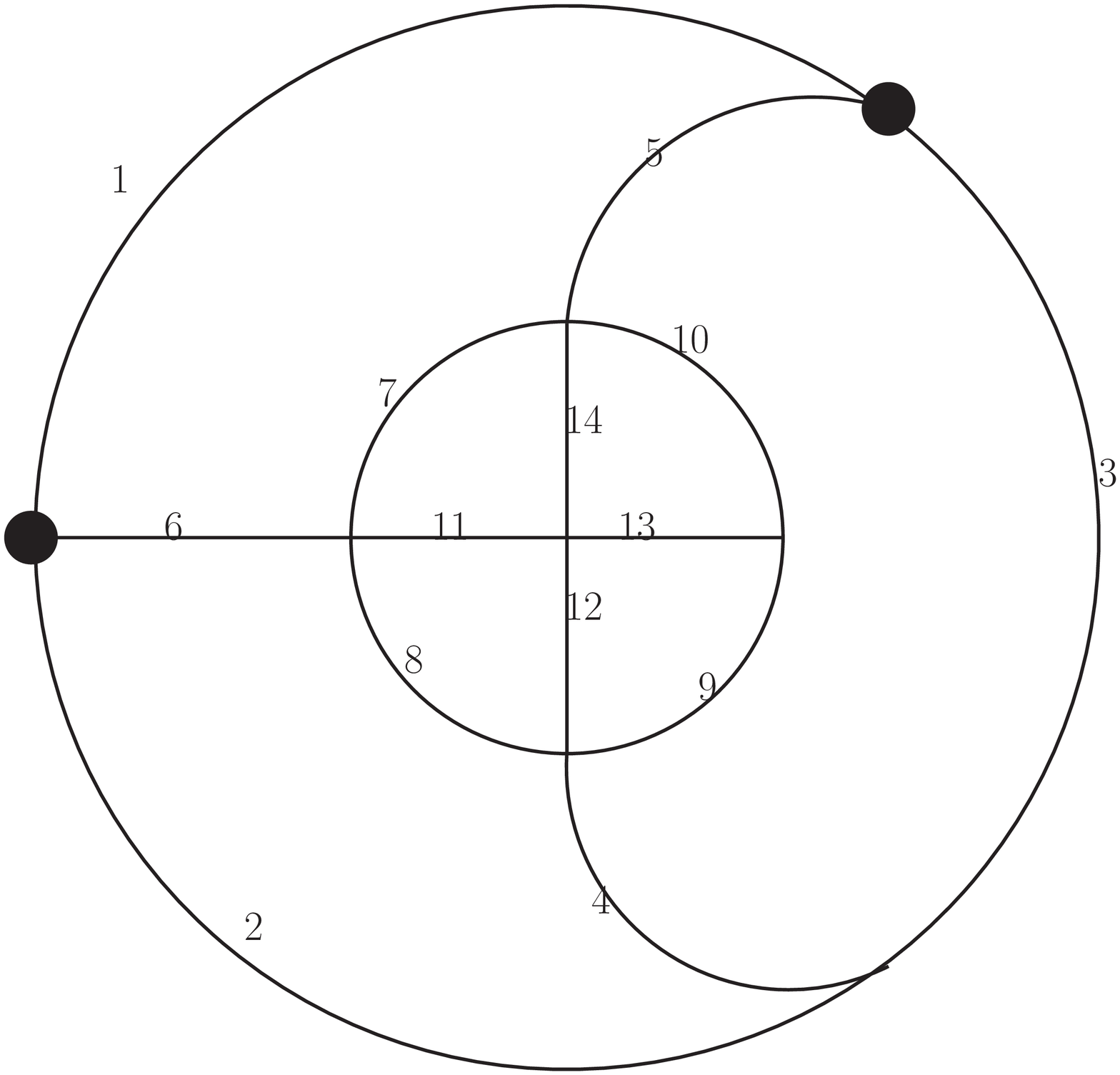}}\;}
\def\wtwftb{\;\raisebox{-15mm}{\epsfysize=36mm\epsfbox{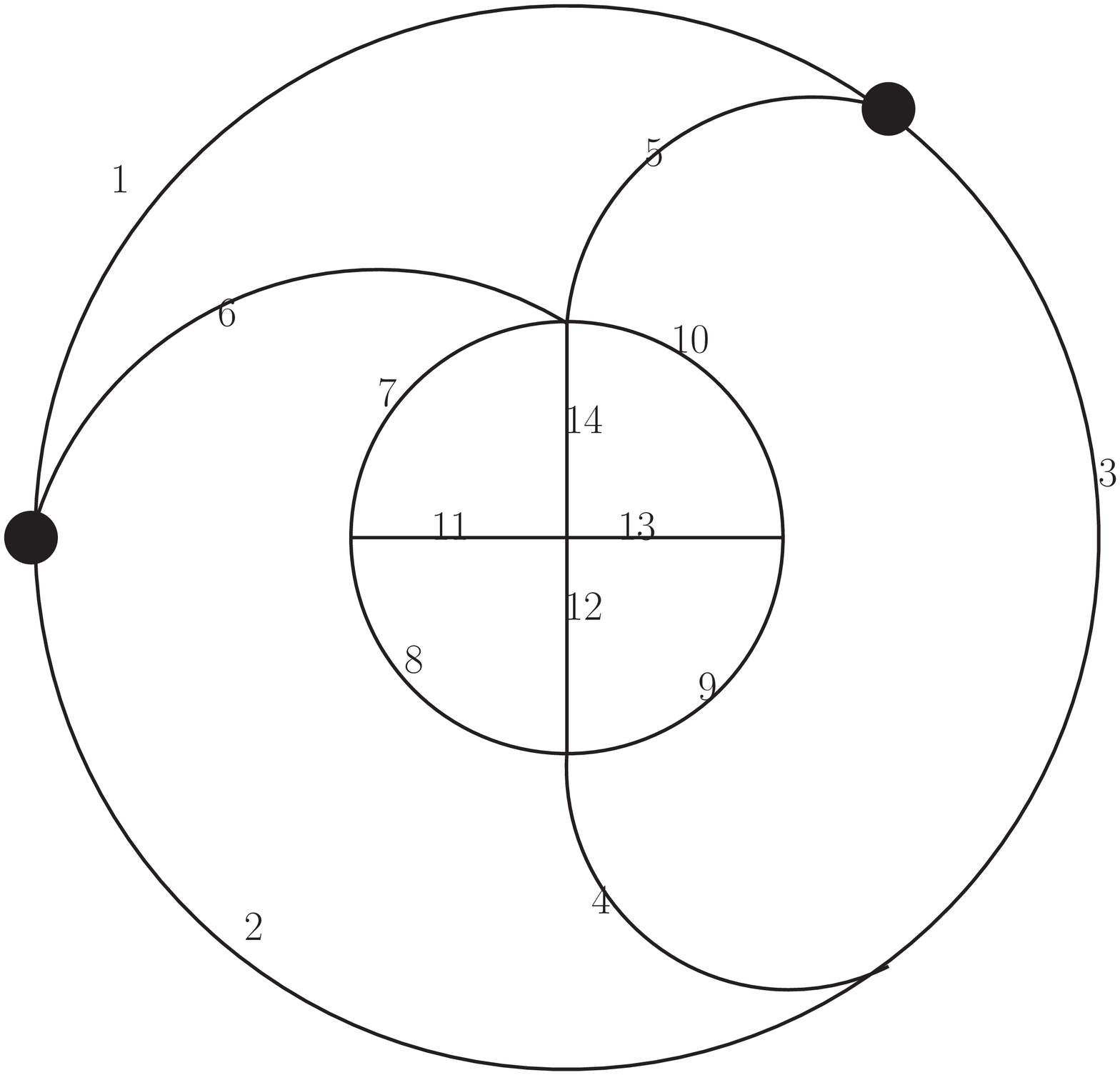}}\;}
\def\wt{\;\raisebox{-10mm}{\epsfysize=24mm\epsfbox{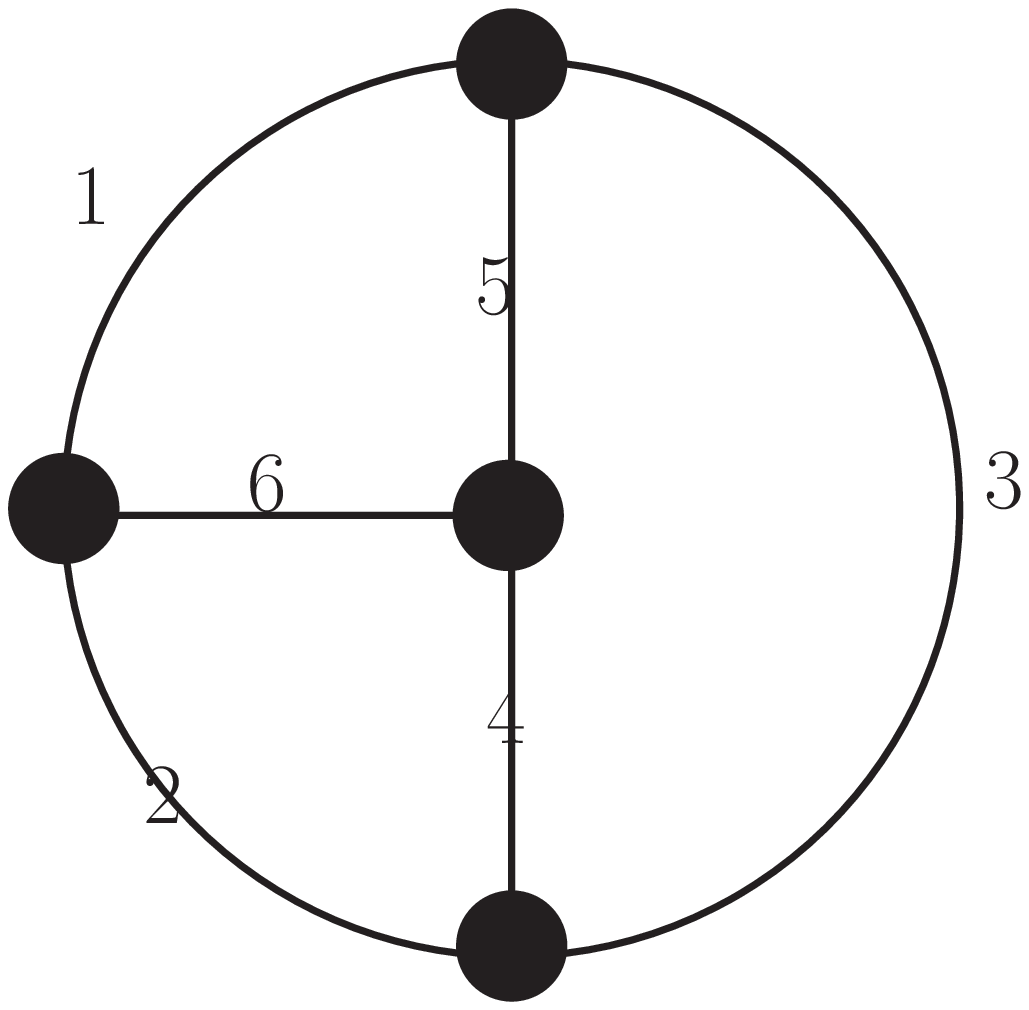}}\;}
\def\wtb{\;\raisebox{-10mm}{\epsfysize=24mm\epsfbox{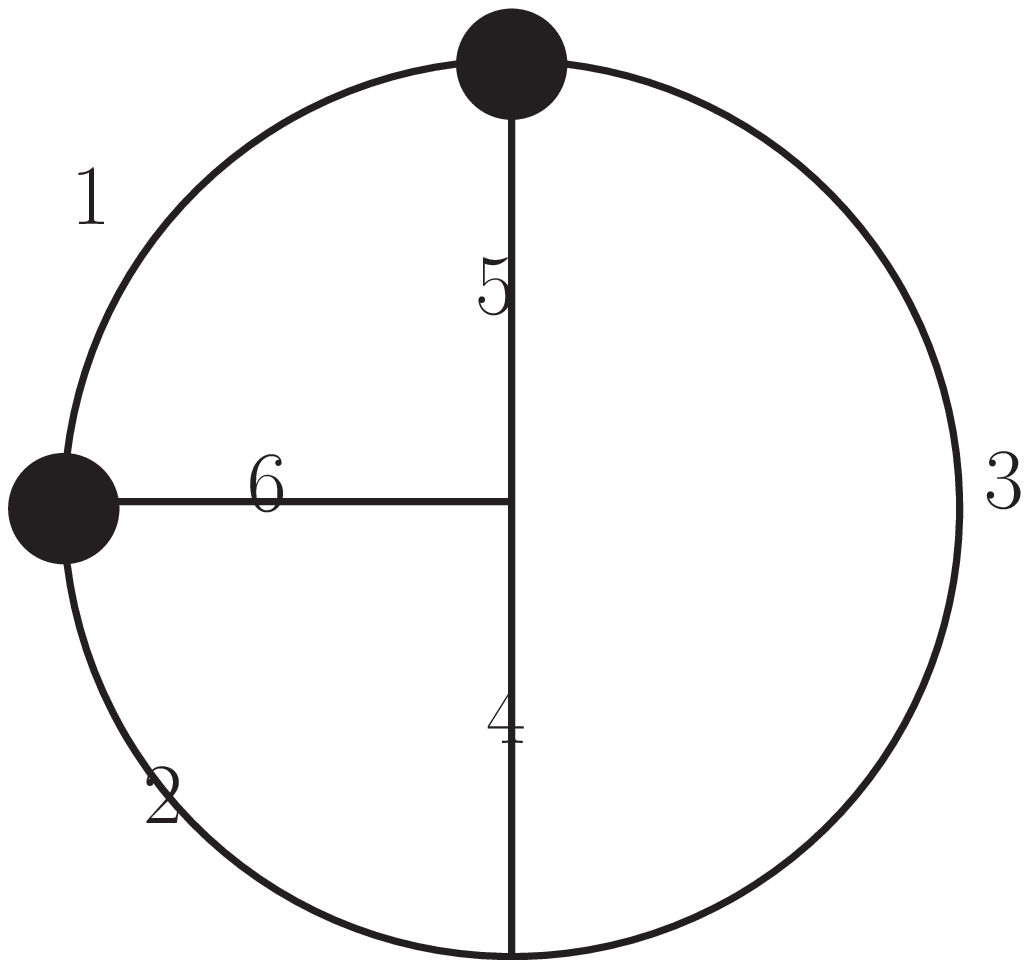}}\;}
\def\wtbbull{\;\raisebox{-10mm}{\epsfysize=24mm\epsfbox{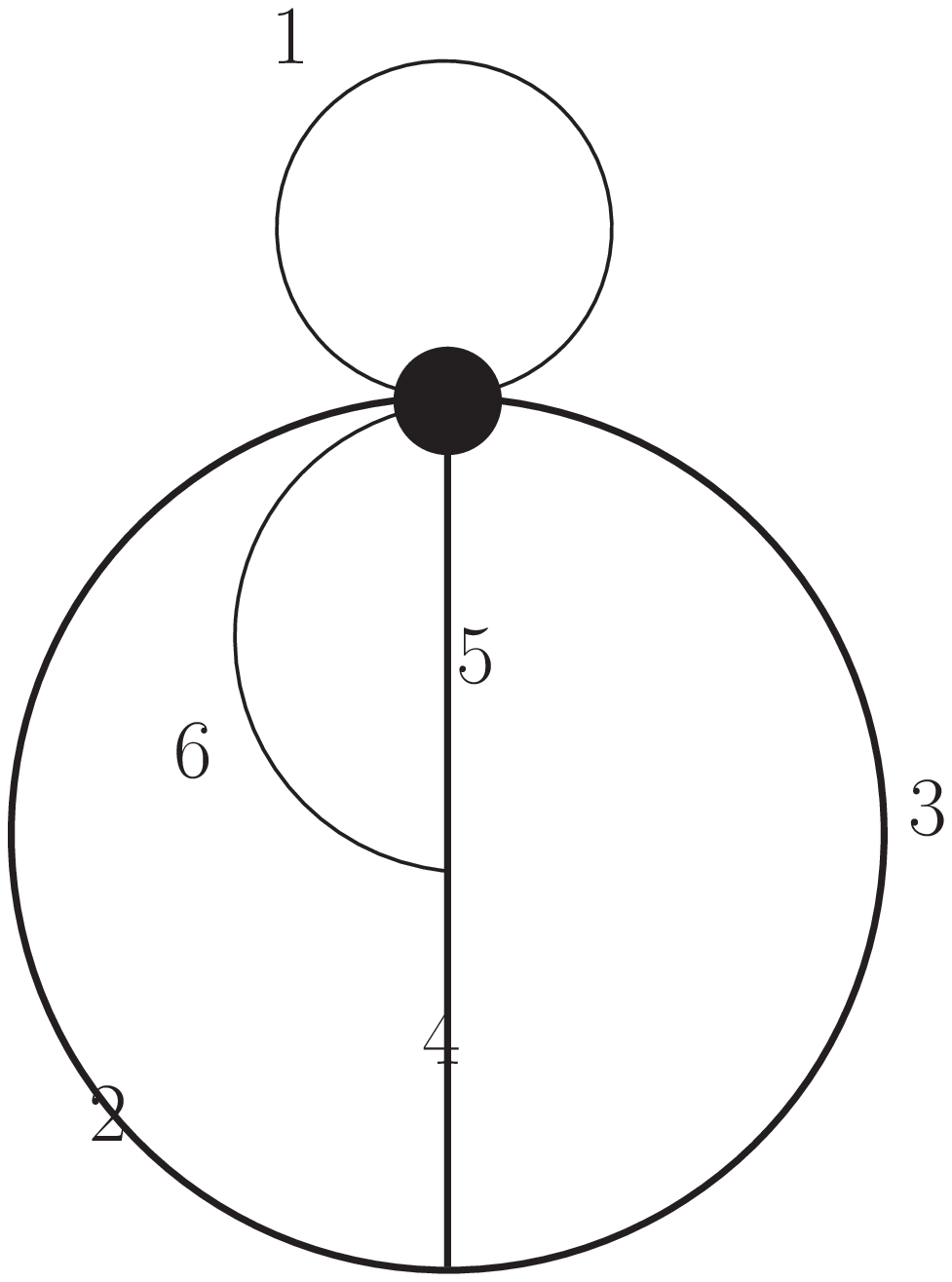}}\;}
\def\wtbups{\;\raisebox{-10mm}{\epsfysize=24mm\epsfbox{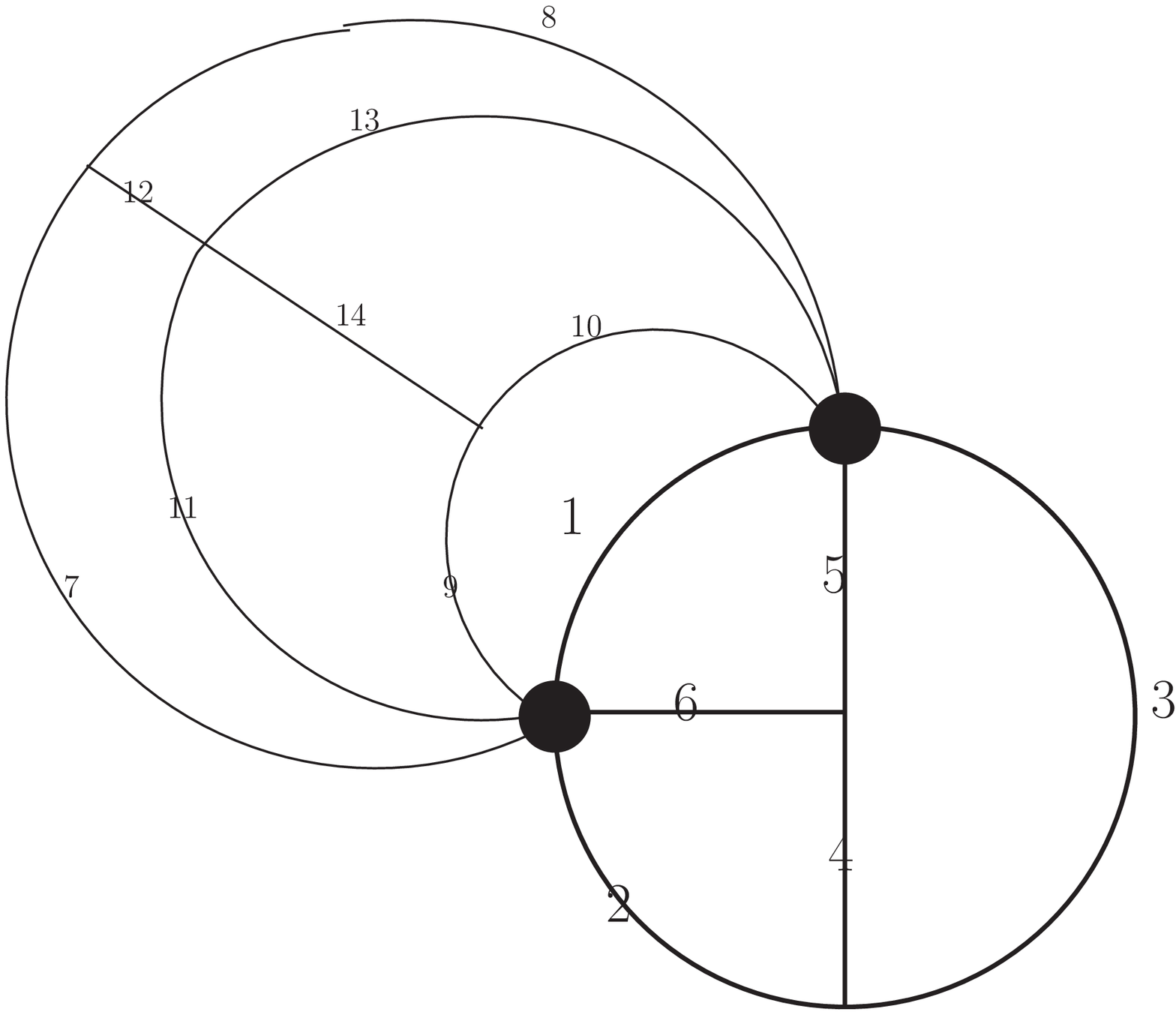}}\;}
\def\wf{\;\raisebox{-10mm}{\epsfysize=24mm\epsfbox{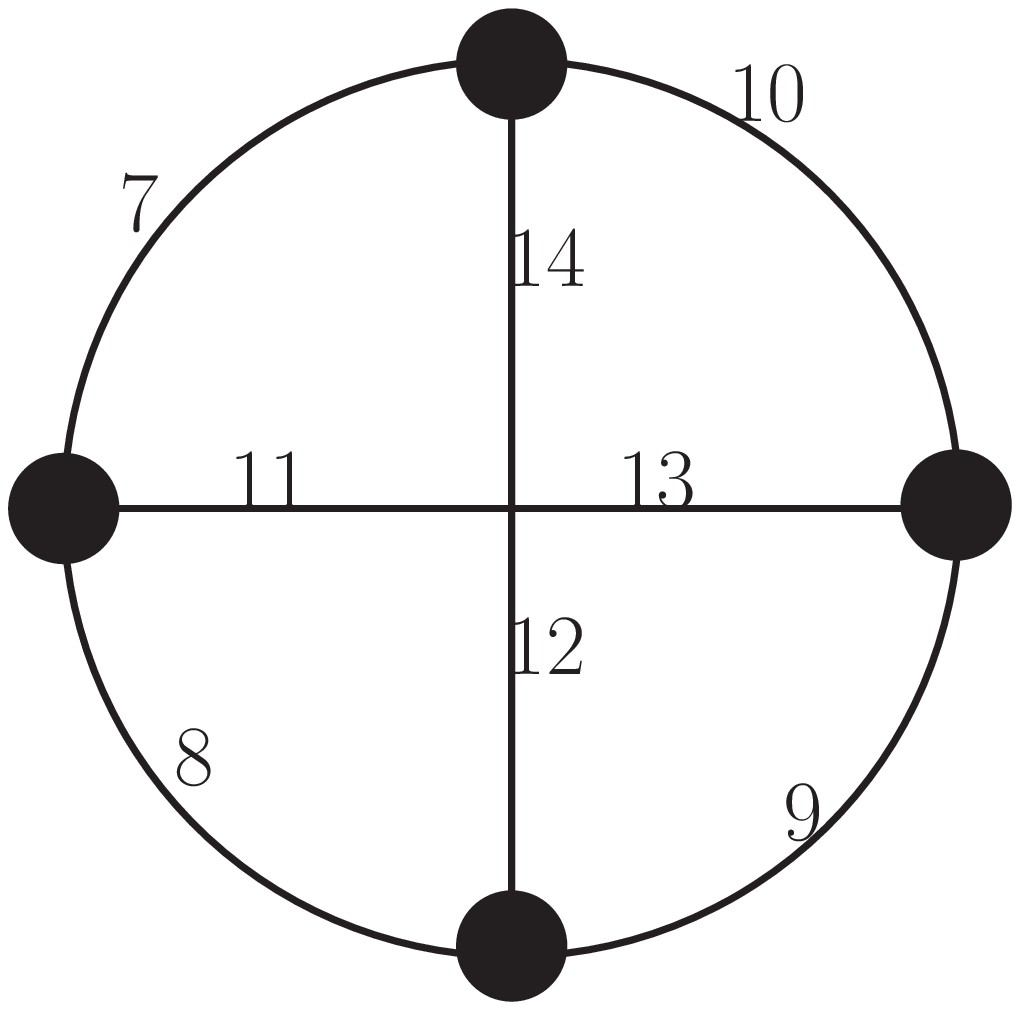}}\;}
\def\wfb{\;\raisebox{-10mm}{\epsfysize=24mm\epsfbox{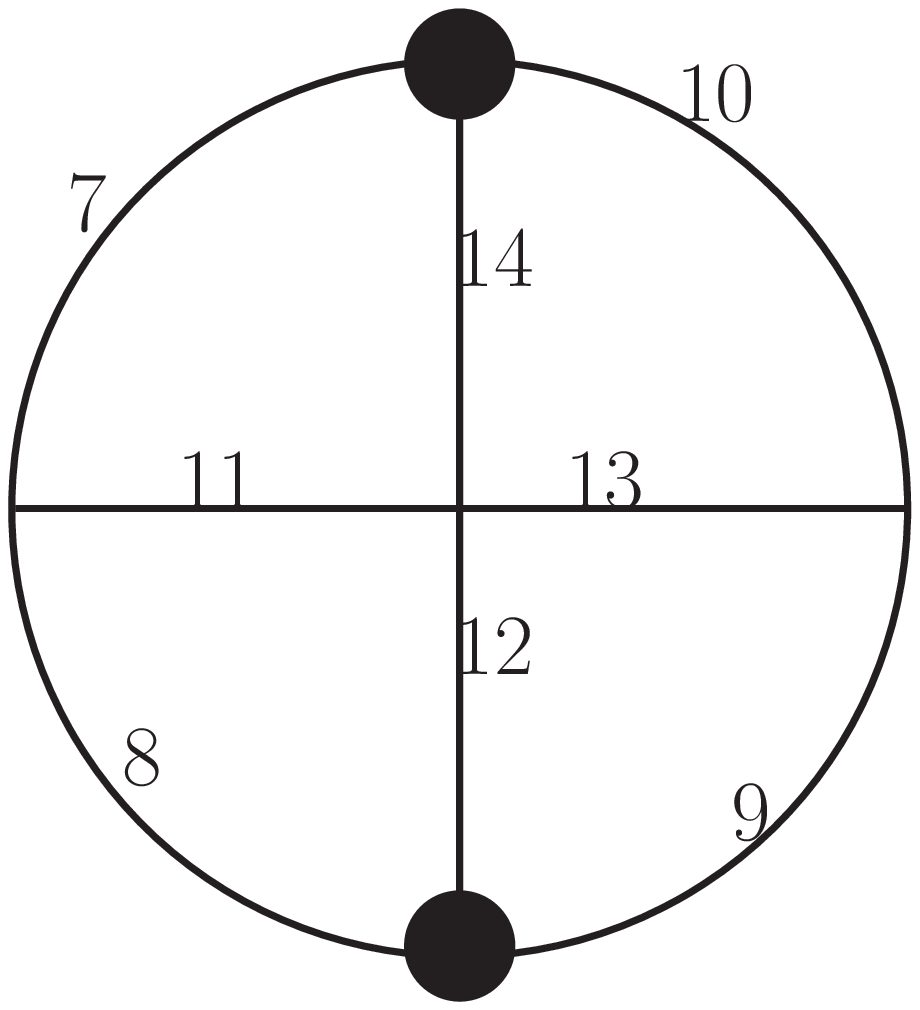}}\;}

\def\dcqs{\;\raisebox{-10mm}{\epsfysize=24mm\epsfbox{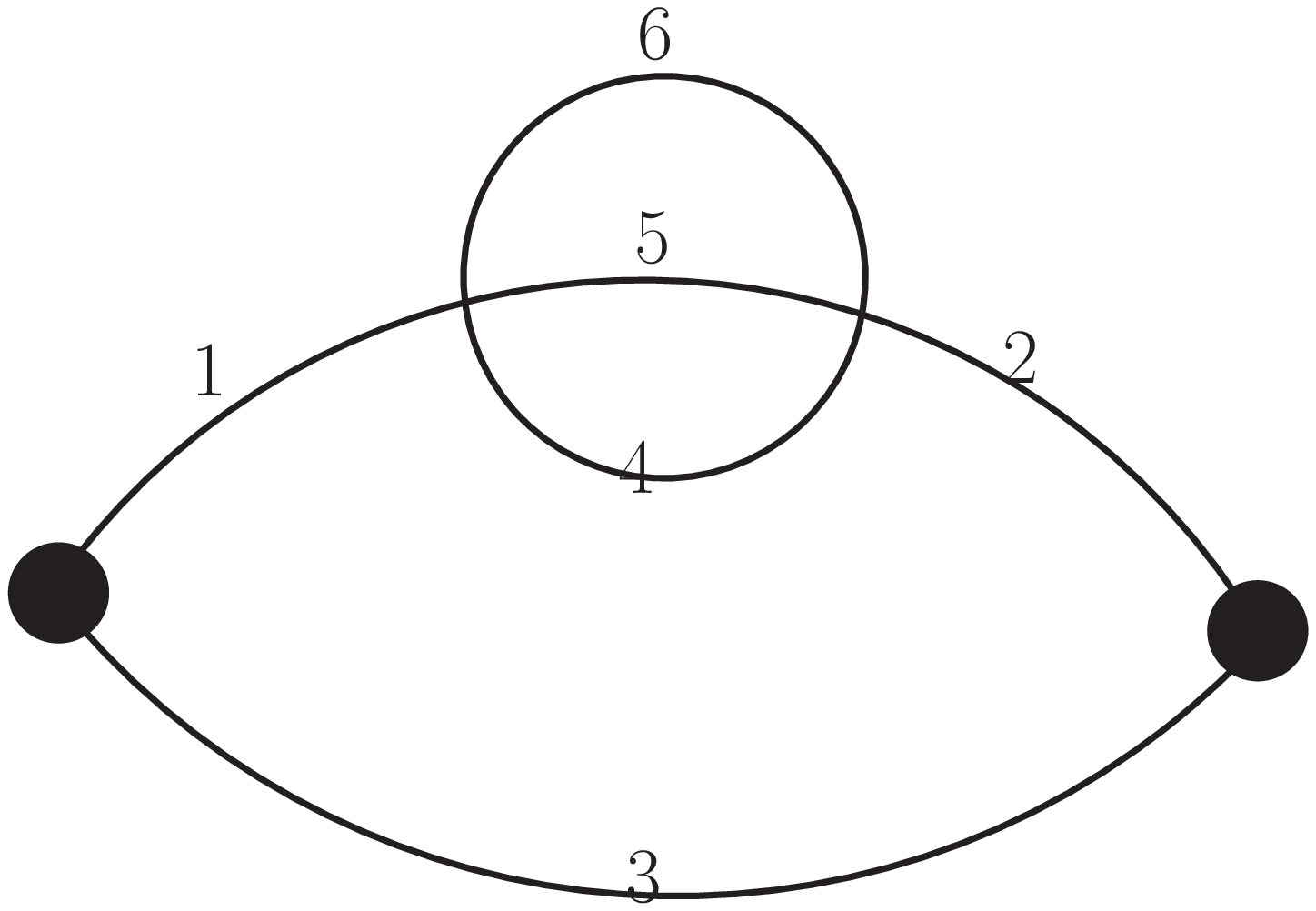}}\;}
\def\dcqssq{\;\raisebox{-10mm}{\epsfysize=24mm\epsfbox{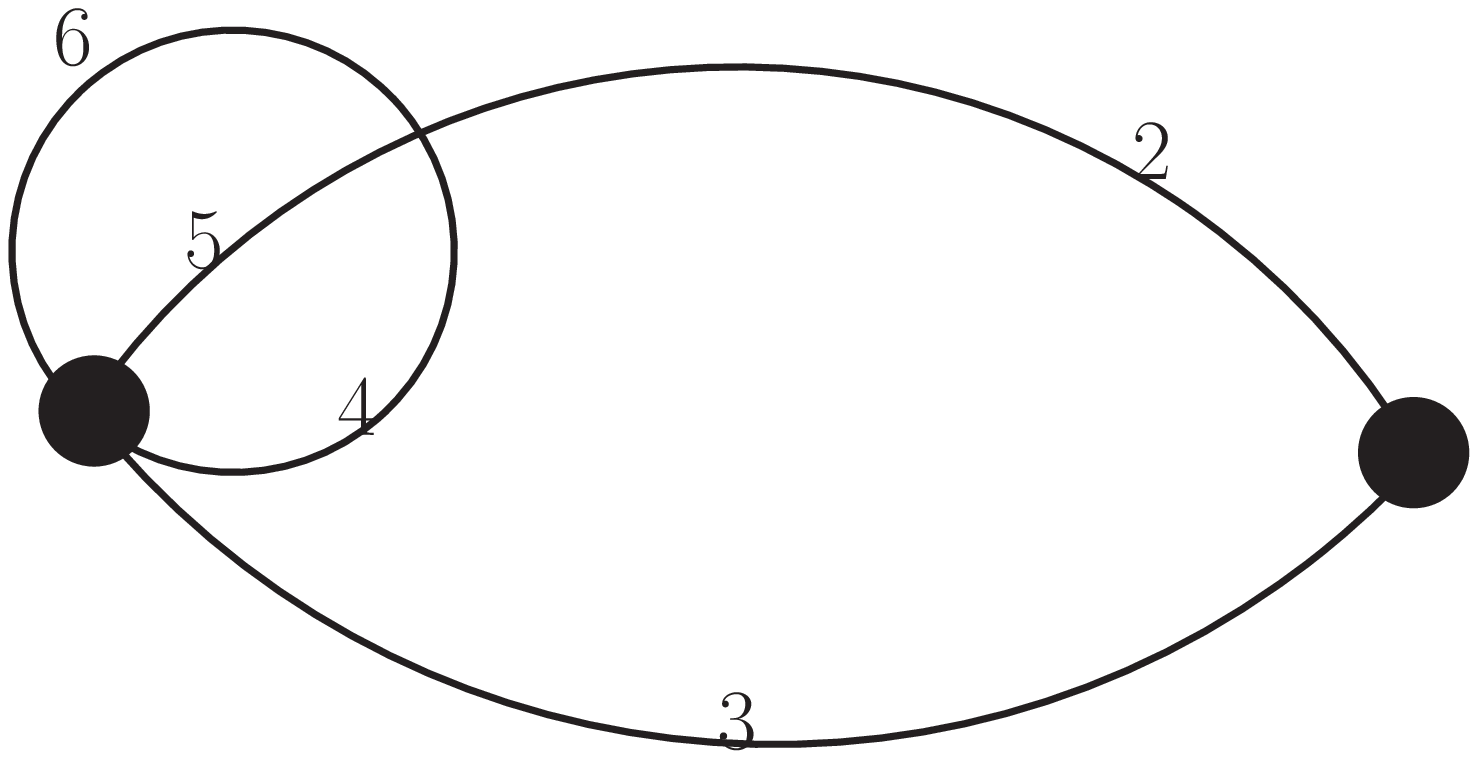}}\;}
\def\subqul{\;\raisebox{-10mm}{\epsfysize=24mm\epsfbox{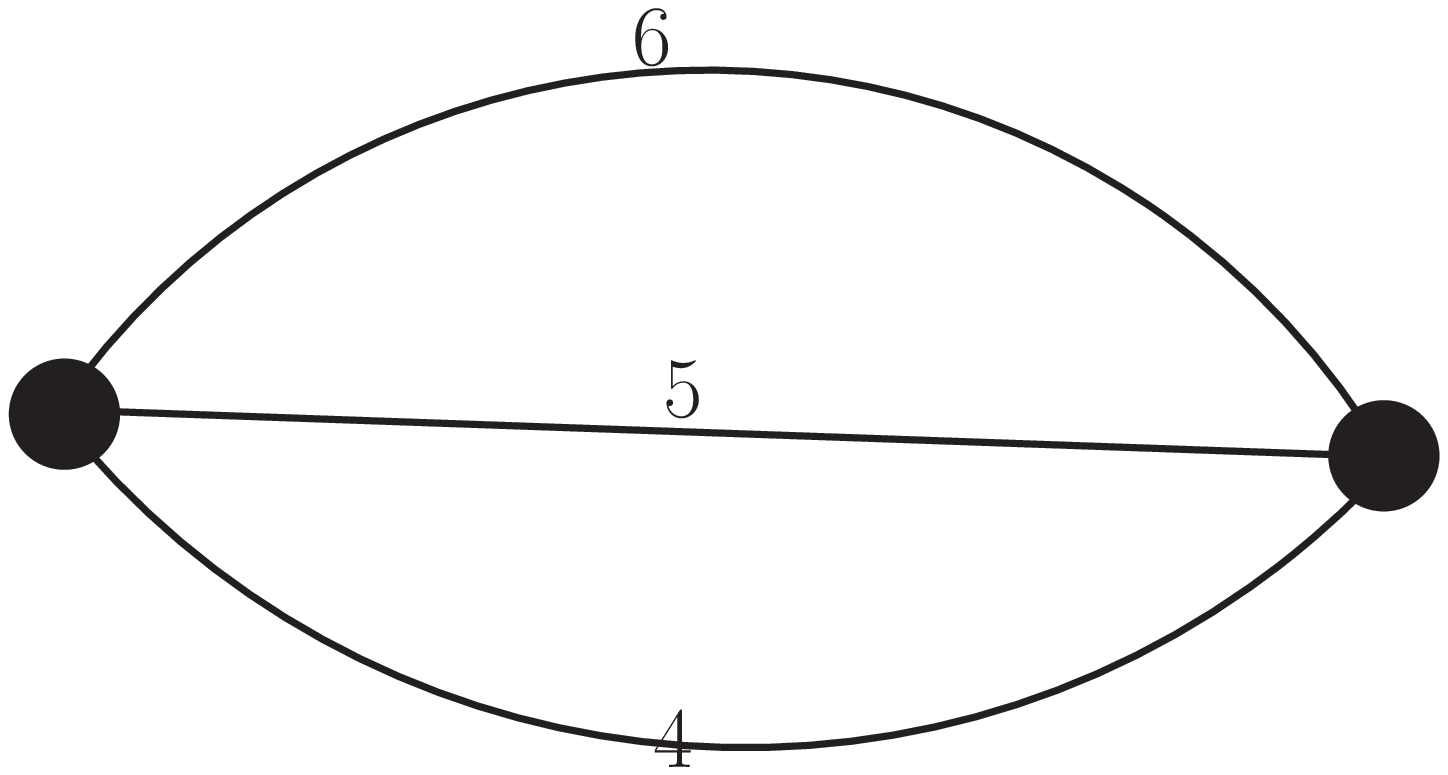}}\;}
\def\codc{\;\raisebox{-10mm}{\epsfysize=24mm\epsfbox{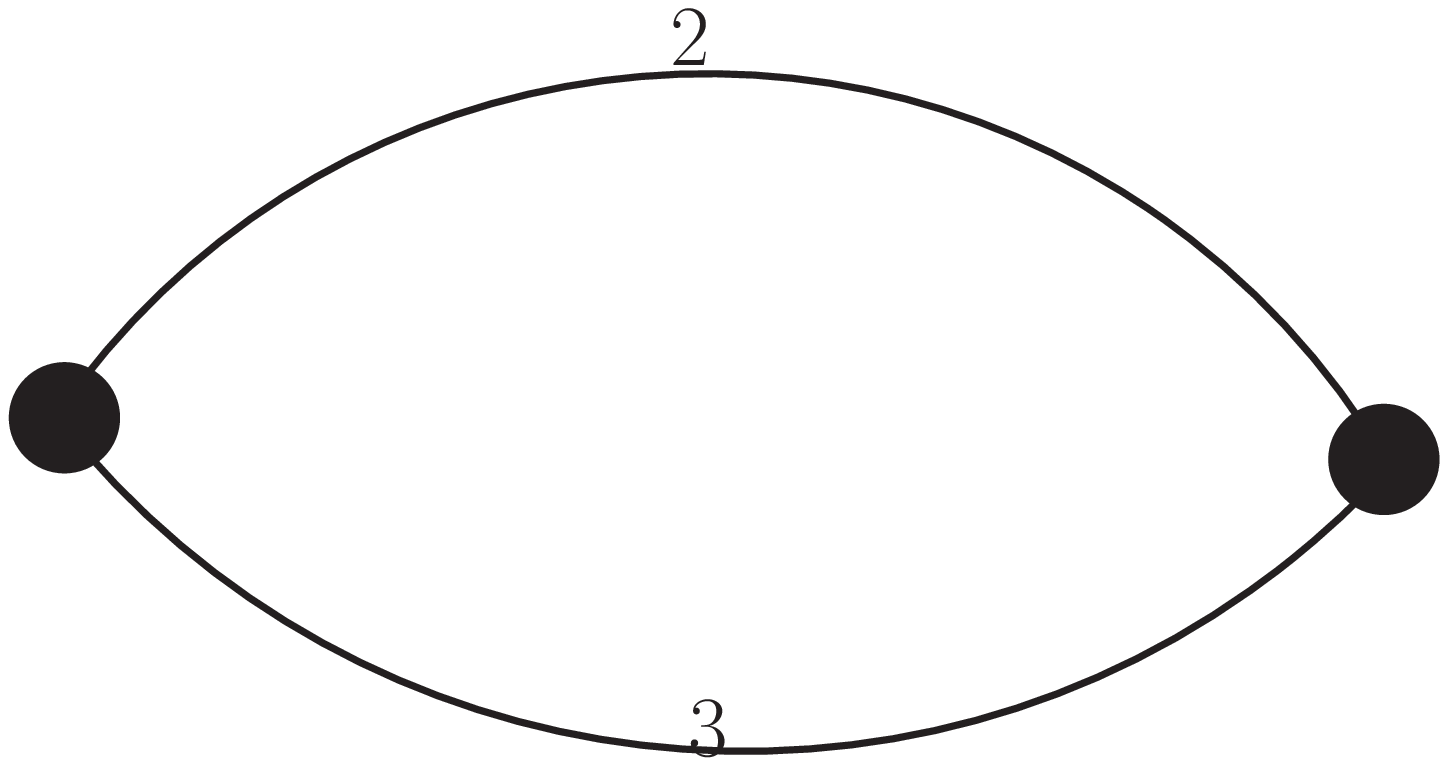}}\;}
\def\codcsq{\;\raisebox{-10mm}{\epsfysize=24mm\epsfbox{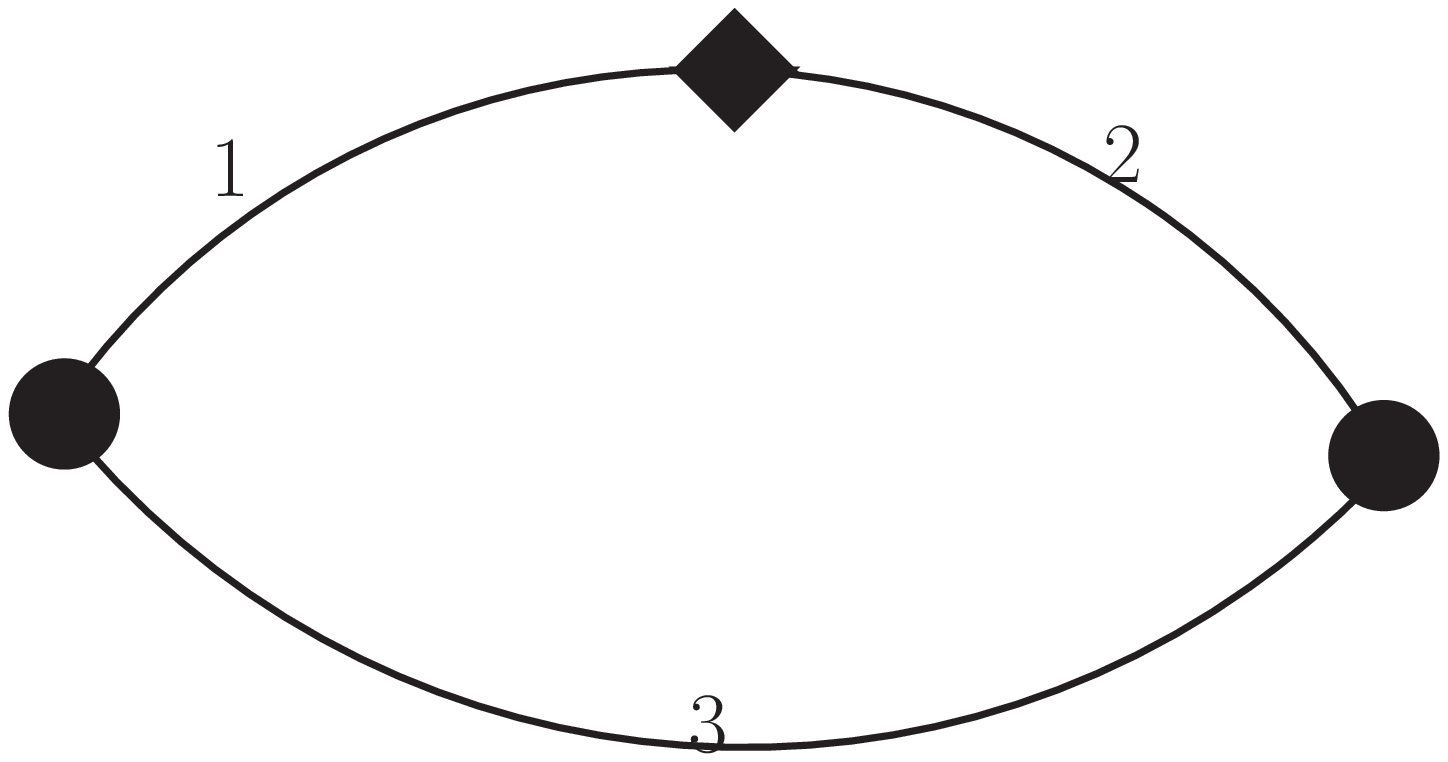}}\;}

\def\wttfsmall{\;\raisebox{-1mm}{\epsfysize=3mm\epsfbox{w334.eps}}\;}
\def\wttftsmall{\;\raisebox{-1mm}{\epsfysize=3mm\epsfbox{w3342.eps}}\;}
\def\wtttfsmall{\;\raisebox{-1mm}{\epsfysize=3mm\epsfbox{w3324.eps}}\;}
\def\wtttssmall{\;\raisebox{-1mm}{\epsfysize=3mm\epsfbox{w332.eps}}\;}

\def\tltp{\;\raisebox{-10mm}{\epsfysize=24mm\epsfbox{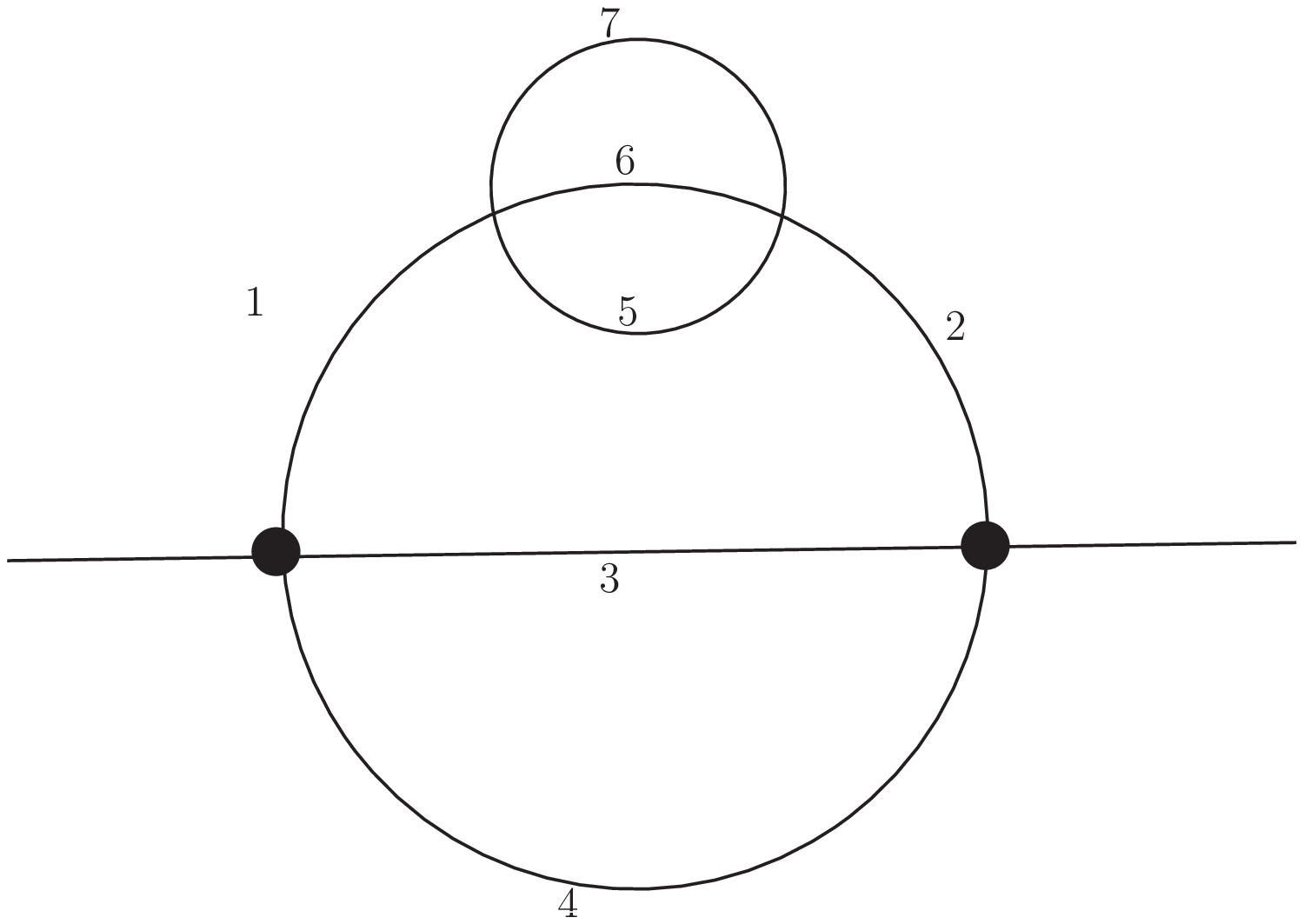}}\;}
\def\tltpdt{\;\raisebox{-10mm}{\epsfysize=24mm\epsfbox{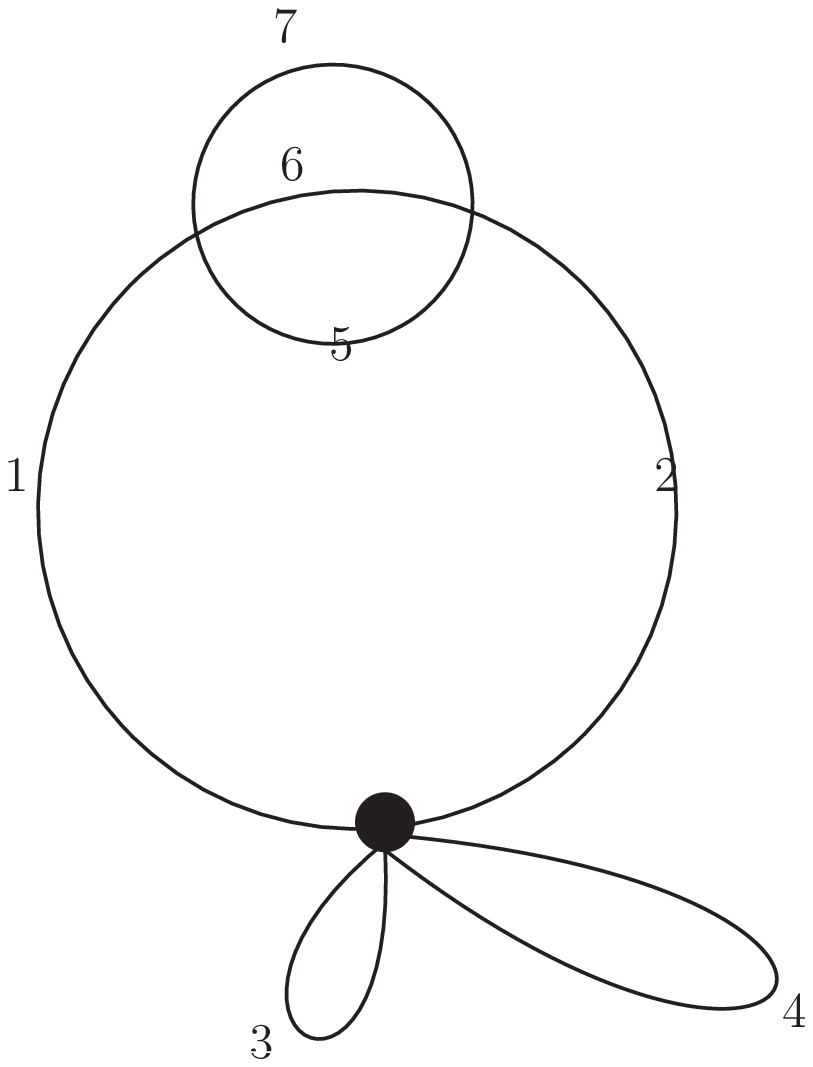}}\;}
\def\tltpsq{\;\raisebox{-10mm}{\epsfysize=24mm\epsfbox{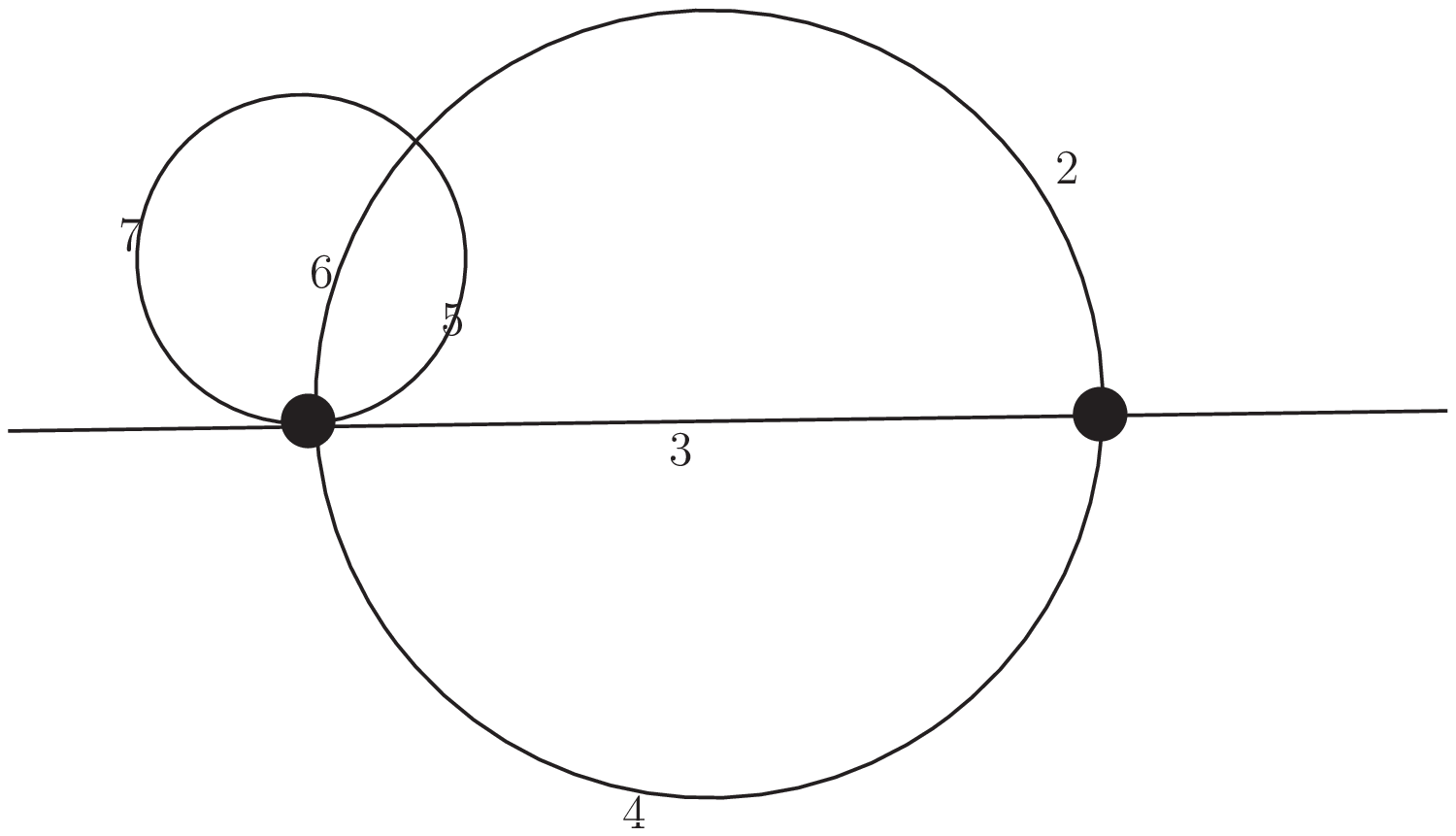}}\;}
\def\tltpsqdt{\;\raisebox{-10mm}{\epsfysize=24mm\epsfbox{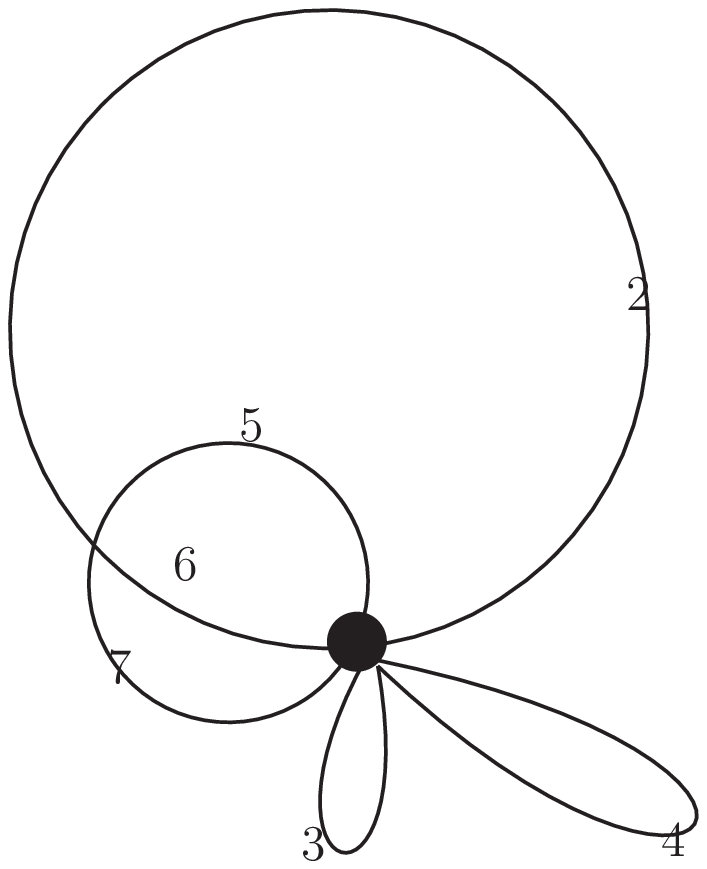}}\;}
\def\oltpdot{\;\raisebox{-10mm}{\epsfysize=24mm\epsfbox{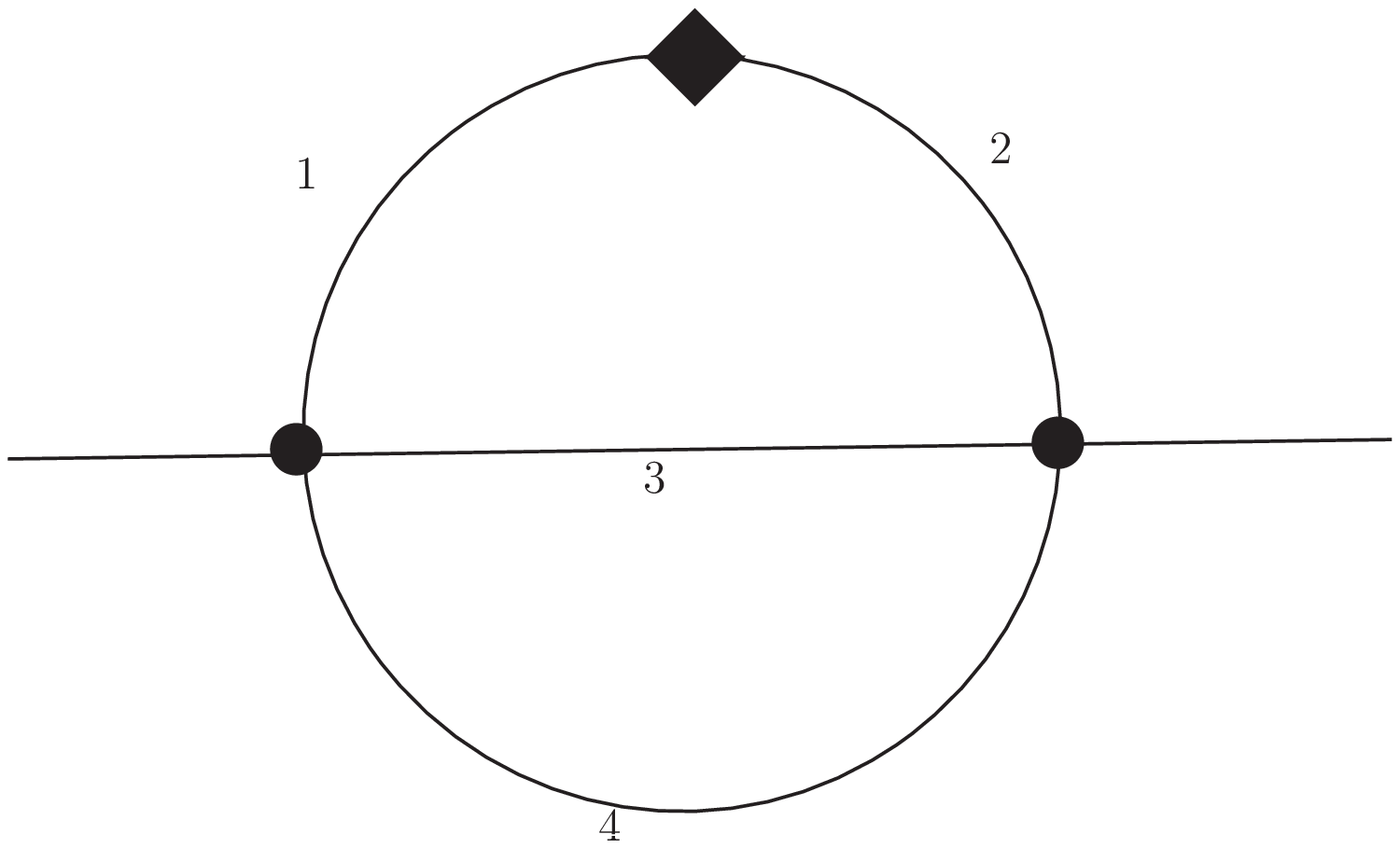}}\;}
\def\oltpsq{\;\raisebox{-10mm}{\epsfysize=24mm\epsfbox{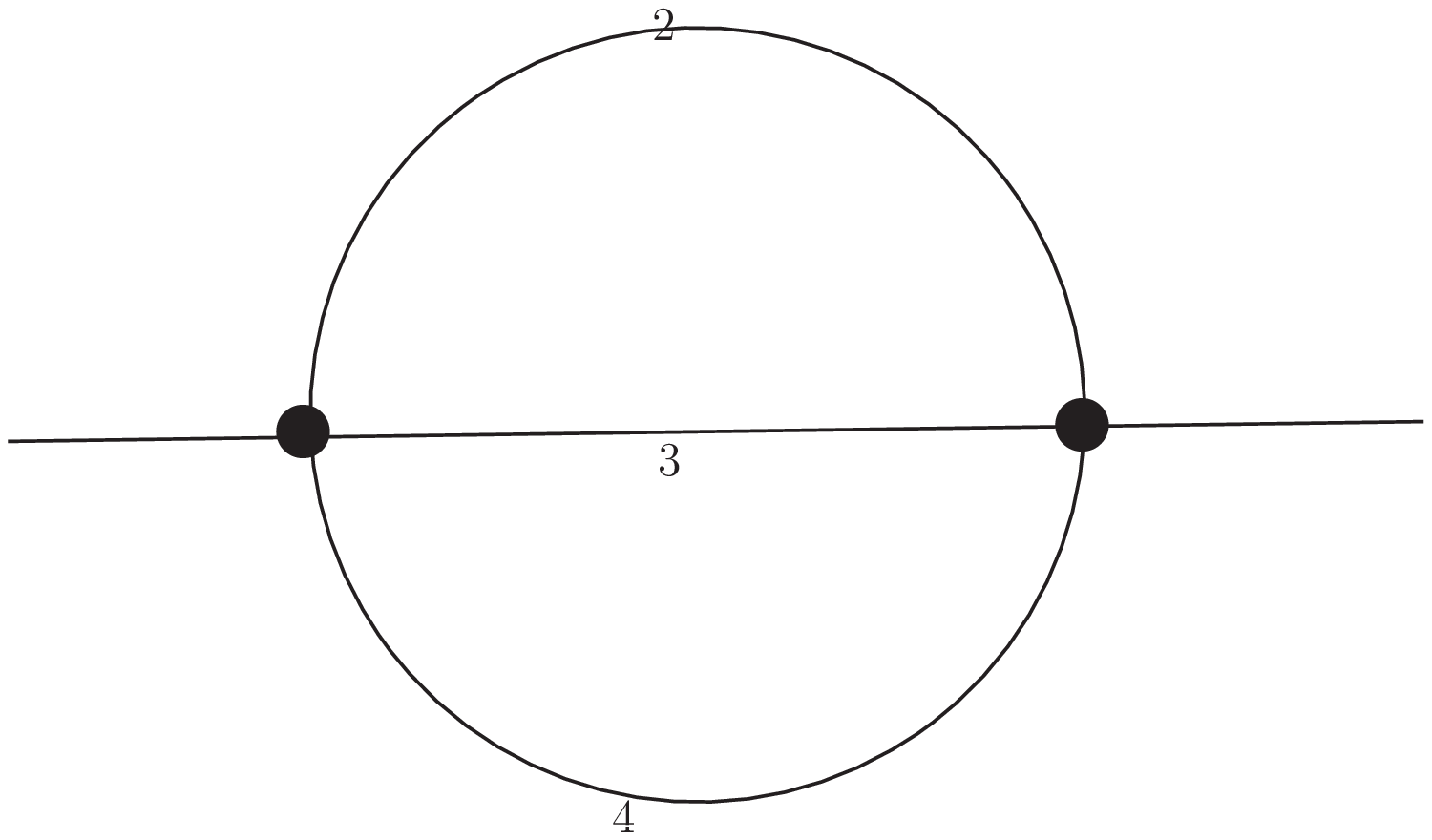}}\;}
\def\suboltp{\;\raisebox{-10mm}{\epsfysize=24mm\epsfbox{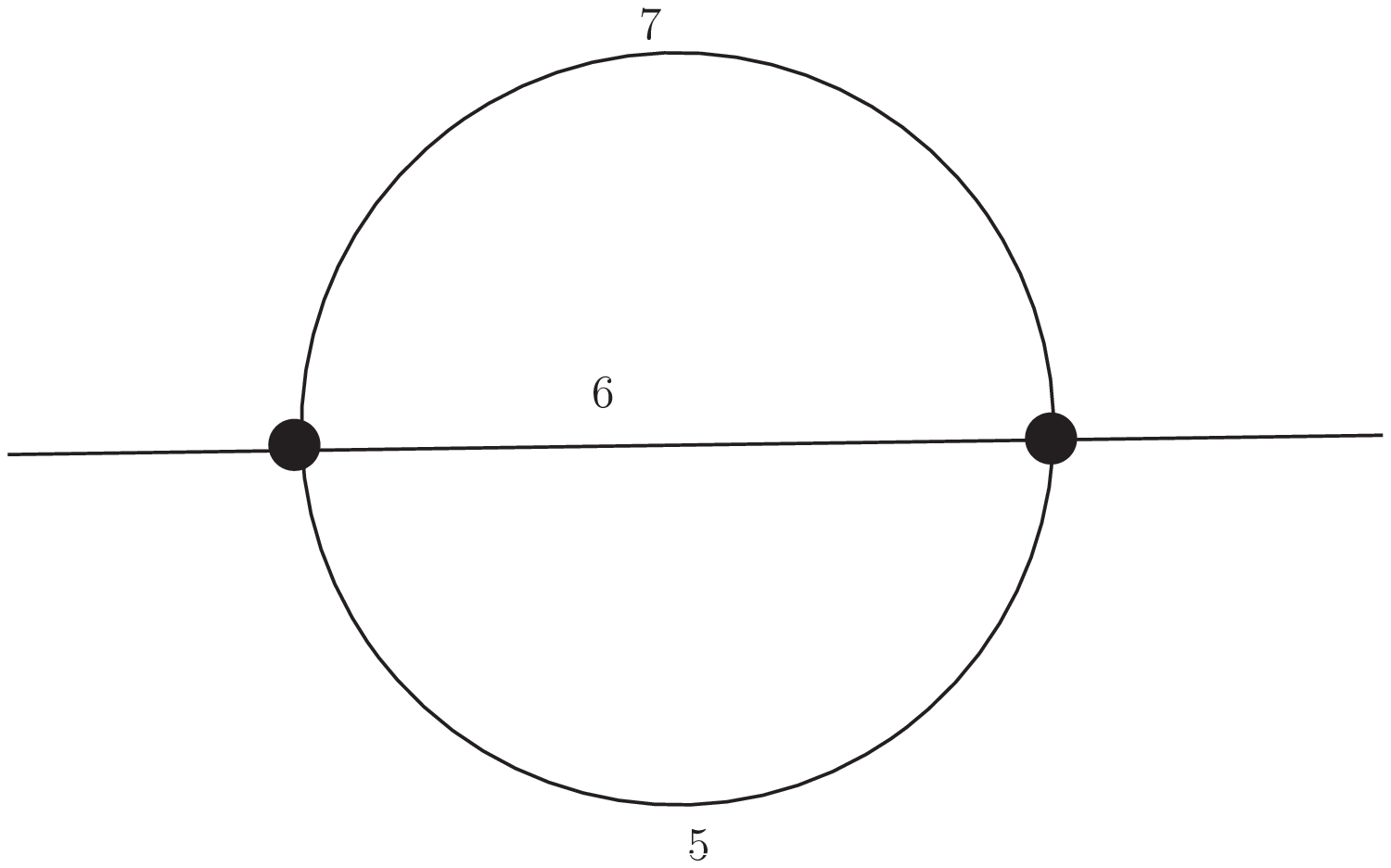}}\;}
\def\suboltpdt{\;\raisebox{-10mm}{\epsfysize=24mm\epsfbox{suboltpdt.eps}}\;}

\def\wts{\;\raisebox{-5mm}{\epsfysize=10mm\epsfbox{wts.eps}}\;}
\def\wtts{\;\raisebox{-5mm}{\epsfysize=10mm\epsfbox{wtts.eps}}\;}
\def\subwts{\;\raisebox{-5mm}{\epsfysize=10mm\epsfbox{subwts.eps}}\;}
\def\subwtts{\;\raisebox{-5mm}{\epsfysize=10mm\epsfbox{subwtts.eps}}\;}
\def\wtssmall{\;\raisebox{-1mm}{\epsfysize=2mm\epsfbox{wts.eps}}\;}
\def\wttssmall{\;\raisebox{-1mm}{\epsfysize=2mm\epsfbox{wtts.eps}}\;}
\def\subwtssmall{\;\raisebox{-1mm}{\epsfysize=2mm\epsfbox{subwts.eps}}\;}
\def\subwttssmall{\;\raisebox{-1mm}{\epsfysize=2mm\epsfbox{subwtts.eps}}\;}

\def\olh{\;\raisebox{-0mm}{\epsfysize=2mm\epsfbox{olh.eps}}\;}
\def\olholh{\;\raisebox{-0mm}{\epsfysize=2mm\epsfbox{olholh.eps}}\;}
\def\olv{\;\raisebox{-1mm}{\epsfysize=4mm\epsfbox{olv.eps}}\;}
\def\tlr{\;\raisebox{-1mm}{\epsfysize=4mm\epsfbox{tlr.eps}}\;}
\def\tll{\;\raisebox{-1mm}{\epsfysize=4mm\epsfbox{tll.eps}}\;}
\def\tlld{\;\raisebox{-8mm}{\epsfysize=20mm\epsfbox{tlld.eps}}\;}
\def\tlg{\;\raisebox{-1mm}{\epsfysize=4mm\epsfbox{tlg.eps}}\;}
\def\flg{\;\raisebox{-1mm}{\epsfysize=4mm\epsfbox{flg.eps}}\;}
\def\tlgo{\;\raisebox{-1mm}{\epsfysize=4mm\epsfbox{tlgo.eps}}\;}
\def\tlgofnr{\;\raisebox{-2mm}{\epsfysize=6mm\epsfbox{tlgofnr.eps}}\;}
\def\tlgofnl{\;\raisebox{-2mm}{\epsfysize=6mm\epsfbox{tlgofnl.eps}}\;}

\def\One{\mathbb{I}}

%\bibliographystyle{plain}
%\bibliography{main}
\section{Introduction: The decomposition of $\mathrm{Spec}_{\mathrm{Feyn}}(H)$}
In this paper we analyse periods appearing in the renormalization of divergent Feynman graphs $\Gamma$.
Throughout, for any graph $\Gamma$, we let $\Gamma^{[1]}$ be the set of internal edges of $\Gamma$, and $\Gamma^{[0]}$ 
be the set of vertices of $\Gamma$.

For each vertex $v\in\Gamma^{[0]}$, we assign a four-momentum  $q(v)$; to each edge $e\in\Gamma^{[1]}$, we assign 
a mass $m_e$. 

Our study holds for quantum field theories quite generally, but we specialise in the following to the case of a scalar field theory with four-valent vertices in four dimensions of space-time, for concreteness.   

Amplitudes in quantum field theory can be written as a function of a chosen
energy variable $L=\ln(S/S_0)$.   
Here, $S$ sets the scale of the process  under consideration. We take $S$ to be a suitable linear combination 
of scalar products $q(v)\cdot q(w)$ of external momenta and squared masses $m_e^2$ such that $0<S\in\mathbb{R}$.
Dimensionless scattering angles $\Theta$ like $q(v)\cdot q(w)/S$ and $m_e^2/S$ 
are defined with respect to the choice of $S$. Throughout, we will assume that $S$ is chosen such that it only vanishes when 
all external momenta vanish. In these variables, amplitudes can be calculated as a perturbation expansion 
in terms of Feynman graphs $\Gamma$ as $\sum_\Gamma\mathbf{\Phi}^R(\Gamma)$. Here, the renormalzed Feynman rules $\mathbf{\Phi}^R$ are expressed in terms of such
angle and scale variables, and the graphs $\Gamma$ form a Hopf algebra $H$. For any choice of 
angle and scale variables, $\mathbf{\Phi}^R$ is in the group $\mathrm{Spec}(H)_{\mathbb{C}}$, and the restriction of this group to maps which originate from evaluation of graphs by Feynman rules defines a sub-group 
$G=\mathrm{Spec}_{\mathrm{Feyn}}(H)\subset\mathrm{Spec}(H)_{\mathbb{C}}$.

Such a chosen decomposition of the variables reflects itself then in a chosen decomposition of the group 
$G$ into two subgroups $G_{\mathrm{1-s}}$ and
$G_{\mathrm{fin}}$.  Elements $\mathbf{\Phi}\in G_{\mathrm{1-s}}$ are of the form 
\begin{equation}\mathbf{\Phi}(\Gamma)=\sum_{j=1}^{\mathrm{cor}(\Gamma)}p_j L^j,\end{equation}
where the coefficients $p_j$ are periods in the sense of algebraic geometry and are independent of the angles $\{\Theta\}$. The integer $\mathrm{cor}(\Gamma)$
is defined below, it basically counts how often divergent sub-graphs are nested in a graph.

We allow for renormalization conditions which are defined by kinematic constraints on Green-functions: 
we demand that such Green functions, regarded as functions of $S$ and $\{\Theta\}$, vanish (up to a specified order)
at a reference point (in $S,\{\Theta\}$-space) given by $S_0,\{\Theta_0\}$. 
We implement these constraints graph by graph.
Hence renormalized Green functions
as well as renormalized Feynman rules become functions of $S,S_0,\Theta,\Theta_0$. Here, $\Theta,\Theta_0$ stand
for the whole set of angles in the Feynman rules.

Elements $\mathbf{\Phi}_{\mathrm{fin}}\in G_{\mathrm{fin}}$
are of the form
\begin{equation}\mathbf{\Phi}_{\mathrm{fin}}(\Gamma)={c^\Gamma_0}(\Theta),\end{equation}
with ${c^\Gamma_0}(\Theta)$ an $L$-independent function of the angles. 

We hence obtain the decomposition of $G$ 
as a map $\mathbf{\Phi}^R\to (\mathbf{\Phi}_{\mathrm{fin}},\mathbf{\Phi}^R_{\mathrm{1-s}})$, 
which proceeds then by a twisted conjugation:
\begin{equation}
G\ni \mathbf{\Phi}^R(S,S_0,\Theta,\Theta_0)=\mathbf{\Phi}^{-1}_{\mathrm{fin}}(\Theta_0)\star\mathbf{\Phi}_{\mathrm{1-s}}^R(S,S_0)\star\mathbf{\Phi}_{\mathrm{fin}}(\Theta),
\end{equation}
with $\mathbf{\Phi}_{\mathrm{fin}}(\Theta_0),\mathbf{\Phi}_{\mathrm{fin}}(\Theta)\in G_{\mathrm{fin}}$ and $\mathbf{\Phi}^R_{\mathrm{1-s}}(S,S_0)\in G_{\mathrm{1-s}}$. The group law $\star$ and inversion ${}^{-1}$ are defined through the Hopf algebra underlying $G$.

In particular, an overall divergent Feynman graph with divergent sub-graphs 
will evaluate to
\begin{equation}\mathbf{\Phi}^R(\Gamma)=\sum_{j=0}^{\mathrm{cor}(\Gamma)}c_j^\Gamma(\Theta,\Theta_0) L^j,\label{gendiv}\end{equation}
while for an overall convergent graph
\begin{equation}\mathbf{\Phi}^R(\Gamma)=\sum_{j=0}^{\mathrm{cor}(\Gamma)-1}c_j^\Gamma(\Theta,\Theta_0) L^j,\label{genconv}\end{equation}
and the product structure manifests itself in particular in the facts that in Eq.(\ref{gendiv}),
$c^\Gamma_{\mathrm{cor}(\Gamma)}$ is in fact an angle-independent coefficient, 
while Eq.(\ref{genconv}) is non-leading in its $L$ expansion,
(it stops at  $c^\Gamma_{\mathrm{cor}(\Gamma)}-1$) with angle dependent coefficients throughout.

We proceed by first reviewing forests in Feynman graphs, which are the appropriate tool to label 
divergent sectors. They are most efficiently organized using the Hopf algebra structure of rooted trees.
Renormalized Feynman rules for a given Feynman graph $\Gamma$ evaluate the graph using a sum over all forests in the graph, 
evaluating divergent sub-graphs to renormalized values in each forest using suitable renormalization conditions. In a final step
the remaining overall divergence is subtracted using renormalization conditions on the whole graph.

Each forest corresponds to a sector in the initial Feynman integral which has a localized singularity along the corresponding sub-graphs. 
Signs in the sum over all forests are arranged such that in the total expression, the Feynman integral allows for integration in all sectors.
The underlying combinatorics corresponds then to a Hopf algebra.

We next introduce Feynman rules in parametric space. We put emphasis on giving
them as integrands for projective spaces. We discuss in detail subtractions for 
quadratic sub-divergences in accordance with BPHZ renormalization conditions for the propagator
(the 1PI self-energy is supposed to vanish 
to first order at $q^2=m^2$, with  $m^2$ being the physical mass). 

The resulting unrenormalized amplitudes, after suitable partial integrations, have strictly logarithmic singularities, 
overall as well as in proper sub-graphs. 

Inclusion-exclusion is then the mechanism which eliminates these logarithmic singularities by suitable subtractions along any singular sector, making use of the combinatorial properties of the underlying Hopf algebra structure. 

Logarithmic singularities are angle-independent.
The same combinatorial inclu\-sion-exclusion properties alluded to above allow us then to 
compensate singularities by using a different evaluation of graphs in which angle dependences have been eliminated,
by choosing enough masses $m_e$ and momenta $q(v)$ to be zero.
In particular, we choose an element $\bar{\mathbf{\Phi}}\in G_{\mathrm{1-s}}$ which maps a Feynman graph 
to the desired $L$-polynomial
\begin{equation}\bar{\mathbf{\Phi}}(\Gamma)=\sum_{j=1}^{\mathrm{cor}(\Gamma)}p_j L^j.\end{equation}
We define graphs $\bar{\Gamma}$ such that 
\begin{equation}\bar{\mathbf{\Phi}}(\Gamma)=\mathbf{\Phi}(\bar{\Gamma}).\end{equation}
Such graphs (dubbed 1-scale graphs) form a Hopf algebra $H_{\mathrm{1-s}}$ again.
It separates the periods completely from the angle-dependence of Feynman amplitudes 
on the level of graphs.  
 
We are then ready to decompose 
\begin{equation}\mathbf{\Phi}^R(\Gamma)=\sum_{j=0}^{\mathrm{cor}(\Gamma)}c_j(\Theta) L^j,\end{equation}
as announced above. This corresponds to defining a matrix $M_{ij}(\Gamma)$, where $i,j$ run over all forests
of a graph $\Gamma$. While this matrix defined below plays no prominent role in the conceptual developments discussed here, it will figure prominently in the transition of our approach to concrete algorithms.

For the Hopf algebra of 1-scale graphs, we give renormalization proofs and a derivation of the renormalization group from scratch, using mere combinatorial properties of graph polynomials and Hopf algebras and not allowing ourselves any further analytic input. 

The renormalized amplitudes of single-scale graphs are given by explicit integrals involving graph polynomials, and are motivic, i.e., periods of mixed Hodge structures of certain diagrams of hyper-surfaces.  In particular, this leads to a rigorous definition of the weight of renormalized amplitudes, and opens the possibility of a qualitative study of 
renormalized amplitudes from the point of view of algebraic geometry.

Summarizing, we provide an approach which delivers 
a renormalized Feynman graph as a well-defined sum of projective integrals in parametric space in a form
suitable to an algorithmic determination of its periodic content, on combining our results with the studies started in \cite{BrCMP,BrFeyn}.

\section{Forests}
It is useful to collect some notation first.
\subsection{Definitions}
For a 1PI superficially divergent graph $\Gamma$, we define a  forest $f$ to be a collection of 1PI proper superficially divergent
sub-graphs $\Gamma_i\subset\Gamma$, $i\in\mathcal{I}^f_\Gamma$ for some index set $\mathcal{I}^f_\Gamma$,
such that either they are disjoint: $\Gamma_i\cap\Gamma_j=\emptyset$, or contained in each other:
$\Gamma_i\subset\Gamma_j$ or $\Gamma_j\subset\Gamma_i$. In particular, a forest $f$ is a product of 1PI graphs:
$f=\prod_i\gamma_i$. By $\Gamma/f$ we denote the graph obtained by contracting the graphs $\gamma_i$ to points in $\Gamma$.

For a 1PI superficially divergent graph $\Gamma$, we define a maximal forest to be a forest and
furthermore, we demand that $p^f_\Gamma:=\Gamma/[\cup_{i\in \mathcal{I}^f_\Gamma}\Gamma_i]$ has no divergent sub-graph. 
We hence call the index set $\mathcal{I}^f_\Gamma$ maximal for $\Gamma$. 

For $f\ni \Gamma_i\subset\Gamma$, each index set $\mathcal{I}^f_\Gamma$ defines an index set $\mathcal{I}_i^f$ of
all forests strictly contained in $\Gamma_i$, i.e.\ such that 
$\Gamma_j\subset\Gamma_i$ $\forall j\in\mathcal{I}_i^f$. 

We call a forest complete, if $\mathcal{I}^f_\Gamma$ is maximal for $\Gamma$ and $\mathcal{I}_i^f$ 
maximal for each proper 1PI superficially divergent
sub-graph $\Gamma_i$  of $\Gamma$.

Each finite graph $\Gamma$ has a finite number $|C(\Gamma)|$ of complete forests. Here, the set of all
such complete forests is denoted by $C(\Gamma)$. Examples are below.

Such complete forests are in one-to-one correspondence with decorated rooted trees 
where the set of decorations $p_v$ (at vertices $v$) is given by 1PI superficially divergent  graphs free of sub-divergences,
\begin{equation}
 p_v:=\Gamma_i/\cup_{j\in \mathcal{I}^f_i}\Gamma_j.
\end{equation}

From now on, we write in obvious abuse of notation  $T\in C(\Gamma)$ for such a decorated rooted tree.

Note that the power set $P_E(T)$  of edges $E(T)$ of such a tree $T$ gives all possible cuts $c$ at the tree $T$: any $c\in P_E(T)$ 
defines, for a connected tree $T$, a union of connected components $T-c$ obtained by removing the edges $c$, with $R^c(T)$ the unique component containing
the root of $T$, and $P^c(T)$ the union of the remaining components. We have $\cup_{T\in C(\Gamma)} 2^{E(T)}$ as the set of all cuts available altogether, and denote by $(c,T)$ an element of this set.

$P^c(T)$ corresponds to a forest of 
$\Gamma$, with each of its connected components corresponding to a graph $\gamma_i$ in $f=\prod_i\gamma_i$
\footnote{$c\to P^c(T)$ furnishes a surjective map $F$ from $\cup_{T\in C(\Gamma)} 2^{E(T)}$ to the forests $f$ of $\Gamma$.
The set of pre-images $f_c=F^{-1}(f)$ gives a partition of  $\cup_{T\in C(\Gamma)} 2^{E(T)}$ which is a bijection
with the forests of $\Gamma$.}.

After having determined the set $C(\Gamma)$, all (non-empty) 
forests of $\Gamma$ are in bijection with (non-empty) sets $f_c$ of some cuts $(c,T)$. We describe them as follows.

If we let $|T|$ be the number of vertices of a tree $T$, a tree $T$ allows for $2^{|T|}$ cuts including the empty one.
For a graph $\Gamma$, this gives us $\sum_{T\in C(\Gamma)}2^{|T|}$ cuts $c$. By construction, a forest $f=\prod_i \gamma_i$ of a graph $\Gamma$ assigns to a graph the product $(\Gamma/f)\prod_i\gamma_i$. We have $|C(\Gamma/f)|\prod_i |C(\gamma_i)|$
cuts $c_i$ corresponding to the same forest, and let $f_c$ be the set of cuts $(c,T)$ which correspond to the same forest $f$. 

We often notate a cut $(c,T)$ using $T$ with marked edges, and notate the union $f_c$ then as a sum of such trees.
We have $\sum_{T\in C(\Gamma)}2^{|T|}=\sum_f |f_c|$ by construction\footnote{The cardinality $|f_c|$ of $f_c$ gives the number of sectors in $f$ and $\Gamma/f$.}.

We can hence label the forests of a graph $\Gamma$ by subsets of edges on some of the trees $T\in C(\Gamma)$:
\begin{equation}
 \sum_f = \sum_{T\in C(\Gamma)}\sum_{c\in f_c}=\sum_{T\in C(\Gamma)}\sum_{c\in P_E(T)}.
\end{equation}
Furthermore, we identify the empty forest (of $\Gamma$) with $\Gamma$ and write $\sum_f^\emptyset$ when we include it in the sum. If we allow forests also to contain $\Gamma$ itself, we double the sum of forests and write
$\sum_{[f]}$ for the corresponding sum.
\begin{example}\label{example1}
Consider the graph
\begin{equation}\Gamma=\oldc.\label{firstexa}\end{equation}
It has subgraphs $$\gamma_{34}=\otf,\gamma_l=\dcl,$$ and $$\gamma_r=\dcr.$$
We have $\gamma_{34}\subset\gamma_l$ and $\gamma_{34}\subset\gamma_r$.
Its forests are
\begin{eqnarray}
f_0 & = & \emptyset,|{f_\emptyset}_c|=2,\label{f0}\\
f_1 & = & \gamma_{34},|{f_1}_c|=2,\label{f1}\\
f_2 & = & \gamma_l,|{f_2}_c|=1,\label{f2}\\
f_3 & = & \gamma_r,|{f_3}_c|=1,\label{f3}\\
f_4 & = & \gamma_{34},\gamma_l,|{f_4}_c|=1,\label{f4}\\
f_5 & = & \gamma_{34},\gamma_r,|{f_5}_c|=1\label{f5}.
\end{eqnarray}
The forest $f_1$ is neither maximal nor complete. The forests $f_2$ and $f_3$ are both maximal, but incomplete.
The forests $f_4,f_5$ are both complete.
Hence, $C(\Gamma)=\{f_4,f_5\}$ is a two-element set. 

If we add the graph $\Gamma$ itself to the forests, we double the set, 
for each  $f_i$, we now have $f_i$ and $f_i\cup\Gamma$.
 
The decorated trees $T_4,T_5$ are complete forests. They are given as:
\begin{equation}
\tfour = T_4 \leftrightarrow \foroldcl,
\end{equation}
and
\begin{equation}
\tthree = T_5\leftrightarrow \foroldcr.
\end{equation}

We can find the decorations by shrinking all graphs in the subforests of a given forest:
we assign to the two maximal complete forests two rooted trees, the root corresponding to the vertex at the outermost box
\footnote{Also, we can describe those trees as $T_4:=(((3,4),1,2),5,6)$ and $T_5:=(((3,4),5,6),1,2)$,
where we indicate the tree structure by bracket configurations and decorations by the edge labels of the corresponding primitive graphs.
If we notate forests in trees by square brackets $[\ldots ]$ corresponding to cuts, then the correspondences are as follows:
$f_0 \leftrightarrow (((3,4),1,2),5,6) + (((3,4),5,6),1,2)$,
$f_1 \leftrightarrow (([3,4],1,2),5,6) + (([3,4],5,6),1,2)$,
$f_2 \leftrightarrow ([(3,4),1,2],5,6)$,
$f_3 \leftrightarrow ([(3,4),5,6],1,2)$,
$f_4 \leftrightarrow ([[3,4],1,2],5,6)$,
$f_5 \leftrightarrow ([[3,4],5,6],1,2)$.
The forests corresponding to $f_i\cup\Gamma$  are then notated by replacing the outermost $(\ldots)$ pair of brackets by $[\ldots]$}.
\end{example}

\subsection{Hopf structures}
We summarize the relevant Hopf algebra structures \cite{BergbKr} as follows.
\subsubsection{For trees}
For the free commutative algebra of decorated rooted trees $H_{\mathrm{Dec}}$ (typically, decorations are provided by either
the graphs $p_v$ or their set of edge labels) we have a co-product $\Delta_T$ defined by
\begin{equation}
 \Delta_T\circ B_+^{p}(\cdot)=B_+^{p}(\cdot)\otimes 1+(\mathrm{id}\otimes B_+^{p})\Delta_T,
\end{equation}
and an antipode given by
\begin{equation}
 S(T)=-T-\sum_{c\in P_E(T)}(-1)^{|c|}P^c(T)R^c(T),
\end{equation}
where $R^c(T)$ contains the root with decoration $p$ and $P^c(T)$ are the other trees in $T-c$.
$B_+^p$ are Hochschild 1-cocycles, see \cite{BergbKr} for details.

We let $\mathrm{shad}:H_{\mathrm{Dec}}\to H_\emptyset$ be the map which forgets decorations.
\subsubsection{For graphs}
For graphs we have a Hopf algebra of graphs $H_\Gamma$ with co-product 
\begin{equation} \label{DirkDeltadef}
 \Delta_G(\Gamma)=\Gamma\otimes 1+ 1\otimes\Gamma+\sum_{\gamma\subset\Gamma}\gamma\otimes\Gamma/\gamma,
\end{equation}
where $\gamma$ is a disjoint union $\gamma=\cup_i\gamma_i$ of 1PI graphs which are superficially divergent.

The antipode is given by
\begin{equation}
 S(\Gamma)=-\Gamma-\sum_f (-1)^{|f|}\gamma_f\otimes\Gamma/\gamma_f.
\end{equation}
We have a Hopf algebra homomorphism $\rho:H_\Gamma\to H_{\mathrm{Dec}}$ given by  $\rho(\Gamma)=\sum_{T\in C(\Gamma)}T$ and with 
\begin{equation}
 [\rho\otimes\rho]\Delta_G=\Delta_T\rho.
\end{equation}
For any Hopf algebra $H\in (H_{\mathrm{Dec}},H_\Gamma)$ , we let $P$ be the projection into the augmentation ideal. 
We set $\sigma:=S\star P\equiv m_H(S\otimes P)\Delta $, which vanishes on scalars $\mathbb{Q}\One$.
For the Hopf algebra of graphs, one has $\sigma(\Gamma)=\sum_f^\emptyset (-1)^{|f|}f(\Gamma/f)$. 
We need a well-known lemma:
\begin{lem}
Let $\mathrm{id}_{\mathrm{Aug}}$ be the identity map $\mathrm{Aug}\to \mathrm{Aug}$ in the augmentation ideal. We have 
\begin{equation}
 \mathrm{id}_{\mathrm{Aug}}=\sum_{j=1}^\infty \sigma^{\star j}=:\sum_{j=1}^\infty \sigma_j.\label{empty}\end{equation}
\end{lem}
Note that the sums terminate when applied to any element of finite degree in the Hopf algebra.
Hopf algebras $H$  allow for a co-radical filtration 
\begin{equation}
\mathbb{Q}\One=H^{(0)}\subset H^{(1)}\cdots\subset H^{(n)}\subset
\cdots\subset H.
\end{equation}
The maps $\sigma_j$ vanish on elements in the Hopf algebra which are in $H^{(k)}$, $k<j$, and the coradical filtration is defined by the kernels of $\sigma_j$: elements in $H^{(k)}$ vanish when acted upon by $\sigma_j,\forall j>k$.

Now any map $\sigma_j$ above corresponds to a finite sum over forests $\bar{\sigma_j}$. 
As the empty forest corresponds to the identity map of a graph $\Gamma$, we can write for forests
\begin{equation}
 \emptyset=\sum_{j=1}^\infty \bar{\sigma_{j}}.
\end{equation}
The following gives an example for the maps $\bar{\sigma}_j$, acting on the graph $\Gamma$ of Example \ref{example1}.
\begin{example}
\begin{eqnarray}
\bar{\sigma_1}: & & \soldc-\soldcbm-\soldcbl\nonumber\\ & & -\soldcbr+\soldcbml+\soldcbmr,\\
\bar{\sigma_2}: & & \soldcbm+\soldcbl+\soldcbr\nonumber\\ & & -2\soldcbml-2\soldcbmr,\label{multtwo}\\
\bar{\sigma_3}: & & \soldcbml+\soldcbmr.
\end{eqnarray}
Note the multiplicity two generated in two terms in $\bar{\sigma_2}=\bar{\sigma_1}\star\bar{\sigma_1}$
in line (\ref{multtwo}),
coming from the fact that the subgraphs $\gamma_2,\gamma_3$ and the cograph $\Gamma/\gamma_1$ are acted upon
by $\bar{\sigma_1}$ with the same results.
\end{example}
\subsection{(Co-)ideals}
For any decorated rooted trees $T_1,T_2$ we say that $T_1$ is equivalent to $T_2$, $T_1\sim T_2$, if the trees agree as undecorated trees, $\mathrm{shad}(T_1)=\mathrm{shad}(T_2)$,
and at each vertex $v$, the decorations are by primitive graphs $p_1(v),p_2(v)$ which differ only by the the choice of external momenta $q(v)$
at vertices of $p_i$ and masses $m_e$ at internal edges of $p_i$. We say that two graphs are equivalent, $\Gamma_1\sim \Gamma_2$, if $C(\Gamma_1)$ and $C(\Gamma_2)$ 
are two sets of trees which are pairwise equivalent, and similar if the $\Gamma_i$ are sums of graphs.
In particular, if $\Gamma=\sum_i\gamma_i$ say, $C(\Gamma)=\cup_iC(\Gamma_i)$.
\subsection{Notation for forests}
Finally, we denote a graph with forests generated by a particular $T\in C(\Gamma)$ as follows.
For a graph with edges labelled $1,\ldots,m$ say, for each $T\in C(\Gamma)$ 
we can label the vertices $v\in T$ by those edges of $\Gamma$
which correspond to $p_v$. Note that this loses the information how and by which edges a decoration $p_v$ was connected to $p_w$, where $w$ is the vertex above (closer to the root than) the vertex $v$. This ambiguity will soon give us the freedom to define the desired 1-scale structures.
\section{Derivation of renormalized Feynman rules in parametric space}
We turn to the derivation of Feynman rules.
In our Hopf algebra $H_\Gamma$, we have graphs $\Gamma$ with labelled edges $e\in\Gamma^{[1]}$. To a graph $\Gamma$,
we will assign forms $\Phi_\Gamma$ which depend on the edge labels $A_e$,
the squared masses $m_e^2$, and the momenta $q(v)$, $v\in\Gamma^{[0]}$. Physicists may wish to consider these external momenta $q(v)$ as external edges, with a splitting as in say $q(v)=q_1+q_2$ corresponding to two external edges at $v$, if so desired (for example to achieve homogeneity in the valence of vertices). 

We assume that for a product of graphs $\Gamma_1\Gamma_2$, labels are not repeated. The forms $\Phi_\Gamma$ have the structure
$\Phi_\Gamma=f_\Gamma(\{A_e\})\Omega_\Gamma$, with $f_\Gamma(\{A_e\})$ a function of all the edge variables and $\Omega_\Gamma$ a standard form, see below.
With unrepeated labels, $\Phi_{\Gamma_1\Gamma_2}=f_{\Gamma_1}f_{\Gamma_2}\Omega_{\Gamma_1\cup\Gamma_2}$.

Renormalized Feynman rules make use of the Hopf algebra $H_\Gamma$ to construct a linear combination of forms $\Phi^R_\Gamma$
such that it can be integrated against positive real projective  $\mathbb{P}^{|\Gamma^{[1]}|-1}$-space.
We write $\mathbf{\Phi}^R(\Gamma)\in G=\mathrm{Spec}_{\mathrm{Feyn}}(H)$ for the resulting integral.  
\subsection{Schwinger parametrization and the exponential integral}
We first define the two graph polynomials $\psi,\varphi$. Both are configuration polynomials \cite{KrBl}. 
We define them here though using spanning trees and spanning forests. 
We have (for a connected graph $\Gamma$)
\begin{equation}
 \psi_\Gamma:=\sum_{T}\prod_{e\not\in T} A_e,
\end{equation}
for spanning trees $T$ and edges $e$ of $\Gamma$.
Furthermore, we let $q(v)$ be the external momentum entering a vertex $v\in\Gamma$ (it can be zero),
and for a subset of vertices $X\subset\Gamma$, we let $Q(X)=\sum_{v\in X}q(v)$.
Then,
\begin{equation}
\varphi_\Gamma:=\sum_{T_1\cup T_2}Q(T_1)\cdot Q(T_2)\prod_{e\not\in T_1\cup T_2}A_e,
\end{equation}
where $T_1\cup T_2$ is a spanning two-forest.
Note: $Q(T_1)=-Q(T_2)$, $Q(T_1)^2=Q(T_2)^2=-Q(T_1)\cdot Q(T_2)$. 
Note that if $\Gamma$ has only two distinct vertices $v_1,v_2$ say at which external momenta $q,-q$ enter 
(we call such a graph a two-point graph) we can write
$\varphi_\Gamma=-q^2 \psi_{\Gamma^\bullet}$, where $\Gamma^\bullet$ is the graph obtained from $\Gamma$ by identifying the two vertices $v_1,v_2$. We extend these definitions to products of graphs as follows.
For $\gamma=\prod_i\gamma_i$,
\begin{equation}
\psi_\gamma=\prod_i\psi_{\gamma_i},\,\varphi_\gamma=\sum_i\left(\varphi_{\gamma_i}\prod_{j\not= i}\psi_{\gamma_j}\right).
\end{equation}
For 1-scale graphs $\gamma$, $\varphi_\gamma$ becomes the circular join introduced below, see Eq.(\ref{PsiGcirc}).

Define $Q_{vw}:=q(v)\cdot q(w)$, let $S:=\sum_{v,w\in \Gamma^{(0)}}c_{vw}Q_{vw}$
a real ($c_{vw}\in\mathbb{R}$) linear combination of scalar products $Q_{vw}$ which vanishes only when all external momenta $q(v)$ vanish.
We say that $S$ is in general kinematic position. Let $\Theta_{vw}:=Q_{vw}/S$ and $\Theta_e:=m_e^2/S$.
\begin{equation}
 \varphi_\Gamma(\Theta):=\frac{\varphi_\Gamma}{S},\; \phi_\Gamma(S,\Theta):=S\phi_\Gamma(\Theta),\, \phi_\Gamma(\Theta):=\varphi_\Gamma(\Theta)
 +\psi_\Gamma \left(\sum_{e}A_e \Theta_e\right). 
\end{equation} 
We usually write $\phi_\Gamma\equiv\phi_\Gamma(S,\Theta)$ in the decomposed form (and in slight abuse of notation) as $\phi_\Gamma=S\phi_\Gamma(\Theta)$.
Extension to products is defined as before.

For a two-point graph $\Gamma$, we define
\begin{equation}
 \varphi_\Gamma =-q^2 \psi_{\Gamma^\bullet},\; \phi_\Gamma(\Theta):=-\psi_{\Gamma^\bullet}
 +\psi_\Gamma \left(\sum_{e}A_e \frac{m_e^2}{q^2}\right). 
\end{equation}
For such a two-point graph $\gamma$, we also define
\begin{equation}
 \bar{\psi}_{\gamma^\bullet}:=\left(-\psi_{\Gamma^\bullet}
 +\psi_\Gamma \left(\sum_{e}A_e \frac{m_e^2}{m^2}\right)\right), \label{defpsibar}
\end{equation}
with $m^2$ fixing the scale at which a two-point graph is subtracted.

We have for any $\gamma\subset\Gamma$, with $\gamma=\cup_i\gamma_i$, $\psi_\gamma=\prod_i\psi_{\gamma_i}$,
\begin{prop}
 \begin{eqnarray}
\psi_\Gamma & = & \psi_{\Gamma/\gamma}\psi_\gamma+R^\Gamma_\gamma, |R^\Gamma_\gamma|_\gamma=|\psi(\gamma)|_\gamma+1,\label{cher1}\\
\phi_\Gamma(\Theta) & = & \phi_{\Gamma/\gamma}(\Theta)\psi_\gamma
+\bar{R}^\Gamma_\gamma(\Theta), |\bar{R}^\Gamma_\gamma(\Theta)|_\gamma\geq |\psi(\gamma)|_\gamma+1,\label{cher2}
   \end{eqnarray}   
and $|\phi_\Gamma|=|\psi_\Gamma|+1$, and $|U|_V$ is the degree of $U$ in the edge variables of $V$, and $|U|=|U|_U$.
\end{prop}
Note that $\phi_{\Gamma/\gamma}(\Theta)$ can be zero, for example when masses are zero and $Q(T_i)=0$ for all two-forests of $\Gamma/\gamma$.
\begin{proof}
From the definitions via spanning trees and two-forests.
\end{proof}

We now let $\Box_\Gamma$ be the hypercube $\mathbb{R}_+^{|\Gamma^{(1)}|}$, and consider the integrand
obtained from a  Schwinger parametrization of a Feynman graph $\Gamma$, 
\begin{equation}
\Phi_\Gamma(S,\Theta):= 
\frac{dA_1\cdots dA_{|\Gamma^{(1)}|}e^{+\frac{S\phi_\Gamma(\Theta)}{\psi_\Gamma}}}{\psi^2_\Gamma}.
\end{equation}
This unrenormalized integrand cannot be integrated yet in the edge variables $A_e$ against $\Box_\Gamma$.
Its renormalized counterpart has the form (say for logarithmic divergences, the general case is below and has the same structure)
\begin{eqnarray}
\Phi^R_\Gamma(S,S_0,\Theta,\Theta_0) & = & \sum_{[f]}^\emptyset (-1)^{|f|}\Phi_f(S_0,\Theta_0)\Phi_{\Gamma/f}(S,\Theta).\label{tradfor}\\
 & = & \sum \Phi^{-1}_{\Gamma^\prime}(S_0,\Theta_0)\Phi_{\Gamma^{\prime\prime}}(S,\Theta),
\end{eqnarray}
where we used Sweedler's notation $\Delta_G(\Gamma)=\sum \Gamma^\prime\otimes\Gamma^{\prime\prime}$ in the second line. 
This is the traditional forest formula\footnote{Note that we use $\psi_\emptyset=1$, $\phi_\emptyset(\Theta)=0$.}.

In the following, we will renormalize this integrand using kinetic renormalization schemes. 
For that, we let $2s_\Gamma\equiv 2sd(\Gamma)$ be the superficial
degree of divergence of $\Gamma$ (in the example of a massive scalar field theory with quartic interactions):
\begin{equation}
 2s_\Gamma=4|\Gamma|-2|\Gamma^{[1]}|.
\end{equation}
Then, all vertex graphs $\Gamma$ have $s_\Gamma=0$ 
together with  $|\Gamma^{[1]}|=2|\Gamma|$, while for all propagator graphs, $s_\Gamma=1$
with $|\Gamma^{[1]}|=2|\Gamma|-1$.

Let us introduce new variables $A_e\to a_e, A_e=t a_e$, and $dA_1\cdots dA_{|\Gamma^{[1]}|}\to dt\wedge \Omega_\Gamma$,
with $\Omega_\Gamma$ the usual $(|\Gamma^{[1]}|-1)$-form $A_1dA_2\wedge\cdots\wedge dA_{|\Gamma^{[1]}|}-A_2dA_1\cdots\pm\cdots$, see Eq.(\ref{Omegag}).
We find
\begin{equation}
\Phi_\Gamma:= 
\frac{dt}{t}\wedge\frac{\Omega_\Gamma e^{t\frac{S\phi_\Gamma(\Theta)}{\psi_\Gamma}}}{t^{s_\Gamma}\psi^2_\Gamma}.\label{basint}
\end{equation}
We want to study the overall $t$-integration as a function of the superficial degree of divergence $s_\Gamma$ first.
Concretely, we are interested to define and find the limit in the $t$-integration
\begin{equation}
\lim_{c\to 0} \int_c^\infty \Phi_\Gamma,
\end{equation}
where $c\in\mathbb{R}_+$.
We use renormalization conditions on $\Phi_\Gamma\equiv\Phi_\Gamma(S,\Theta)$.

Kinetic renormalization conditions imply that we choose values $S_0,\Theta_0$ for the scale and for the angles, such that the renormalized amplitudes of a graph $\Gamma$, together with their first $s_\Gamma$ derivatives in an expansion around that point,
vanish.\\
For $s_\Gamma=0$, we can simply subtract at a chosen $S_0,\Theta_0$: 
\begin{equation}
\Phi_\Gamma(S,\Theta)\to \Phi_\Gamma(S,\Theta)-\Phi_\Gamma(S_0,\Theta_0)
\end{equation}
which takes care of the overall divergence in the graph $\Gamma$.\\
For $s_\Gamma=1$, we are dealing with a quadratically divergent propagator function. We will subtract at $q^2=m^2$.
Note that there are no angles $\Theta_{vw}$  for a two-point function, the $\Theta_e$ remain though.
Kinetic renormalization conditions are determined by the requirement that the renormalized amplitude vanishes at $q^2=m^2$, together 
with its first derivative $\partial_{q^2}$, so that the pole in the propagator has a on-shell unit residue\footnote{For a massless propagator, 
vanishing of $\Phi^R_\Gamma$ at $q^2=0$ and of $\Phi^R_\Gamma/q^2$ at $q^2=\mu^2$ are also convenient renormalization conditions.}.

For overall convergent Green functions, the limit $c\to 0$ can be taken without having to impose overall constraints, 
the introduction of constraints (which correspond to  over-subtractions) to impose some favoured kinematics 
would be possible though in this set-up.
\subsection{$s_\Gamma=0$}
Let us start with the case $s_\Gamma=0$.
The limit is
\begin{equation}
 \lim_{c\to 0}\int_c^\infty [\Phi_\Gamma(S,\Theta)-\Phi_\Gamma(S_0,\Theta_0)]=
\frac{\Omega_\Gamma\ln{\frac{S\phi_\Gamma(\Theta)}{S_0\phi_\Gamma(\Theta_0)}}}{\psi^2_\Gamma},
\end{equation}
using that for small $c>0$,
\begin{equation}
 \int_c^\infty \frac{e^{-tX}dt}{t}=-\ln c+\ln X +\gamma_E+\mathcal{O}(c).
\end{equation}
Here, $\gamma_E$ is the Euler--Mascheroni constant. Note that we can decompose the logarithm as 
$$
\ln{\frac{\frac{S}{S_0}\phi_\Gamma(\Theta)}{\phi_\Gamma(\Theta_0)}}=\ln(S/S_0)+\ln(\phi_\Gamma(\Theta)/\phi_\Gamma(\Theta_0)),$$ 
(we assume $S/S_0>0$). We assume also that the angles $\Theta,\Theta_0$ are chosen such that we are off Landau singularities. 
Approaching such singularities means studying the corresponding variation of the logarithm above.

Let us now look at logarithmic sub-divergences.
A typical term in the forest formula provides an integrand of the form 
\begin{equation}
\frac{e^{+\frac{S\phi_{\Gamma/f}(\Theta)}{\psi_{\Gamma/f}}}}{\psi^2_{\Gamma/f}}\frac{e^{+\frac{S_0\phi_{f}(\Theta_0)}{\psi_{f}}}}{\psi^2_{f}}
-
\frac{e^{+\frac{S_0\phi_{\Gamma/f}(\Theta_0)}{\psi_{\Gamma/f}}}}{\psi^2_{\Gamma/f}}\frac{e^{+\frac{S_0\phi_{f}(\Theta_0)}{\psi_{f}}}}{\psi^2_{f}}.
\end{equation}
Combining each of the two products of exponentials into a single exponential and using the exponential integral as above
delivers 
\begin{equation}
M^\Gamma_f:=\frac{\ln{\frac{S\phi_{\Gamma/f}(\Theta)\psi_f+S_0\phi_f(\Theta_0)\psi_{\Gamma/f}}{S_0\phi_{\Gamma/f}(\Theta_0)\psi_f+S_0\phi_f(\Theta_0)\psi_{\Gamma/f}}}}{\psi^2_{\Gamma/f}\psi^2_f} \Omega_\Gamma.
\end{equation}
Summing over all forests including the empty one delivers the renormalized integrand as the homogeneous of degree zero form
\begin{equation}\label{renFR}
\Phi^R_\Gamma:=\sum_f^\emptyset (-1)^{|f|} M^\Gamma_f.
\end{equation}
$\Phi^R_\Gamma$ is an integrand which can, this is just a rewriting of the forest formula,  be integrated against $\mathbb{P}^{|\Gamma^{[1]}|-1}(\mathbb{R}_+)$.
An explicit proof from scratch is given below though, after we decomposed Feynman rules suitably. 
\begin{rem}
If we write $\varphi_\Gamma=S\sum_{v,w}\Theta_{vw}c_\Gamma^{vw}$ for some monomials 
$$c_\Gamma^{vw}=c_\Gamma^{vw}(A_1,\ldots,A_{|\Gamma^{[1]}|}),$$ in edge variables $A_e$, then
\begin{equation}
-\partial_{\Theta_{v,w}}\frac{\phi_{\Gamma/f}(\Theta)}{\psi_{\Gamma/f}}=\frac{1}{\psi_{\Gamma/f}}c_{\Gamma/f}^{vw}(\partial_{m_1^2,\ldots,\partial_{m^2_{|\Gamma^{[1]}|}}}).
\end{equation}
As an extra power of $\psi_{\Gamma/f}$ in the denominator is equivalent to a shift in the dimension, we can re-derive
in this way the usual recursion and dimension shift relations between master integrals.
\end{rem}

Next we treat other overall degrees of divergence still with logarithmic sub-divergences, 
and will then treat the general case including quadratic sub-divergences.
\subsection{$s_\Gamma=1$}
Let us next look at the case $s_\Gamma=1$. We have $\Phi_\Gamma=\Phi_\Gamma(q^2,m^2)$ and for $s_\Gamma=1$ we look at the Taylor expansion $\Phi_\Gamma=\Phi_\Gamma(m^2,m^2)+(q^2-m^2)\partial_{q^2}\Phi_\Gamma(m^2,m^2)
+\mathcal{O}((q^2-m^2)^2)$.
Partially integrating Eq.(\ref{basint}), we find
\begin{equation}
 \int_c^\infty \Phi_\Gamma=\int_c^\infty
\frac{dt}{t}\wedge\frac{\Omega_\Gamma e^{-t\frac{S\phi^S_\Gamma}{\psi_\Gamma}}}{t\psi^2_\Gamma}=
\frac{1}{c}\left[\frac{\Omega_\Gamma e^{-c\frac{S\phi^S_\Gamma}{\psi_\Gamma}}}{\psi^2_\Gamma}\right]
-\frac{S\phi_\Gamma^S}{\psi_\Gamma}\int_c^\infty
\frac{dt}{t}\wedge\frac{\Omega_\Gamma e^{-t\frac{S\phi^S_\Gamma}{\psi_\Gamma}}}{\psi^2_\Gamma}.
\label{quadr}
\end{equation}
We set $S=q^2$, so $\Theta_e=m_e^2/q^2$, and $S_0=m^2,\Theta_{e,0}=m_e^2/m^2$. If all masses are equal, we simply have $\Theta_0=1$.

In Eq.(\ref{quadr}) we have on the rhs a boundary term (the term $\sim 1/c$ in $[\cdots]$ brackets) and a logarithmically divergent integral $\sim \ln c$.
Subtracting at $q^2=m^2$ gives for the boundary term a contribution:
\begin{equation}
-\left[\frac{\Omega_\Gamma (q^2-m^2)\psi_{\Gamma^\bullet}}{\psi^3_\Gamma}\right]
\end{equation}
where $\Gamma^\bullet$ is the graph obtained from identifying the two external vertices of $\Gamma$.
Note that this boundary term contains no higher terms in $(q^2-m^2)$ in the limit $c\to 0$.
So in the limit this is an expression linear in $(q^2-m^2)$ which hence vanishes in kinematic renormalization.

For the logarithmic divergent integral on the rhs of Eq.(\ref{quadr}) we find 
\begin{equation}
-\frac{S\phi_\Gamma^S}{\psi_\Gamma}\int_c^\infty
\frac{dt}{t}\wedge\frac{\Omega_\Gamma e^{-t\frac{S\phi^S_\Gamma}{\psi_\Gamma}}}{\psi^2_\Gamma}
+\frac{S_0\phi_\Gamma^{S_0}}{\psi_\Gamma}\int_c^\infty
\frac{dt}{t}\wedge\frac{\Omega_\Gamma e^{-t\frac{S_0\phi^{S_0}_\Gamma}{\psi_\Gamma}}}{\psi^2_\Gamma}.
\end{equation}
Subtraction of the term linear in $(q^2-m^2)$ delivers a renormalized integrand free of overall divergences which reads as :
\begin{equation}
 M^\Gamma_\emptyset=\frac{m^2\bar{\psi}_{\Gamma^\bullet}}{\psi_\Gamma}\frac{
\left[x^\Gamma_\emptyset\ln(1+x^\Gamma_\emptyset)+\ln(1+x^\Gamma_\emptyset)-x^\Gamma_\emptyset\right]\Omega_\Gamma}{\psi^2_\Gamma},\label{quadroverall}
\end{equation}
where we write $\phi_\Gamma(q^2,m^2)/\phi_\Gamma(m^2,m^2)$ as $1+x^\Gamma_\emptyset$, 
with $$x^\Gamma_\emptyset=\frac{(q^2-m^2)\psi_{\Gamma^\bullet}}{m^2\bar{\psi}_{\Gamma^\bullet}}.\label{xzerogamma}$$
We remind the reader of Eq.(\ref{defpsibar}) which gives $\bar{\psi}_{\Gamma^\bullet}$.
\subsection{$s_\Gamma=-1$} In that overall convergent case, with 
 $1+x^\Gamma_\emptyset:=\frac{S\phi_\Gamma(\Theta)}{S_0\phi_\Gamma(\Theta_0)}$ similarly, 
we have 
\begin{equation}
 M^\Gamma_\emptyset=\left(\frac{S_0\phi_\Gamma(\Theta_0)}{\psi_\Gamma}\right)^{-1}\frac{\Omega_\Gamma}{\psi^2_\Gamma(1+x^\Gamma_\emptyset)}.
\end{equation}
Similarly for the cases $s_\Gamma<-1$. One finds for $s=-k$:
\begin{equation}
 M^\Gamma_\emptyset=\left(\frac{S_0\phi_\Gamma(\Theta_0)}{\psi_\Gamma}\right)^{-k}\frac{(k-1)!(-1)^{k-1}\Omega_\Gamma}{\psi^2_\Gamma(1+x^\Gamma_\emptyset)^k}.\label{convoverall}
\end{equation}
\subsection{Quadratic sub-divergences}
We saw before that a partial integration improved an overall quadratically divergent integrand to an overall logarithmic one,
plus a boundary term which captured the overall quadratic divergence. That boundary term was eliminated by the BPHZ renormalization conditions for the  mass term, leaving us to treat a logarithmic divergent integrand.

We want similarly to use partial integrations to treat quadratic sub-divergences.

We will treat quadratic sub-divergences according to their partial ordering, which is discussed in detail below, see \S\ref{sectQuadratic}. 
To start,  assume $\gamma\subset\Gamma$ is a quadratically divergent sub-graph of $\Gamma$
and $\gamma$ is itself free of quadratic divergent sub-graphs.   

We first consider the Feynman integral in momentum space. Label the edges of $\gamma$ by $1,\cdots,2|\gamma|-1$.
We use a parametric representation for the internal propagators of $\gamma$ and integrate out the internal loop momenta of $\gamma$.
We let $e_\gamma,f_\gamma$ denote the two edges (connectors) which connect $\gamma$ to the other vertices of 
$\Gamma-\gamma$, with $1/P_{e_\gamma}(q),1/P_{f_\gamma}(q)$ their propagators.

This gives a factor
\begin{equation}
F_\gamma:=\frac{e^{-\frac{\phi_\gamma}{\psi_\gamma}}dA_1\cdots dA_{2|\gamma|-1}}{P_{e_\gamma}(q)P_{f_\gamma}(q)\psi^2_\gamma},
\end{equation}
in the Feynman integrand.

Isolating a variable $t_\gamma=A_1$ via $A_i=t_\gamma a_i$, $i=2,\cdots|\gamma^{[1]}|$ 
allows us to study the behaviour of the integrand with respect to the overall divergence of the quadratically divergent sub-graph.

\begin{equation}
F_\gamma:=\frac{e^{-t_\gamma\frac{\phi_\gamma}{\psi_\gamma}}da_2\cdots  da_{2|\gamma|-1}}{P_{e_\gamma}(q)P_{f_\gamma}(q)\psi^2_\gamma}\wedge\frac{dt_\gamma}{t_\gamma^2}.
\end{equation}
In obvious abuse of notation, we still write $\phi_\gamma,\psi_\gamma$ after the change of variables.

Partially integrating $\int_{c_\gamma}^\infty F_\gamma$ in $t_\gamma$ delivers
\begin{eqnarray}
\int_{c_\gamma}^\infty F_\gamma 
 & = & 
\int_{c_\gamma}^\infty \frac{e^{-t_\gamma\frac{\phi_\gamma}{\psi_\gamma}}da_2\cdots  da_{2|\gamma|-1}}{P_{e_\gamma}(q)P_{f_\gamma}(q)\psi^2_\gamma}\wedge\frac{dt_\gamma}{t_\gamma^2}\\
 & = &
\int_{c_\gamma}^\infty\frac{-1}{P_{e_\gamma}(q)P_{f_\gamma}(q)}\underbrace{\frac{\phi_\gamma e^{-t_\gamma\frac{\phi_\gamma}{\psi_\gamma}}da_2\cdots  da_{2|\gamma|-1}}{\psi^3_\gamma}\wedge\frac{dt_\gamma}{t_\gamma}}_{X}\\
 & + & 
\frac{1}{P_{e_\gamma}(q)P_{f_\gamma}(q)}\underbrace{ \left[
\frac{e^{-c_\gamma\frac{\phi_\gamma}{\psi_\gamma}}da_2\cdots  da_{2|\gamma|-1}}{\psi^2_\gamma}\frac{1}{c_\gamma} 
  \right] }_{Y}.
\end{eqnarray}
We now impose BPHZ conditions for the sub-graph $\gamma$.
We start with the boundary term $Y\equiv Y(q^2,m^2)$.
We have 
\begin{equation}
Y(q^2,m^2)=Y(m^2,m^2)+
\frac{\overbrace{[\phi_\gamma(q^2,m^2)-\phi_\gamma(m^2,m^2)]}^{=(q^2-m^2)\psi_{\gamma^\bullet}}da_2\cdots  da_{2|\gamma|-1}}{\psi^3_\gamma}+\mathcal{O}(c_\gamma).
\end{equation}
In BPHZ conditions we subtract the term $Y(m^2,m^2)$ and the term linear in $(q^2-m^2)$. Hence, in the limit $c_\gamma\to 0$,
the $Y$ term leaves no contribution.

Let us now turn to the $X\equiv X(q^2,m^2)$ term. 
We use $\phi_\gamma(m^2,m^2)=m^2\overline{\psi}_{\gamma^\bullet}$.
Then\footnote{If we renormalize the subgraph $\gamma$ at $m_\gamma^2\not= m_{f_\gamma}^2$, we set $(q^2-m_\gamma^2)=(q^2-m_{\gamma_f}^2)+\Delta_m^2$, and compensate by a redefinition of $\bar{\psi}_{\gamma^\bullet}$.}, using $(q^2-m^2)=P_{f_\gamma}(q)$ and $\phi_\gamma(q^2,m^2)=(q^2-m^2)\psi_{\gamma^\bullet}+\phi_\gamma(m^2,m^2)$,
\begin{equation}
\frac{1}{P_{e_\gamma}(q)P_{f_\gamma}(q)}X=\frac{X_1}{P_{e_\gamma}}+\frac{m^2X_2}{P_{e_\gamma} P_{f_\gamma}},
\end{equation}
with
\begin{equation}
X_1:=\frac{(q^2-m^2)\psi_{\gamma^\bullet} e^{-\frac{\phi_\gamma}{\psi_\gamma}}dA_1dA_2\cdots  dA_{2|\gamma|-1}}{\psi^3_\gamma}
\end{equation}
and
\begin{equation}
X_2:=\frac{\overline{\psi}_{\gamma^\bullet} e^{-\frac{\phi_\gamma}{\psi_\gamma}}dA_1dA_2\cdots  dA_{2|\gamma|-1}}
{\psi^3_\gamma}.
\end{equation}
Note that we returned from $t_\gamma,a_i$ to $A_1,A_2,\ldots$ variables.
Summarizing, BPHZ conditions render a renormalized $X_R(q^2,m^2)$ as
\begin{eqnarray}
X_R(q^2,m^2) & = & P_{f_\gamma}(q)[X_1(q^2,m^2)-X_1(m^2,m^2)]\nonumber\\
 & + & \left[X_2(q^2,m^2)-X_2(m^2,m^2)\left(1-\frac{P_{f_\gamma}(q)\psi_{\gamma^\bullet}}{\psi_\gamma}\right)\right].\label{quadrsub}
\end{eqnarray}
The last term in the second line makes sure that we subtract from $X_2(q^2,m^2)$ its value at $q^2=m^2$ and its first
Taylor coefficient in the expansion in $(q^2-m^2)$.
 
We can now continue to parametrize the remaining propagators in $\Gamma-\gamma$ treating quadratic subdivergences in the order as dictated by their partial ordering (see below) and integrate out all loop momenta to find the final renormalized parametric representations below, now for 
the cases of quadratic and logarithmic sub-divergences jointly.
Summarizing,  we will give them for overall degrees of divergence $s_\Gamma\in\{-1,0,1\}$,
from overall convergence to quadratic divergence.
\subsection{The final result}
As before, we let $s_\Gamma=\mathrm{sd}(\Gamma)/2$.
Define 
\begin{equation}
1+ x^\Gamma_f:=\frac{\frac{S}{S_0}\phi_{\Gamma/f}(\Theta)\psi_f
+\phi_f(\Theta_0)\psi_{\Gamma/f}}{
\phi_{\Gamma/f}(\Theta_0)\psi_f+\phi_f(\Theta_0)\psi_{\Gamma/f}},\label{finalxvar}
\end{equation}
and the following series (which could be easily generalized to $l_j(x),j\in\mathbb{Z}$):
\begin{eqnarray}
 l_{-1}(x) & = & \frac{1}{(1+x)}  =  \sum_{k=0}^\infty (-1)^k x^k\\
 l_{0}(x) & = & \ln{(1+x)}  =  \sum_{k=1}^\infty (-1)^{k-1} \frac{x^k}{k}\\
 l_{1}(x) & = & (q^2-m^2)l_0(x)+\left(m^2\partial_{m^2}\frac{1}{x^\prime}\right)l_0(x)^{[1]},\label{lonex}
\end{eqnarray}
where we need to explain the notation in Eq.(\ref{lonex}) as follows.
First, $l_1(x)$ is needed only for two-point functions, in which case any variable $x=x(q^2,m^2)$, and we denote 
$x^\prime=\partial_{q^2}x(q^2,m^2)$. Then, 
\begin{equation}
m^2\partial_{m^2}\frac{1}{{x^\Gamma_f}^\prime}=m^2\frac{\bar{\psi}_{{\Gamma/f}^\bullet}}{\psi_{{\Gamma/f}^\bullet}}.
\end{equation}
Second, for any function $l_i(x)$, we denote by $l_i(x)^{[k]}$ the function $l_i(x)-\sum_{j=0}^i c_{i,j}x^j$,
where $c_{i,j}$ is the $j$-th Taylor coefficient in the expansion of $l_i$ around $x=0$.  Note that $l_0(x)=l_0^{[0]}(x)$
and $l_1(x)=l_1^{[1]}(x)$ by construction.

Let $n^X_2$ be the number of (sub)-graphs $\gamma \subseteq X$ with $\mathrm{sd}(\gamma)=2$. Define 
\begin{equation}
 \omega_X:=\prod_{\gamma\subseteq X,\mathrm{sd}(\gamma)=2}\frac{\psi_{\gamma^\bullet}}{\psi_\gamma}.
\end{equation}

For an integer $i\in\mathbb{Z}$, we let $\theta:\mathbb{Z}\to \{0,1\}$ be the function $\theta(i)=1,i>0$, $\theta(i)=0,i\leq 0$.

For all $n_2^\Gamma$ quadratic sub-graphs of $\Gamma$, choose a set $C_2$ of   
$n_2^\Gamma-s_\Gamma\theta(s_\Gamma)$ distinct edges in $\Gamma$ such that for each $\gamma$ at least one connector (the two edges in $\Gamma-\gamma$ which connect to $\gamma$) of $\gamma$ is an element  of $C_2$. A subscript $C_2/I$ denotes suppression of the edge
variables in $C_2/I$ in the corresponding expression.

We define $E_j^{C_2/I}$ to be the $j$-th elementary symmetric polynomial in the $|C_2/I|$ variables 
$\frac{m_e^2\bar{\psi}_{{\gamma_e}^\bullet}}{\psi_{{\gamma_e}}}$, $e\in C_2/I$. $E_0^{C_2/I}:=1$.
Then, with 
\begin{eqnarray}
 & & M(x_f^\Gamma)   :=\label{eqform}\\
 & &  \sum_{I\subseteq C_2}(-1)^{|I|}
 \left\{ \omega_{\Gamma/f}\omega_f\Omega_\Gamma \sum_{j=0}^{|C_2/I|}(-1)^j 
 \left(\frac{\phi_{\Gamma/f(\Theta_0)}}{\psi_{{\Gamma/f}^\bullet}}\right)^{\theta(|I|+j)\theta(s_\Gamma)}\times\right.\nonumber\\
 & & \left. \times\left(\frac{\overbrace{\phi_{\Gamma/f}(\Theta_0)\psi_f+\phi_f(\Theta_0)\psi_{\Gamma/f}}^{\mathrm{den}(1+x^\Gamma_f),\, \mathrm{see\, Eq.(\ref{finalxvar})}}}{\psi_{\Gamma/f}}\right)^{(s_\Gamma-|I|-j)\theta(-s_\Gamma+|I|+j)}
 \times\right.\nonumber\\ & & \left.
\times E_j^{C_2/I}
\frac{l_{s_\Gamma-|I|-j}^{[s_\Gamma]}(x^\Gamma_f)}{\underbrace{\psi_{\Gamma/f}^{2}\psi_f^{2}}_{-2|\Gamma|}}\right\}_{C_2/I}  \prod_{e\in I}m_e^2
\frac{\bar{\psi}_{{\gamma_e}^\bullet}}{\psi_{{\gamma_e}^\bullet}},\nonumber
\end{eqnarray}
we have
\begin{thm}\label{finalr}
 \begin{equation}
\mathbf{\Phi}^R_\Gamma=\sum^\emptyset_f(-1)^{|f|}\int_{\mathbb{P}^{n}(\mathbb{R}_+)}M(x_f^\Gamma),\label{finalresult}  
 \end{equation}
converges when integrated against projective $\mathbb{P}^n(\mathbb{R}_+)$ 
spaces of dimension $n=|\Gamma^{[1]}|-|C_2|+|I|$ for each $\Omega_{\Gamma_{C_2/I}}$.
\end{thm}
\begin{proof}
The formula is a straightforward but tedious rewriting of the forest formula in terms of projective integrals,
using Eqs.(\ref{quadroverall},\ref{convoverall},\ref{quadrsub}), and factorizing out $\omega_{\Gamma/f},\omega_f$
as this is convenient for the treatment of 1-scale graphs later on. 
Again, the direct proof is a consequence of the corresponding result for 1-scale graphs given below, and the factorization of Feynman rules to which we turn in the next section.
\end{proof}
\begin{prop}
 Tadpoles vanish in renormalization.
\end{prop}
\begin{proof}
 $\frac{S}{S_0}\phi_{\Gamma/f}(\Theta)=\phi_{\Gamma/f}(\Theta_0)$  if $\Gamma/f$ is a tadpole, 
 so $x^\Gamma_f=0$ and $l_i(0)=0$, $i=0,1$.
\end{proof}
\begin{rem}
Vanishing tadpoles define an obvious ideal and co-ideal. 
\end{rem}
\subsection{Auxiliary Feynman rules}\label{aux}
To next introduce Feynman rules such that the desired factorization
\begin{equation}
\mathbf{\Phi}^R_\Gamma=\mathbf{\Phi}_{\mathrm{fin}}(\Theta_0)^{-1}\star\mathbf{\Phi}^R_{\mathrm{1-s}}\star\mathbf{\Phi}_{\mathrm{fin}}(\Theta),
\end{equation}
can be established, we have to enlarge the set of Feynman rules.

For that, we first introduce auxiliary Feynman graphs.
Let $p$ be some decoration of $T\in C(\Gamma)$ at a vertex below the root. Let $E_p^{\bar{p}}$ be the set of 
edges which connect $p$ to the decoration $\bar{p}$ above $p$. 
We can assume $|E_p^{\bar{p}}|\geq 2$ since $\Gamma$ is one-particle irreducible.

Let $V_p^{\bar{p}}\subseteq p^{[0]}$ 
be the set of vertices of $p$ which are connected to
edges $e\in E_p^{\bar{p}}$. Let $V_p^{\mathrm{ext}}\subset \Gamma^{[0]}\cap p^{[0]}$
be external vertices of $\Gamma$ ($q(v)\not= 0$) which are also vertices of $p$.
Set $V_p:=V_p^{\bar{p}}\cup V_p^{\mathrm{ext}}$. Set $q_p:=\sum_{v\in V_p^{\mathrm{ext}}}q(v)$.

Both $E_p^{\bar{p}},V_p$ have cardinality two or greater. 

For $V_p$, we can hence choose two distinct vertices $v_p,w_p$.
For $E_p^{\bar{p}}$, we can choose a partition  into two non-empty sets $E_p^{\bar{p}}=F_p\cup G_p$. 
Set $q(v_p)=q_p$ and set $q(v)=0$ for all other vertices in $p^{[0]}$.

Let now $c\in P_E(T)$ be a cut of $T\in C(\Gamma)$, so $c$ is a collection of some edges of $T$.
To such a collection of edges corresponds a set of decorations $D(c)$ at the endpoint (away from the root) of these edges.
Each element in $D(c)$ plays the role of a decoration $p$ above.

The corresponding $\bar{p}$ decorate the vertices above $p$ for those edges $c$.
  
For each $p\in D(c)$, choose $v_p, w_p$ and partition $E_p^{\bar{p}}=F_p\cup G_p$. 

Attach all edges in $F_p$ to $v_p$ and all edges in $G_p$ to $w_p$. 
This defines, for each $T\in C(\Gamma)$ and each cut $c$ of $T$,  
some new graph $\Gamma_c^T$, which clearly depends on choices of $v_p,w_p$ and on choices of partitions of $E_p^{\bar{p}}$. 
Call these choices the choice of a 1-scale structure. 

In particular, consider a complete forest, that is some $T$ with all its edges chosen, $c=P_E(T)$.
We chose a 1-scale structure  for that complete forest $f$ defining $T$.
It induces  a compatible choice for any other forest in the same given tree $T$ (choices can be made independent between different trees in $C(\Gamma)$ though).

Do this for any $T\in C(\Gamma)$.

We now have a set of graphs $\Gamma_{f}=\sum_{c\in f_c}\Gamma^T_c$ for each forest of $\Gamma$.
By construction, $|V_p|=2$ for all $p,\bar{p}\in \Gamma_c^T$.
Finally, we set $m_e=0$ for all edges $e\not\in (\Gamma_f/f)^{[1]}$.
\begin{lem}\label{onescale}
i) For complete forests, we have the equivalence $$\sum_{f\;\mathrm{complete}}\Gamma_f\sim\Gamma.$$ ii)  For any forest $f$, $$|C(\Gamma_f)|=1.$$
\end{lem}
Note that this implies that under such choices, no new singular sectors are generated. The graphs $\Gamma^T_c$ are in particular 
free of overlapping divergences.
\begin{proof}
By assumption, all decorations $p$ in any $T\in C(\Gamma)$ are log-divergent. 
For any choice, consider the pair $\bar{p},p$, with $V_p\in p^{[0]}$ as described.
Define $t_p^{\bar{p}}$ to be a minimal spanning tree of $V_p$ on internal edges of $p$,
so there is no other spanning tree of  $V_p$ in $p$ which contains fewer edges.

Let $x_p\in\bar{p}^{[0]}$ be the vertex into which $p$ is inserted. For any $p,\bar{p}$ decorations of $T$
with $\bar{p}$ above $p$,  and $T\in C(\Gamma)$,
there are $q,\bar{q}$ in $T_f=C(\Gamma_{f})$ such that the minimal spanning tree $t_q^{\bar{q}}\subseteq t_p^{\bar{p}}$ is a path 
from $v_p$ to $w_p$ \footnote{In theories with derivative couplings, where vertices might have negative weight in powercounting,
one must (and can) demand that $\mathrm{sd}(t^{\bar{q}}_q)>0$.}. Any
not necessarily proper subset of cycles $\mathrm{cyc}\subset \bar{p}$ containing $x_p$ has 
$\mathrm{sd}(\mathrm{cyc})\leq 0$, 
as $\bar{p}$ is primitive by construction. As $\mathrm{sd}(\mathrm{cyc}-x_p\cup t_q^{\bar{q}})<
\mathrm{sd}(\mathrm{cyc})$, i) follows.\\ ii) is obvious.
\end{proof}
Two examples might be helpful.
\begin{example}
The graph
\begin{equation}
\gamma=\wf\subset \wtwf=\Gamma
\end{equation}
say is connected through edges $4,5,6$ to the graph
\begin{equation}
\Gamma/\gamma=\wt.
\end{equation}
Also, at the vertex of edges $9,10,13$, an external momentum enters $\Gamma$. So here, the set $V_\gamma$ consists of four vertices, three vertices in $V_\gamma^{\Gamma/\gamma}$ and a single vertex in $V_\gamma^{\mathrm{ext}}$. The graphs $\gamma,\Gamma/\gamma$ form a pair $p,\bar{p}$. Let $x_p$ 
be the vertex $4,5,6$ into which $\gamma$ is inserted into $\Gamma/\gamma$.

The edges $11,12,13,14$ form a minimal spanning tree for the four vertices in $V\gamma$.
In 
\begin{equation}
\Gamma_\gamma=\wtwft
\end{equation}
the connectors 4,5,6 have changed their endpoints. The divergent subgraph 
\begin{equation}
\wfb=\gamma^2\subset\Gamma_\gamma
\end{equation}  
is a 1-scale graph (in particular the external momentum at the vertex of edges $9,10,13$ is sent to zero), and $\Gamma/\gamma=\Gamma_\gamma/\gamma^2$.
$V_{\gamma^2}$ consists of two vertices, with spanning tree the path $t$ given by edges $12,14$.
Any 1PI subgraph of $\Gamma/\gamma$ is superficially convergent (if it is a proper subgraph) or log divergent
(if it is $\Gamma/\gamma$ itself) as $\Gamma/\gamma$ is free of sub-divergences.
In $\Gamma_\gamma$, such a sub-graph has the vertex $x_p$ replaced by the path $t$. Its powercounting improves by the weight of the edges in $t$.
So by construction, for any given forest $f$, we have a graph $\Gamma_f$ where the subgraphs indexed by the forest
 are 1-scale graphs. 
\end{example}
This ends the first example. Lets consider another one.
\begin{example} 
\begin{equation}
\Gamma=\oldc.
\end{equation}
We know $|C(\Gamma)|=2$. It has two maximal complete forests
$f_4,f_5$. 
A possible choice of 1-scale structures for them is
\begin{eqnarray}
\Gamma_{f_4} & = & \noldcl,(m_1=m_2=m_3=m_4=0)\\
\Gamma_{f_5} & = & \noldcr, (m_5=m_6=m_3=m_4=0).
\end{eqnarray}
We then have $\Gamma_{f_0}=\Gamma$, $\Gamma_{f_1}=\Gamma_{m_3=m_4=0}$, $\Gamma_{f_2}=\Gamma_{{f_4}}$,
$\Gamma_{f_3}=\Gamma_{{f_5}}$, where we refer to the forests in Eqs.(\ref{f0}-\ref{f5}). 
\end{example}

We have now constructed graphs such that all divergent sub-graphs depend only on a single-scale.
Still, the set $\{v\in \Gamma_f^{[0]}|q(v)\not= 0\}$ might have cardinality $>2$.

But then, we can partition this set into two non-empty subsets $X,Y$ say and choose 
a vertex $x\in X, y\in Y$ for each. We then set $q(x):=\sum_{v\in X}q(v)$ and
$q(y):=\sum_{v\in Y}q(v)$, and all other $q(v)=0$. This gives a new graph $\Gamma_f^2$
which itself is a 1-scale graph.

Define ${\phi_2}_{\Gamma_f}=
\varphi_{\Gamma_f^2}$. Note that we use the massless $\varphi$ on 1-scale graphs.
Define furthermore
\begin{equation}
x_2^{\Gamma_f}:= \frac{{\phi}_{\Gamma_f}(\Theta)-{\varphi_2}_{\Gamma_f}}{{\varphi_2}_{\Gamma_f}}.
\end{equation}
Note that $x_2^{\Gamma_f}$ is independent of $S,S_0,\Theta_0$, by construction:
$x_2^{\Gamma_f}=x_2^{\Gamma_f}(\Theta)$.
We have
\begin{lem}
 \begin{equation}
\mathbf{\Phi}_{\mathrm{fin}}(\Gamma):=\sum^\emptyset_f(-1)^{|f|}\int_{\mathbb{P}^{n}(\mathbb{R}_+)}M(x_2^{\Gamma_f})  
 \end{equation}
converges. Here, the notation is defined similarly to  Eqs.(\ref{eqform},\ref{finalresult}).
\end{lem}
\begin{proof}
By construction, in each sector, singularities drop out, reflected by the scale independence of $x_2^{\Gamma_f}$,
and we use the same forest sum as in the proof of convergence for $\Phi^R_{\mathrm{1-s}}$ below.
Compared to renormalization by counterterms, where the remainder terms $\bar{R}^\Gamma_f$ 
are replaced by $\phi_f(\Theta_0)\psi_{\Gamma/f}$, we now have different remainder terms $\bar{R}^{\Gamma_f}_f$ 
for them. This is  due to the fact that $\phi_{\Gamma/f}=\phi_{\Gamma^f/f}$. Hence, along a forest $f$,
$\phi_\Gamma=\phi_{\Gamma/f}\psi_f+\bar{R}^\Gamma_f$ and $\phi_{\Gamma_f}=\phi_{\Gamma/f}\psi_f+\bar{R}^{\Gamma_f}_f$. 
\end{proof}
\section{Double iteration}
Let us now define
\begin{equation}
 x^\Gamma_{g,f}=\frac{\frac{S}{S_0}\phi_{\Gamma_f/g}(\Theta)\psi_g
+\phi_g(\Theta_0)\psi_{\Gamma_f/g}}{\phi_{\Gamma_f/g}(\Theta_0)\psi_g+\phi_g(\Theta_0)\psi_{\Gamma_f/g}}
\end{equation}
and we work with $M(x^\Gamma_{f,g})$.
We then have
\begin{equation}
 \mathbf{\Phi}^R(\Gamma)=
\int_{\mathbb{P}^{n}(\mathbb{R}_+)}\sum^\emptyset_g(-1)^{|g|}M(x^\Gamma_{\emptyset,g}),
\end{equation}
\begin{equation}\label{doublefin}
 \mathbf{\Phi}_{\mathrm{fin}}(\Gamma)=
\int_{\mathbb{P}^{n}(\mathbb{R}_+)}\sum^\emptyset_f(-1)^{|f|}M(x^\Gamma_{f,\emptyset}).
\end{equation}
Finally, we define
\begin{equation}\label{forestsum}
 \mathbf{\Phi}^R_f(\Gamma)=
\int_{\mathbb{P}^{n}(\mathbb{R}_+)}\sum^\emptyset_g(-1)^{|g|}M(x^\Gamma_{f,g}).
\end{equation}
Then, first, (due to $\emptyset=\sum_j \bar{\sigma}_j$)
\begin{equation}
 \Phi^R(\Gamma)=\sum_{j=1}^\infty \Phi^R_{\bar{\sigma}_j}(\Gamma),
\end{equation}
which is a terminating sum on the rhs due to the finite co-radical filtration $\mathrm{cor}(\Gamma)$ of every graph $\Gamma$.

Furthermore, we have our decomposition:
\begin{thm}
 $$\mathbf{\Phi}^R_{\bar{\sigma}_j}(\Gamma-\Gamma_2)=\sum_{k=0}^{j-1}c_k(\Theta,\Theta_0)L^k,$$ 
 $$\mathbf{\Phi}^R_{\bar{\sigma}_j}(\Gamma_2)=\sum_{k=1}^{j}d_k L^k.$$ 
\end{thm}
In particular, $\mathbf{\Phi}^R_{\mathrm{1-s}}(\Gamma):=\sum_{j=1}^\infty \mathbf{\Phi}^R_{\bar{\sigma}_j}(\Gamma_2)$.

This can be written in Sweedler's notation as (see Eq.(\ref{tradfor}))
\begin{eqnarray}\mathbf{\Phi}^R_{\mathrm{1-s}}(\Gamma) & = & \sum \mathbf{\Phi}^{-1}(\Gamma^{\prime}_2)(m_e=0)(S_0)\mathbf{\Phi}(\Gamma^{\prime\prime}_2)(S)(m_e=0),\\  & \equiv & \mathbf{\Phi_2}^{-1}(\Gamma^\prime)
\mathbf{\Phi_2}(\Gamma^{\prime\prime}),\end{eqnarray}
where $\mathbf{\Phi_2}(\Gamma)=\mathbf{\Phi}(\Gamma_2)$ and 1-scale graphs $\Gamma_2$  are assumed massless.
We use unrenormalized Feynman rules in the affine representation and  1-scale subgraphs which are  evaluated to counterterms 
$\mathbf{\Phi}^{-1}(\Gamma^{\prime}_2)(m_e=0)(S_0)$.
The sum terminates at $\mathrm{cor}(\Gamma)$.

By construction, 
\begin{equation}
\mathbf{\Phi}_{\mathrm{fin}}(\Gamma)(\Theta)=\sum\mathbf{\Phi_2}(\Gamma^\prime)^{-1}\mathbf{\Phi}(\Gamma^{\prime\prime}).
\end{equation}
Then, $\mathbf{\Phi}_{\mathrm{1-s}}(\Gamma)=\mathbf{\Phi}^{-1}(S_0)\star \mathbf{\Phi}(S)$, where it is understood that on the rhs of the equation  
sub- and co-graphs $\Gamma^{\prime\prime},\Gamma^\prime$ of $\Gamma$  are projected to 1-scale graphs $\Gamma^{\prime\prime}_2,\Gamma^\prime_2$.
We finally have our decomposition
\begin{equation}
\mathbf{\Phi}^R\equiv S^\mathbf{\Phi}_R\star\mathbf{\Phi}=\underbrace{({\mathbf{\Phi_2}}^{-1}\star\mathbf{\Phi}(\Theta_0))^{-1}}_{\mathbf{\Phi}_{\mathrm{fin}}^{-1}(\Theta_0)}
\star\mathbf{\Phi}^R_{\mathrm{1-s}}\star\underbrace{({\mathbf{\Phi_2}}^{-1}\star\mathbf{\Phi}(\Theta))}_{\phi_{\mathrm{fin}}(\Theta)}.
\end{equation}

\section{Examples} Next, we work out examples to acquaint the reader with our approach. These examples consider the decomposition into angles and scales. It is not our aim to do the final integrations, but to exhibit the structure
of properly renormalized Feynman integrands, suitably decomposed as announced. In the following sections, we exclusively
explore properties of $G_{\mathrm{1-s}}$.
\subsection{Overall log}
We start with overall logarithmic divergent graphs.
\subsubsection{Primitive log}
Consider 
\begin{equation}
\Gamma=\wt,
\end{equation}
and
\begin{equation}
\Gamma^2=\wtb.
\end{equation}
The graph $\Gamma$ is the one we are interested in, the 1-scale graph $\Gamma^2$ is supposed to have only two vertices at which $q(v)\not=0$, and all internal masses in it are set to zero.

Trivially, $|C(\Gamma)|=1$, with $C(\Gamma)=\{T\}$, and $T=\bullet_\Gamma$ primitive in the Hopf algebra of trees,
and $\Gamma$ is primitive in the Hopf algebra of graphs,
$\Delta(\Gamma)=\Gamma\otimes\One+\One\otimes\Gamma$.

For the kinematics, we let $S=p_1^2+p_2^2+p_3^2+2p_1\cdot p_2+2p_2\cdot p_3+2p_3\cdot p_1$ (which defines the variable angles $\Theta^{ij}=p_i\cdot p_j/S,\Theta^e=m_e^2/S$)
and subtract symmetrically say at $S_0, \Theta_0^{ij}=\frac{1}{3}(4\delta_{ij}-1)$ and $\Theta_0^e=m_e^2/S_0$,
which specifies the fixed angles $\Theta_0$.

The renormalized form is
\begin{equation}
\Phi^R_\Gamma=\frac{\ln\frac{\frac{S}{S_0}\phi_\Gamma(\Theta)}{\phi_\Gamma(\Theta_0)}}{\psi_\Gamma^2}\Omega_\Gamma.
\end{equation}
To find the desired decomposition, we use
\begin{equation}
\Delta^2(\Gamma)=\Gamma\otimes\One\otimes\One+\One\otimes\Gamma\otimes\One+\One\otimes\One\otimes\Gamma.
\end{equation}
We then have
\begin{equation}
\Phi^R_\Gamma=\Phi^{-1}_{\mathrm{fin}}(\Theta_0)(\Gamma)+\Phi^R_{\mathrm{1-s}}(S/S_0)(\Gamma)+
\Phi_{\mathrm{fin}}(\Theta)(\Gamma).
\end{equation}
We have
\begin{equation}
\Phi^{-1}_{\mathrm{fin}}(\Theta_0)(\Gamma)=-\frac{\ln\frac{\phi_\Gamma(\Theta_0)}{\psi_{{\Gamma^2}^\bullet}}}{\psi_\Gamma^2}\Omega_\Gamma,
\end{equation}
\begin{equation}
\Phi^R_{\mathrm{1-s}}(S/S_0)(\Gamma)=\frac{\ln\frac{S}{S_0}}{\psi_{\Gamma^2}^2}\Omega_\Gamma,
\end{equation}
which integrates to the renormalized value $\mathbf{\Phi}^R_{\mathrm{1-s}}(S/S_0)(\Gamma)=6\zeta(3)\ln\frac{S}{S_0}$.
Finally, 
\begin{equation}
\Phi_{\mathrm{fin}}(\Theta)(\Gamma)=\frac{\ln\frac{\phi_\Gamma(\Theta)}{\psi_{{\Gamma^2}^\bullet}}}{\psi_\Gamma^2}\Omega_\Gamma.
\end{equation}
These integrands indeed all converge, which is synonymous for us to say that they can be integrated against $\mathbb{P}^{|\Gamma^{[1]}|-1}(\mathbb{R}_+)$.   
\subsubsection{Log with log sub}
Consider the graph
\begin{equation}
\Gamma=\wtwf.
\end{equation}
It has a subgraph 
\begin{equation}
\gamma=\wf
\end{equation}
and a co-graph
\begin{equation}
\Gamma/\gamma=\wt.
\end{equation}
It suffices to consider a small Hopf algebra (co-)generated by 
$\One,\Gamma,\gamma,\Gamma/\gamma$, with the only non-trivial co-product
\begin{equation}
\Delta(\Gamma)=\Gamma\otimes\One+\One\otimes\Gamma+\gamma\otimes\Gamma/\gamma.
\end{equation}
Obviously, $|C(\Gamma)|=1$. We have $C(\Gamma)=\{T\}$, with 
\begin{equation}
T=\loglog,
\end{equation}
with $T$ just a tree on two vertices\footnote{In light of \cite{BergbKr}, note that we can write
\begin{equation}
T=B_+^{\Gamma/\gamma}(\bullet_\gamma), 
\end{equation}
and $B_+^{\Gamma/\gamma}$ a Hochschild closed map for this Hopf algebra:
\begin{equation}
\rho(\Gamma)=T, \Delta_T(T)=T\otimes\One+\One\otimes T+\bullet_\gamma\otimes\bullet_{\Gamma/\gamma}=\Delta_T B_+^{\Gamma/\gamma}(\bullet_\gamma),
\end{equation}
with $\Delta_T B_+^{\Gamma/\gamma}=B_+^{\Gamma/\gamma}\otimes \One+(\mathrm{id}\otimes B_+^{\Gamma/\gamma})\Delta_T$.}.

Forests are $f_0=\emptyset$ and $f_1=\gamma$. 

The renormalized integrand is then
\begin{equation}
\Phi^R_\Gamma(S/S_0,\Theta,\Theta_0)=\left\{ 
\frac{\ln\frac{\frac{S}{S_0}\phi_\Gamma(\Theta)}{\phi_\Gamma(\Theta_0)}}{\psi_\Gamma^2}-
\frac{
\ln 
\frac{
 \frac{S}{S_0}\phi_{\Gamma/\gamma}(\Theta)\psi_\gamma+\phi_\gamma(\Theta_0)\psi_{\Gamma/\gamma}
 }
{ \phi_{\Gamma/\gamma}(\Theta_0)\psi_\gamma+\phi_\gamma(\Theta_0)\psi_{\Gamma/\gamma}}
}
{\psi_{\Gamma/\gamma}^2\psi^2_\gamma} 
\right\}\Omega_\Gamma.
\end{equation}
One immediately checks that this is integrable in the edge variables against $\mathbb{P}^{13}(\mathbb{R}_+)$,
using nothing more than our cherished remainder properties of the graph polynomials Eqs.(\ref{cher1},\ref{cher2}).

For the decomposition as a twisted conjugation, we stick with our conventions that a superscript ${}^2$ indicated that a graph 
is 1-scale with respect to its external momenta flow ($q(v)\not= 0$ only for two vertices, all internal masses zero), while a subscript ${}_f$ means that the graphs in the forest $f$ are made 1-scale as subgraphs, also with zero internal masses in them.
Then, $\Gamma_\emptyset=\Gamma$ (the forest which is made 1-scale is empty), and we choose
\begin{equation}
\Gamma^2\equiv \Gamma^2_\emptyset=\wtwfb,
\end{equation}
\begin{equation}
\Gamma_\gamma=\wtwft,
\end{equation}
and
\begin{equation}
\Gamma^2_\gamma=\wtwftb.
\end{equation}
This gives also a new sub-graph $\gamma^2$ of $\Gamma_\gamma$,
\begin{equation}
\gamma^2=\wfb,
\end{equation}
and a new co-graph of $\Gamma^2_\emptyset$,
\begin{equation}
\Gamma^2_\gamma/\gamma^2=\Gamma^2_\emptyset/\gamma=\wtb.
\end{equation}
We have by definition
\begin{equation}
\phi_{\Gamma_\emptyset}=\phi_\Gamma=S\phi_{\Gamma}(\Theta),
\end{equation}
\begin{equation}
\phi_{\Gamma_\emptyset^2}=S\psi_{{\Gamma_\emptyset^2}^\bullet},
\end{equation}
and
\begin{equation}
\phi(\Gamma_\gamma)=S\phi_{\Gamma/\gamma}(\Theta)\psi_\gamma+S\psi_{\gamma^\bullet}\psi_{\Gamma/\gamma}.
\end{equation}
Also,
\begin{equation}
\phi(\Gamma_\gamma^2)=S\psi_{{\Gamma/\gamma}^\bullet}\psi_\gamma+S\psi_{\gamma^\bullet}\psi_{\Gamma/\gamma}.
\end{equation}
It follows that
\begin{equation}
1+x_2^{\Gamma_\emptyset}=\frac{\phi_{\Gamma_\emptyset}(\Theta)}{\psi_{{\Gamma^2_\emptyset}^\bullet}}
\end{equation}
and
\begin{equation}
1+x_2^{\Gamma_\gamma}=\frac{\phi_{\Gamma_\gamma}(\Theta)}{\psi_{{\Gamma^2_\gamma/\gamma^2}^\bullet}\psi_{\gamma}+\psi_{{\gamma^2}^\bullet}\psi_{\Gamma_\gamma^2/\gamma^2}}.
\end{equation}

We also need 
\begin{equation}
{\Gamma^2_\gamma/\gamma^2}^\bullet=\wtbbull,
\end{equation}
and 
\begin{equation}
G=\wtbups.
\end{equation}
We have $\Gamma^2_\gamma/\gamma^2=\Gamma^2_\emptyset/\gamma$ and 
$\Gamma/\gamma=\Gamma_\gamma/\gamma^2$. Also, $\psi_\Gamma=\psi_{\Gamma^2_\emptyset},\psi_{\Gamma_\gamma}=\psi_{\Gamma_\gamma^2}$,
and $\psi_\gamma=\psi_{\gamma^2}$.

To evaluate $\Phi_{\mathrm{fin}}^{-1}(\Theta_0)\star\Phi^R_{\mathrm{1-s}}(S/S_0)\star\Phi_{\mathrm{fin}}(\Theta)$,
we need
\begin{equation}
\Delta^2(\Gamma)=\Gamma\otimes\One\otimes\One+\One\otimes\Gamma\otimes\One+\One\otimes\One\otimes\Gamma+\One\otimes\gamma\otimes\Gamma/\gamma
+\gamma\otimes\One\otimes\Gamma/\gamma+\gamma\otimes\Gamma/\gamma\otimes\One.
\end{equation}
We find
\begin{eqnarray}\label{exa1}
\Phi^R_\Gamma & = & \Phi_{\mathrm{fin}}(\Theta)(\Gamma)+\Phi_{\mathrm{fin}}^{-1}(\Theta_0)(\Gamma)+\Phi^R_{\mathrm{1-s}}(\gamma)
\Phi_{\mathrm{fin}}(\Theta)(\Gamma/\gamma)\nonumber\\ & & +\Phi_{\mathrm{1-s}}^R(\Gamma)+\Phi_{\mathrm{fin}}^{-1}(\Theta_0)(\gamma)\Phi_{\mathrm{fin}}(\Theta)(\Gamma/\gamma)+\Phi_{\mathrm{fin}}^{-1}(\Theta_0)(\gamma)\Phi^R_{\mathrm{1-s}}(\Gamma/\gamma).
\end{eqnarray}
Note that all the terms on the rhs can be expressed using forest sums as in Eq.(\ref{forestsum}):
\begin{eqnarray}
\Phi^R_\Gamma & = & 
+
\underbrace{\Phi^R_{\bar{\sigma}_1}(\Gamma-\Gamma^2)}_{\mathcal{O}(1)}+\underbrace{\Phi^R_{\bar{\sigma}_2}(\Gamma-\Gamma^2)}_{\mathcal{O}(L)}\label{exa1a}
\\
 & & +
\underbrace{\Phi^R_{\bar{\sigma}_1}(\Gamma^2)}_{\mathcal{O}(L)}+\underbrace{\Phi^R_{\bar{\sigma}_2}(\Gamma^2)}_{=\Phi^R_{\mathrm{1-s}}(\Gamma),\mathcal{O}(L^2)},\label{exa1b}
\end{eqnarray}
indicating orders in $L=\ln S/S_0$.
We have in particular (we invite the reader to confirm convergence of these expressions again directly from the remainder properties Eqs.(\ref{cher1},\ref{cher2}))
\begin{equation}
\Phi_{\mathrm{fin}}(\Theta)(\Gamma)=\left\{
\frac{\ln
\frac{\phi_{\Gamma}(\Theta)}{\psi_{{\Gamma^2_\emptyset}^\bullet}}
}{\psi^2_{\Gamma^2_\emptyset}}
-
\frac{\ln
\frac{\phi_{\Gamma_\gamma}(\Theta)}{\psi_{{\Gamma^2_\gamma}^\bullet}}
}{\psi^2_{\Gamma_\gamma^2}}
\right\}\Omega_\Gamma.
\end{equation}
Furthermore
\begin{equation}
\Phi^R_{\mathrm{1-s}}(S/S_0)(\Gamma)=\left\{
\frac{
\ln
\frac{S}{S_0}
}{\psi^2_{\Gamma^2}}
-
\frac{
\ln
\frac{\overbrace{ \frac{S}{S_0}\psi_{{\Gamma^2/\gamma^2}^\bullet}\psi_{\gamma^2}+\psi_{{\gamma^2}^\bullet}\psi_{\Gamma^2/\gamma^2}}^{\Upsilon_{\gamma^2;\Gamma^2}\left(\frac{S}{S_0}\right)}}{\psi_{{\Gamma^2/\gamma^2}^\bullet}\psi_{\gamma^2}+\psi_{{\gamma^2}^\bullet}\psi_{\Gamma^2/\gamma^2}}
}{\psi^2_{\Gamma^2/\gamma^2}\psi^2_{\gamma_2}}
\right\} \Omega_\Gamma
\end{equation}
The term $\sim\ln^2(S/S_0)$ in $\Phi^R_\Gamma$, generated by $\Phi_{\mathrm{1-s}}$ only, 
is obviously $60\zeta(3)\zeta(5)$ by our considerations, 
while (\ref{exa1}) reveals three sources for a term linear in $L=\ln(S/S_0)$: First, we have 
\begin{equation}
\int\frac{\psi_{{\Gamma^2_\gamma/\gamma^2}^\bullet}}{\psi^2_{\Gamma^2_\gamma/\gamma^2}\psi_{\gamma^2}\Upsilon_{\gamma^2;\Gamma^2_\gamma}(1)}\Omega_\Gamma,
\end{equation}
(patently inaccessible by standard methods) but note that 
$\Upsilon_{\gamma^2;\Gamma^2_\gamma}(1)=\psi_G$, so that this expression is reduced to the study of graph hypersurfaces, started in \cite{BrCMP,BrFeyn}. 

For the other two terms linear in $L$, they give
$\Phi_{\mathrm{fin}}^{-1}(\Theta_0)(\gamma)\Phi^R_{\mathrm{1-s}}(\Gamma/\gamma)$
(which is $20\zeta(5)$ times an unknown function of angles $\Theta_0$, also easily identified in the second term on the right in Eq.(\ref{exa1a})),
as well as 
$\Phi^R_{\mathrm{1-s}}(\gamma)
\Phi_{\mathrm{fin}}(\Theta)(\Gamma/\gamma)$ (which is $6\zeta(3)$ times an unknown function of angles $\Theta$,
see also the first term on the right in Eq.(\ref{exa1b})).

The constant term in $L$ is again an unknown  function of angles $\Theta,\Theta_0$, which equals 
$\Phi^R_\Gamma(1,\Theta,\Theta_0)$.
\subsubsection{Log with quadratic sub}
This is an interesting example as it will involve an overall convergent contribution $l_{-1}^{[0]}$.
Kinematics is as before.
Let
\begin{equation}
\Gamma=\dcqs,
\end{equation}
and
\begin{equation}
\bar{\Gamma}=\dcqssq.
\end{equation}
The subgraph is 
\begin{equation}
\gamma=\subqul.
\end{equation}
The cographs are
\begin{equation}
\Gamma/\gamma=\codcsq,
\end{equation}
\begin{equation}
\bar{\Gamma}/\gamma=\codc.
\end{equation}
The Hopf algebra structure is simply 
\begin{equation}
\Delta(\Gamma)=\Gamma\otimes\One+\One\otimes\Gamma+\gamma\otimes\Gamma/\gamma,
\end{equation}
as co-graphs which are tadpoles vanish.

We have
\begin{eqnarray}
 & & \Phi^R_\Gamma  =  
\frac{\psi_{\gamma^\bullet}}{\psi_\gamma}
\left\{
\frac{l_0^{[0]}(x^{\bar{\Gamma}}_\emptyset)}{\psi^2_{\bar{\Gamma}}}
-\frac{l_0^{[0]}(x^{\bar{\Gamma}}_\gamma))}{\psi^2_{\bar{\Gamma/\gamma}}\psi^2_\gamma}
\right\}
\Omega_{\bar{\Gamma}}\\ 
&  & +
m_\gamma^2\frac{\bar{\psi}_{\gamma^\bullet}}{\psi_\gamma}
\left\{
\left(\frac{S_0\phi_\Gamma(\Theta_0)}{\psi_\Gamma}\right)^{-1}
\frac{l_{-1}^{[0]}(x^{{\Gamma}}_\emptyset)}{\psi^2_{{\Gamma}}}\right. \nonumber\\ & & \left.
-\left(\frac{S_0\phi_{\Gamma/\gamma}(\Theta_0)\psi_\gamma+S_0\bar{\psi}_{\gamma^\bullet}\psi_{\Gamma/\gamma}}{\psi_{\Gamma/\gamma}}\right)^{-1}
\frac{l_{-1}^{[0]}(x^{{\Gamma}}_\gamma))}{\psi^2_{{\Gamma/\gamma}}\psi^2_\gamma}
\right\}\Omega_{{\Gamma}}\nonumber\\ 
&  & -
m_\gamma^2\frac{\bar{\psi}_{\gamma^\bullet}}{\psi_\gamma}\frac{{\psi}_{\gamma^\bullet}}{\psi_\gamma}
\left\{
\left(
\frac{S_0\phi_{\bar{\Gamma}}(\Theta_0)}{\psi_{\bar{\Gamma}}}\right)^{-1}\frac{l_{-1}^{[0]}(x^{\bar{\Gamma}}_\emptyset)}{\psi^2_{\bar{\Gamma}}}
\right. \nonumber\\ & & \left.
-\left(\frac{S_0\phi_{\bar{\Gamma}/\gamma}(\Theta_0)\psi_\gamma+S_0\bar{\psi}_{\gamma^\bullet}\psi_{\bar{\Gamma}/\gamma}}{\psi_{\bar{\Gamma}/\gamma}}\right)^{-1}
\frac{l_{-1}^{[0]}(x^{\bar{\Gamma}}_\gamma))}{\psi^2_{\bar{\Gamma/\gamma}}\psi^2_\gamma}
\right\}\Omega_{\bar{\Gamma}}.\nonumber
\end{eqnarray}
The first line above and the coefficients of $m_\gamma^2$ in the two following lines
allow for a separate decomposition in $G_{\mathrm{1-s}},G_{\mathrm{fin}}$ as expected: the first line
depends on scales as well as angles and decomposes as before, the two following lines only depend on angles, taking all definitions into account.
Indeed, this must be so as we have an overall dimensionless quantity, and those two lines come from a superficial convergent contribution as they factorize $m_\gamma^2$, and hence the angle $m_\gamma^2/S$, multiplying it all out.
Here,
\begin{equation}
l_{-1}^{[0]}(x^\Gamma_\emptyset)=\frac{1}{1+\frac{\frac{S}{S_0}(\phi_\Gamma(\Theta)-\phi_{\Gamma}(\Theta_0))}{\phi_\Gamma(\Theta_0)}}-1,
\end{equation}
and
\begin{equation}
l_{-1}^{[0]}(x^\Gamma_\gamma)=\frac{1}{1+\frac{\frac{S}{S_0}(\phi_{\Gamma/\gamma}(\Theta)\psi_\gamma-\phi_{\Gamma/\gamma}(\Theta_0))\psi_\gamma}{\phi_{\Gamma/\gamma}(\Theta_0)\psi_\gamma+\bar{\psi}_{\gamma^\bullet}\psi_{\Gamma/\gamma}}}-1,
\end{equation}
similarly for $l_{-1}^{[0]}(x^{\bar{\Gamma}}_\emptyset)$ and $l_{-1}^{[0]}(x^{\bar{\Gamma}}_\gamma)$.

\subsection{Overall quadratic}
Next, we turn to overall quadratic amplitudes.
\subsubsection{Primitive overall quadratic}
We look at the graph 
\begin{equation}
\Gamma=\oltpsq.
\end{equation} $C(\Gamma)=\{\bullet_\Gamma\}$.
Kinematics are given by $S=q^2, S_0=m^2, \Theta_{q^2}=(1-\frac{m^2}{q^2}),\Theta_{q^2}^0=(1-\frac{q^2}{m^2}),\Theta_e=m_e^2/q^2,\Theta^0_e=m_e^2/m^2$.
We find immediately
\begin{eqnarray}
\Phi^R_\Gamma & = &  (q^2-m^2)\frac{A_2A_3A_4}{\psi_\Gamma}
\frac{\ln(1+x^\Gamma_\emptyset)\Omega_\Gamma}{\psi_\Gamma^2}\nonumber\\
 & & +m^2\frac{A_2A_3A_4+(A_2\Theta_2^0+A_3\Theta_3^0+A_4\Theta_4^0)\psi_\Gamma}{\psi_\Gamma}
\frac{\ln^{[1]}(1+x^\Gamma_\emptyset)\Omega_\Gamma}{\psi_\Gamma^2}.
\end{eqnarray}
Here, 
\begin{equation}
x^\Gamma_\emptyset  =  \frac{\frac{q^2}{m^2}\Theta_{q^2}A_2A_3A_4}{ (A_2A_3A_4+(A_2\Theta_2^0+A_3\Theta_3^0+A_4\Theta_4^0)\psi_\Gamma)}.
\end{equation}
This makes it evident that we fulfil the renormalization conditions.

We set $\Phi^R_\Gamma=:(q^2-m^2)\Phi^{(q^2-m^2);R}_\Gamma+m^2\Phi^{m^2;R}_\Gamma+m_2^2\Phi^{m_2^2;R}_\Gamma+m_3^2\Phi^{m_3^2;R}_\Gamma
+m_4^2\Phi^{m_4^2;R}_\Gamma$ and can read off the coefficient functions $\Phi_\Gamma^{x;R}$ immediately. 
Note that they are all dimensionless quantities.

The decomposition into angles and scales happens now for these coefficient functions.
\begin{eqnarray}
\Phi^{q^2-m^2;R}_{\mathrm{1-s}}(\Gamma) & = & \frac{A_2A_3A_4}{\psi_\Gamma}\ln\left( \frac{q^2}{m^2}\right)\frac{\Omega_\Gamma}{\psi^2_\Gamma},\\
\Phi^{m^2;R}_{\mathrm{1-s}}(\Gamma) & = & \frac{A_2A_3A_4}{\psi_\Gamma}\ln\left( \frac{q^2}{m^2}\right)\frac{\Omega_\Gamma}{\psi^2_\Gamma},\\
\Phi^{m_2^2;R}_{\mathrm{1-s}}(\Gamma) & = & \frac{A_2}{1}\ln\left( \frac{q^2}{m^2}\right)\frac{\Omega_\Gamma}{\psi^2_\Gamma},\\
\Phi^{m_3^2;R}_{\mathrm{1-s}}(\Gamma) & = & \frac{A_3}{1}\ln\left( \frac{q^2}{m^2}\right)\frac{\Omega_\Gamma}{\psi^2_\Gamma},\\
\Phi^{m_4^2;R}_{\mathrm{1-s}}(\Gamma) & = & \frac{A_4}{1}\ln\left( \frac{q^2}{m^2}\right)\frac{\Omega_\Gamma}{\psi^2_\Gamma}.
\end{eqnarray}
For the angles:
\begin{eqnarray}
\Phi^{q^2-m^2;R}_{\mathrm{fin}}(\Theta)(\Gamma) & = & \frac{A_2A_3A_4}{\psi_\Gamma}\ln\left( \frac{\varphi_\Gamma+\psi_\Gamma\sum_{j=2}^4 \Theta_j A_j}{\phi_\Gamma}\right)\frac{\Omega_\Gamma}{\psi^2_\Gamma},\\
\Phi^{m^2;R}_{\mathrm{fin}}(\Theta)(\Gamma) & = & \frac{A_2A_3A_4}{\psi_\Gamma}\ln\left( \frac{\varphi_\Gamma+\psi_\Gamma\sum_{j=2}^4 \Theta_j A_j}{\phi_\Gamma}\right)\frac{\Omega_\Gamma}{\psi^2_\Gamma},\\
\Phi^{m_2^2;R}_{\mathrm{fin}}(\Theta)(\Gamma) & = & \frac{A_2}{1}\ln\left( \frac{\phi_\Gamma+\psi_\Gamma\sum_{j=2}^4 \Theta_j A_j}{\varphi_\Gamma}\right)\frac{\Omega_\Gamma}{\psi^2_\Gamma},\\
\Phi^{m_3^2;R}_{\mathrm{fin}}(\Theta)(\Gamma) & = & \frac{A_3}{1}\ln\left( \frac{\phi_\Gamma+\psi_\Gamma\sum_{j=2}^4 \Theta_j A_j}{\varphi_\Gamma}\right)\frac{\Omega_\Gamma}{\psi^2_\Gamma},\\
\Phi^{m_4^2;R}_{\mathrm{fin}}(\Theta)(\Gamma) & = & \frac{A_4}{1}\ln\left( \frac{\phi_\Gamma+\psi_\Gamma\sum_{j=2}^4 \Theta_j A_j}{\varphi_\Gamma}\right)\frac{\Omega_\Gamma}{\psi^2_\Gamma}.
\end{eqnarray}
and
\begin{eqnarray*}
\Phi^{q^2-m^2;R}_{\mathrm{fin}}(\Theta^0)(\Gamma) & = & \frac{A_2A_3A_4}{\psi_\Gamma}\ln\left( \frac{\varphi_\Gamma+\psi_\Gamma\sum_{j=2}^4 \Theta_j^0 A_j}{\phi_\Gamma}\right)\frac{\Omega_\Gamma}{\psi^2_\Gamma},\\
\Phi^{m^2;R}_{\mathrm{fin}}(\Theta^0)(\Gamma) & = & \frac{A_2A_3A_4}{\psi_\Gamma}\left(\ln\left( \frac{\varphi_\Gamma+\psi_\Gamma\sum_{j=2}^4 \Theta_j^0 A_j}{\phi_\Gamma}\right)\right.\nonumber\\ & & \left.
-\frac{\Theta_{q^2}^0\varphi_\Gamma}{\varphi_\Gamma+\psi_\Gamma\sum_{j=2}^4 \Theta^0_jA_j}\right)\frac{\Omega_\Gamma}{\psi^2_\Gamma},\\
\Phi^{m_2^2;R}_{\mathrm{fin}}(\Theta^0)(\Gamma) & = & \frac{A_2}{1}\left(\ln\left( \frac{\varphi_\Gamma+\psi_\Gamma\sum_{j=2}^4 \Theta_j^0 A_j}{\phi_\Gamma}
\right)
-\frac{\Theta_{q^2}^0\phi_\Gamma}{\varphi_\Gamma+\psi_\Gamma\sum_{j=2}^4 \Theta^0_jA_j}\right)\frac{\Omega_\Gamma}{\psi^2_\Gamma},\\
\Phi^{m_3^2;R}_{\mathrm{fin}}(\Theta^0)(\Gamma) & = & \frac{A_3}{1}\left(\ln\left( \frac{\varphi_\Gamma+\psi_\Gamma\sum_{j=2}^4 \Theta_j^0 A_j}{\phi_\Gamma}\right)
-\frac{\Theta_{q^2}^0\phi_\Gamma}{\varphi_\Gamma+\psi_\Gamma\sum_{j=2}^4 \Theta^0_jA_j}\right)\frac{\Omega_\Gamma}{\psi^2_\Gamma},\\
\Phi^{m_4^2;R}_{\mathrm{fin}}(\Theta^0)(\Gamma) & = & \frac{A_4}{1}\left(\ln\left( \frac{\varphi_\Gamma+\psi_\Gamma\sum_{j=2}^4 \Theta_j^0 A_j}{\phi_\Gamma}
\right)
-\frac{\Theta_{q^2}^0\phi_\Gamma}{\varphi_\Gamma+\psi_\Gamma\sum_{j=2}^4 \Theta^0_jA_j}\right)\frac{\Omega_\Gamma}{\psi^2_\Gamma}.
\end{eqnarray*}

\subsubsection{Overall quadratic with quadratic sub}
The graph $\Gamma$ which we want to consider is
$$ \Gamma=\tltp .$$ We also need the graph
$$ \bar{\Gamma}=\tltpsq.$$
We have the sub-graph $$\gamma=\suboltp,$$
and the co-graphs 
$$\Gamma/\gamma=\oltpdot,$$ and
$$ \bar{\Gamma}/\gamma=\oltpsq .$$
We also need 
$$\Gamma^\bullet=\tltpdt$$
and
$$ \bar{\Gamma}^\bullet=\tltpsqdt.$$
For any $X$ as needed we write
\begin{equation}
\bar{\psi}_X=\frac{1}{m^2}\phi_X(m^2,\Theta^0)=\psi_{X^\bullet}+\left(\sum_{e\in X}\Theta_e^0A_e\right)\psi_X
\end{equation}
and
\begin{equation}
\phi_X(q^2,\Theta)=q^2\left(\psi_{X^\bullet}+\left(\sum_{e\in X}\Theta_e A_e\right)\psi_X\right).
\end{equation}
We renormalize the subgraph at $m_\gamma^2$ and $\Gamma$ at $m^2$.
Then, by Thm.\ref{finalr},
\begin{eqnarray}
 & & \Phi^R_\Gamma(q^2,m^2)  =  \\ & + &
\frac{\psi_{\bar{\Gamma}^\bullet}}{\psi_{\bar{\Gamma}}}
\frac{\psi_{\gamma^\bullet}}{\psi_{\gamma}}
l_1^{[1]}(x^{\bar{\Gamma}}_\emptyset)
\frac{\Omega_{\bar{\Gamma}}}{\psi^2_{\bar{\Gamma}}}\label{ququ1}\\ & - & 
\frac{\psi_{{\bar{\Gamma}/\gamma}^\bullet}}{\psi_{{\bar{\Gamma}/\gamma}}}
\frac{\psi_{{\gamma}^\bullet}}{\psi_{{\gamma}}}l_1^{[1]}(x^{\bar{\Gamma}}_\gamma)
\frac{\Omega_{\bar{\Gamma}}}{\psi^2_{\bar{\Gamma}/\gamma}\psi^2_\gamma}\label{ququ2}\\ & + & 
m_\gamma^2 \frac{\bar{\psi}_{\gamma^\bullet}}{\psi_\gamma}l_0^{[1]}(x^\Gamma_\emptyset)
\frac{\Omega_\Gamma}{\psi_\Gamma^2}\label{ququ3}\\ & - & 
m_\gamma^2 \frac{\bar{\psi}_{\gamma^\bullet}}{\psi_\gamma}l_0^{[1]}(x^\Gamma_\gamma)
\frac{\Omega_\Gamma}{\psi_{\Gamma/\gamma}^2\psi^2_\gamma}\label{ququ4}\\ & - & 
m_\gamma^2 \frac{\bar{\psi}_{\gamma^\bullet}}{\psi_\gamma}\frac{\psi_{\gamma^\bullet}}{\psi_\gamma}l_0^{[1]}(x^{\bar{\Gamma}}_\emptyset)
\frac{\Omega_{\bar{\Gamma}}}{\psi_{\bar{\Gamma}}^2}\label{ququ5}\\ & + & 
m_\gamma^2 \frac{\bar{\psi}_{\gamma^\bullet}}{\psi_\gamma}\frac{\psi_{\gamma^\bullet}}{\psi_\gamma}l_0^{[1]}(x^{\bar{\Gamma}}_\gamma)
\frac{\Omega_{\bar{\Gamma}}}{\psi_{\bar{\Gamma}/\gamma}^2\psi^2_\gamma}\label{ququ6}.
\end{eqnarray}
Here, $\Omega_{\bar{\Gamma}}$ is the five-form $A_2dA_3\wedge\cdots\wedge dA_7\pm\cdots$, and $\Omega_\Gamma$
the six-form $A_1dA_2\wedge\cdots\wedge dA_7\pm\cdots$.
Furthermore,
\begin{eqnarray}
l_1^{[1]}(x^{\bar{\Gamma}}_\emptyset) & = & (q^2-m^2)\ln\left(
1+\frac{(q^2-m^2)\psi_{\bar{\Gamma}^\bullet}}{m^2\bar{\psi}_{\bar{\Gamma}^\bullet}}
\right)\\ & + & \frac{\bar{\psi}_{\bar{\Gamma}^\bullet}}{\psi_{\bar{\Gamma}^\bullet}}\left\{
m^2\ln\left(1+\frac{(q^2-m^2)\psi_{\bar{\Gamma}^\bullet}}{m^2\bar{\psi}_{\bar{\Gamma}^\bullet}}\right)\right.\\ & - & 
\left. m^2\frac{(q^2-m^2)\psi_{\bar{\Gamma}^\bullet}}{m^2\bar{\psi}_{\bar{\Gamma}^\bullet}}\right\}.
\end{eqnarray}
and
\begin{eqnarray}
l_1^{[1]}(x^{\bar{\Gamma}}_\gamma) & = & (q^2-m^2)\ln\left(1+\frac{(q^2-m^2)\psi_{{\bar{\Gamma}/\gamma}^\bullet}\psi_\gamma}{m^2\bar{\psi}_{{\bar{\Gamma}/\gamma}^\bullet}\psi_\gamma+\psi_{\gamma^\bullet}\psi_{\bar{\Gamma/\gamma}}}\right)\\ & + & \frac{\bar{\psi}_{{\bar{\Gamma}/\gamma}^\bullet}}{\psi_{{\bar{\Gamma}/\gamma}^\bullet}}\left\{
m^2\ln\left(1+\frac{(q^2-m^2)\psi_{{\bar{\Gamma}/\gamma}^\bullet}\psi_\gamma}{m^2\bar{\psi}_{{\bar{\Gamma}/\gamma}^\bullet}\psi_\gamma+\psi_{\gamma^\bullet}\psi_{\bar{\Gamma/\gamma}}}\right)\right.\\ & - & 
\left. m^2\frac{(q^2-m^2)\psi_{\bar{\Gamma}^\bullet\psi_\gamma}}{m^2\bar{\psi}_{\bar{\Gamma}^\bullet}\psi_\gamma+
\psi_{\gamma^\bullet}\psi_{\bar{\Gamma/\gamma}}}\right\}.
\end{eqnarray}
\begin{equation}
l_0^{[1]}(x^\Gamma_\emptyset)=\ln\left(1+\frac{(q^2-m^2)\psi_{\Gamma^\bullet}}{m^2\bar{\psi}_{\Gamma^\bullet}}\right)
-\frac{(q^2-m^2)\psi_{\Gamma^\bullet}}{m^2\bar{\psi}_{\Gamma^\bullet}},
\end{equation}
and
\begin{equation}
l_0^{[1]}(x^\Gamma_\gamma)=\ln\left(1+\frac{(q^2-m^2)\psi_{{\Gamma/\gamma}^\bullet}\psi_\gamma}{m^2\bar{\psi}_{{\Gamma/\gamma}^\bullet}\psi_\gamma+\psi_{\gamma^\bullet}\psi_{\Gamma/\gamma}}\right)
-\frac{(q^2-m^2)\psi_{{\Gamma/\gamma}^\bullet\psi_\gamma}}{m^2\bar{\psi}_{{\Gamma/\gamma}^\bullet}\psi_\gamma+\psi_{\gamma^\bullet}\psi_{\Gamma/\gamma}},
\end{equation}
and similarly for $l_0^{[1]}(x^{\bar{\Gamma}}_\emptyset)$
and $l_0^{[1]}(x^{\bar{\Gamma}}_\gamma)$.

One now immediately checks that lines (\ref{ququ1})+(\ref{ququ2}) add to a convergent expression which fulfils the renormalization conditions, and so do (\ref{ququ3})+(\ref{ququ4}) and (\ref{ququ5})+(\ref{ququ6}).

From here on, decomposition into $\Phi_{\mathrm{1-s}}$ and  $\Phi_{\mathrm{fin}}$  is lengthy but straightforward.
\section{Graph Polynomials and 1-scale graphs}

\subsection{Single-scale graphs} \begin{defn} A \emph{1-scale graph}    is a  finite disjoint union  of connected graphs, each of which has   exactly two vertices whose external legs carry non-zero momenta, and with vanishing internal masses $m_e$ for all edges $e$.   \end{defn} 

Let $G$  be a connected  1-scale graph with distinguished vertices $v_1\neq v_2$. Let 
$\Gbar$ denote the graph obtained from $G$ by adding a new edge $e$ connecting $v_1$ and $v_2$, and  let  $\Gdot$ be
the graph obtained from $G$ by identifying  $v_1$ and $v_2$. The graphs  $G$, $\Gdot$ are obtained from $\Gbar$ by deleting (contracting) the edge $e$ respectively. Thus
$$ \psi_{\Gbar} = \psi_{G} \, \alpha_e + \psi_{\Gdot}\ .$$
  If $G=\cup_{i=1}^n G_i$ has $n$ connected components, then we define $\Gdot$  to be  $\cup_{i=1}^n\Gdot_i$.

\begin{center}
\fcolorbox{white}{white}{
  \begin{picture}(330,85) (74,-90)
    \SetWidth{1.0}
    \SetColor{Black}
    \Arc(344.5,-12.333)(7.683,-86.269,266.269)
    \Text(330,-11)[lb]{{\Black{$1$}}}
    \Arc[clock](75.111,-51)(57.889,32.378,-32.378)
    \Arc(189.143,-51)(72.143,154.551,205.449)
    \Line(78,-51)(126,-83)
   \CBox(80,-52)(76,-48){Black}{Black}%%JaxoDrawID:Vertex
      \Text(65,-52)[lb]{{\Black{$v_1$}}}
    \Vertex(124,-82){3}
 \CBox(122,-18)(126,-22){Black}{Black}%%JaxoDrawID:Vertex
      \Text(120,-16)[lb]{{\Black{$v_2$}}}
    \Line(78,-51)(126,-19)
    \Arc[clock](193.857,-51)(72.143,25.449,-25.449)
    \Arc(315.062,-51)(64.062,151.059,208.941)
    \Line(211,-51)(259,-19)
    \Line(211,-51)(259,-83)
   \CBox(257,-18)(261,-22){Black}{Black}%%JaxoDrawID:Vertex
    \CBox(209,-53)(213,-49){Black}{Black}%%JaxoDrawID:Vertex
    \Vertex(259,-82){3}
    \Arc[clock](252.179,-62.099)(42.648,164.915,80.796)
    \Vertex(345,-20){3}
    \Vertex(345,-82){3}
    \Arc(369.533,-51)(39.533,128.358,231.642)
    \Arc[clock](325.235,-51)(36.765,57.48,-57.48)
    \Line(345,-20)(345,-83)
    \Text(94,-36)[lb]{{\Black{$1$}}}
    \Text(94,-60)[lb]{{\Black{$2$}}}
    \Text(120,-53)[lb]{{\Black{$3$}}}
    \Text(135,-53)[lb]{{\Black{$4$}}}
      \Text(83,-90)[lb]{\Large{\Black{$G$}}}
    \Text(235,-43)[lb]{{\Black{$1$}}}
    \Text(235,-65)[lb]{{\Black{$2$}}}
    \Text(253,-53)[lb]{{\Black{$3$}}}
    \Text(271,-53)[lb]{{\Black{$4$}}}
    \Text(222,-24)[lb]{{\Black{$e$}}}
      \Text(215,-90)[lb]{\Large{\Black{$\Gbar$}}}
    \Text(324,-53)[lb]{{\Black{$2$}}}
    \Text(348,-53)[lb]{{\Black{$3$}}}
    \Text(366,-53)[lb]{{\Black{$4$}}}
      \Text(315,-90)[lb]{\Large{\Black{$\Gdot$}}}
  \end{picture}
}
\end{center}

\begin{lem} \label{lemphi} Let $G$ be a connected, 1-scale graph with external momentum $q$. Then the second Symanzik polynomial $\phi(q,G)$ is given by 
\begin{equation} \phi(q,G) = q^2 \psi_{\Gdot}\ .  \end{equation}
\end{lem}
\begin{proof} This follows immediately from the definition of $\phi(q,G)$.\end{proof}
We henceforth write $\phi_G$ for $\psi_{\Gdot}$, since all $q$-dependencies are trivial.  Clearly,   $\deg \phi_G = \deg \psi_G + 1$. 
If $G=\cup_{i=1}^n G_i$ with $G_i$ connnected, we define
\begin{equation} \label{productconvention}
\psi_G = \prod_{i=1}^n \psi_{G_i} \qquad \hbox{ and } \qquad \phi_G = \sum_{i=1}^n \phi_{G_i} \prod_{j\neq i} \psi_{G_j}\ .
\end{equation} 
Note that this convention  is \emph{not} compatible with the  contraction-deletion relations, which require that a non-connected graph have vanishing graph polynomial. 
\begin{defn} Let $G$ be a connected 1-scale graph, and let $I\subset E(G)$ be a set of edges of $G$. Define
$\varepsilon_G(I)\in\{0,1\}$ to be $1$ if $I$ meets both distinguished vertices in $G$, and $0$ otherwise.
Equivalently, $\varepsilon_G(I)$ is $0$ if the image of $I$ in    $G^\bullet$ is isomorphic to $I$, and equal to $1$ otherwise. If $G=\cup_{i=1}^n G_i$, then we define $\varepsilon_G = \prod_{i=1}^n \varepsilon_{G_i}$.
\end{defn} 

Suppose that $G$ is a connected graph with vertices of degree at most 4, and let $\gamma \subsetneq G$ be a connected  subgraph.
Define the superficial  degree of divergence of $\gamma$ by
$$sd(\gamma)= 2h_\gamma  - E_{\gamma}\ . $$
Since $\gamma$ is  a strict subgraph, the  average valency of its edges $\alpha$ is $<4$, and hence $E_{\gamma} <2V_{\gamma}$. By Euler's formula $
E_{\gamma}-V_{\gamma}= h_{\gamma}-1$, this implies that  $sd(\gamma)\leq 1$.  We say that $\gamma$ is 
\begin{itemize}
\item \emph{convergent}  if $sd(\gamma)<0$,
\item  \emph{log-divergent}  if $sd(\gamma)=0$,  
\item \emph{quadratically divergent} if $sd(\gamma)=1$.
\end{itemize}
We say that $G$ has at most log-divergent subgraphs if $2 h_I \leq |I|$ for all $I\subset E(G)$. 

\begin{rem} \label{remvertexnumbers} Let $G$ be connected in $\phi^4$. Suppose that $G$ has $v_i$ vertices of valency $i$, where $1\leq i \leq 4$. A simple computation using Euler's formula shows that 
\begin{equation}\label{vertexcount}
v_3 +2\, v_2 = 4 -2\, sd(G)\ . \end{equation}
In particular, suppose that $G$ is quadratically divergent. Then $v_3+2\, v_2= 2$ and $(v_2,v_3)$ is either $(1,0)$ or $(0,2)$.  In particular, every  non-trivial connected quadratic subgraph of a $\phi^4$ graph has precisely two $3$-valent vertices. 
\end{rem}

\subsection{Circular joins}
Let $G_1,\ldots, G_n$ denote $n\geq 2$ connected 1-scale graphs.
\begin{defn} Label the distinguished vertices of each graph $G_i$  by  $v_i,w_i$. 
A  \emph{circular join} $c(G_1,\ldots, G_n)$  of the graphs $G_1,\ldots, G_n$ is the graph obtained by gluing  the $G_i$ together to form a circle, by identifying the vertices
  $v_i$ and  $w_{i+1}$, for $i\in \Z/ n\Z$. It  clearly depends on these choices.  Define  $c(G)= \Gdot$. \end{defn}

\begin{lem} \label{lemcircjoin} Let $c(G_1,\ldots, G_n)$ be any circular join of $G_1,\ldots, G_n$. Then
\begin{equation} \label{PsiGcirc}
\psi_{c(G_1,\ldots, G_n)} = \phi_{G_1\cup \ldots \cup G_n} \ .
\end{equation}
\end{lem}

\begin{proof}  There is an obvious bijection:
$$T(G) \longleftrightarrow \bigcup_{i=1}^n T(G_1)\times \ldots \times T(\Gdot_i) \times \ldots \times T(G_n)\ ,$$
 (disjoint union) where $T(\Gamma)$ denotes the set of spanning trees in a graph $\Gamma$. \end{proof}
When $n=2$, the circular join  of $G_1$, $G_2$ is a  2-vertex join of the  graphs $\Gbar_1, \Gbar_2$.

\begin{defn} Let $\gamma,  \Gamma$ be single-scale graphs where $\Gamma$ is connected.
We define
\begin{equation} \label{Upsilonsdef} 
\Upsilon_{\gamma ; \Gamma} (s) =  s \,  \psi_{\gamma} \phi_{\Gamma}  + \phi_{\gamma} \psi_{\Gamma} \ ,
\end{equation}
where $s$ is an indeterminate which will later correspond to an external momentum.

\end{defn}

\subsection{Hopf algebras of 1-scale graphs} \label{sectHopf1scale} Let $G$ be a  fixed  log-divergent  graph with at most log-divergent subgraphs, and let  $H_G$ 
denote the finitely-generated $\Q$-Hopf algebra of graphs generated by $G$  under the coproduct $\Delta$ $(\ref{DirkDeltadef})$.
 Its underlying vector space is spanned by all logarithmically divergent subgraphs $\gamma$ of $G$, their cographs $G/\gamma$, and their disjoint unions.  That  no such  cograph $G/\gamma$ can have quadratic subdivergences is an easy consequence of the definitions.

Now suppose that $G$ is single-scale, and thus carries a pair of distinguished vertices. Suppose that for every  divergent subgraph $\Gamma$ of $G$ there is a choice of two distinguished connectors (Lemma \ref{onescale}) of $\Gamma$ in $G$, making $\Gamma$ single-scale, in such a way that for all divergent subgraphs $\gamma\subsetneq \Gamma$, 
\begin{equation}\label{epsilonHopf}
\varepsilon_{\Gamma}(\gamma) =0\ .\end{equation}
It follows that the cograph $\Gamma/\gamma$ inherits the two distinguished vertices of $\Gamma$, and therefore is single-scale in a uniquely determined way. The condition
$(\ref{epsilonHopf})$ ensures that the two distinguished vertices of $\Gamma$ are never identified in any of its cographs.

In this case, we have a Hopf algebra of 1-scale, log-divergent graphs $H_G$ and a  coassociative coproduct obtained by summing over all divergent subgraphs:
\begin{eqnarray}\label{singlescalecoproduct}
\Delta:  H_G& \To & H_G\otimes_{\Q} H_G  \\
 \Gamma & \mapsto & \sum_{\gamma\subset \Gamma} \gamma \otimes \Gamma/\gamma \ .\nonumber 
\end{eqnarray} 
The existence of $H_G$ is proved in \S \ref{aux}. 
Since $H_G$ is graded by the loop number and commutative, it automatically inherits an antipode $S:H_G \rightarrow H_G$.

\begin{var} \label{variant}  Suppose that the edges of $G$ are labelled by elements of some set. Then every subgraph and  cograph  of $G$, and hence every generator of $H_G$ inherits a labelling.
 Note that in this case, the choice of two distinguished connectors on each divergent subgraph of $G$ may depend on the labelling. 
Thus  we can have distinct subgraphs $\gamma_1$ and $\gamma_2$ of $G$ which are isomorphic after forgetting their labels, but have different choices of distinguished vertices. In the sequel,  the set of labels on the edges of a graph will  always be distinct.
\end{var}

\subsection{Feynman rules for 1-scale graphs} Let $G$ be any graph  (which is not necessarily connected), whose edges are labelled 
 $1,\ldots, N$,  and 
let $\alpha_i$ denote the corresponding Schwinger parameters. Write
\begin{equation}  
\Omega_{G} = \sum_{i=1}^N (-1)^i \alpha_i d\alpha_1 \wedge \ldots \wedge \widehat{d\alpha_i} \wedge \ldots \wedge d\alpha_N\ .
\label{Omegag}
\end{equation}
 \begin{defn} Let $\gamma, \Gamma$ denote labelled $1$-scale graphs where $\Gamma$ is connected. If $\gamma, \Gamma$ have distinct labels, then 
 define rational forms 
 \begin{equation} \label{omega2def}
\omega_{\gamma  \otimes \Gamma}(s)  =  { s\, \phi_{\Gamma} \over   \psi_{\gamma}  \psi^2_{\Gamma}   \Upsilon_{\gamma; \Gamma}(s)  } \, \Omega_{\gamma \cup \Gamma}
\end{equation} 
where we recall that $\psi_{\gamma}$ is the product of the graph polynomials of its components.
It follows from the degree formulae stated in lemmas \ref{lemphi} and \ref{lemcircjoin} that the homogeneous degree of $\omega_{\gamma\otimes \Gamma}(s)$ with respect to the Schwinger parameters is
\begin{equation} \label{degomega}
\deg\, \omega_{\gamma  \otimes \Gamma}(s) = (E_{\gamma} - 2 h_{\gamma}) +  (E_{\Gamma} - 2 h_{\Gamma}) \ . 
\end{equation}
In particular, if $\gamma, \Gamma$ are logarithmically divergent then $\omega_{\gamma\otimes \Gamma}(s)$ is of degree $0$.

Recall that the  Feynman differential form of any labelled  graph $G$ is 
\begin{equation}  \label{omega1def}
\omega_G = { \Omega_G \over \psi^2_{G}}\ .\end{equation}
It has degree $E_{G} - 2h_G$.  Since $G$ is not necessarily connected, $\psi_G$ can be a product of graph polynomials.  If $1$ denotes the empty graph, we extend 
$(\ref{omega2def})$ and $(\ref{omega1def})$ by
\begin{equation} \label{omega1tens}
 \omega_{\gamma \otimes 1 } (s)=   0 \quad  \hbox{ and } \quad Ê\omega_{1 \otimes \Gamma} (s)  =   \omega_{\Gamma} \ . 
\end{equation}
In particular,  we set $\omega_1=0$.
It can also be convenient to  set $\omega_{\gamma\otimes \Gamma}(s)=0$ (respectively $\omega_{\Gamma}=0$) if $\gamma\cup \Gamma$ (resp. $\Gamma$) have 
 repeated labels. 
\end{defn}

\section{Renormalized Parametric integrals and examples} 
We state  the main formula for the renormalized amplitudes of $1$-scale graphs with logarithmic subdivergences, with some examples.
  The proofs are spread   over  the following sections, and the case of quadratic subdivergences is treated in \S\ref{sectQuadratic}.

\subsection{Renormalised integrands} Throughout, let $H$ denote a  Hopf algebra $H_G$ of labelled 1-scale graphs, as defined in  \S\ref{sectHopf1scale}, variant \ref{variant}.
\begin{defn} Let $\Gamma\in H$ be a labelled $1$-scale graph, with edges labelled $1,\ldots, N$. Denote its image under the preparation map (definition  \ref{defnR}) by
\begin{equation}\label{RGammaexplicit}
R(\Gamma) = 1 \otimes \Gamma + \sum_{i=1}^n a_i \, \gamma_i \otimes \Gamma /\gamma_i\qquad \in H\otimes_{\Q} H\ ,
\end{equation}
where  $a_i =\pm 1$ are signs\footnote{These signs are consistent with the signs in the forest sums $\sum_f^\emptyset (-1)^{|f|}$.}.  Let $q\in \R^4$ denote the overall incoming momentum of $\Gamma$,  let $\mu \in \R^4$ denote a fixed momentum, and  set  $s=q^2/\mu^2$.  We define
\begin{equation} \label{defnomegaren}
\omega^{\mathrm{ren}}_{\Gamma}(s) = \omega_{\Gamma}+ \sum_{i=1}^n a_i    \,  \omega_{\gamma_i\otimes \Gamma/\gamma_i}(s)  \ , 
\end{equation}
to be the renormalised parametric Feynman integrand.  For any linear combination of connected graphs   $\xi=\sum_i c_i \Gamma_i \in H$, we define $\omega^{\ren}_{\xi}(s)= \sum_i c_i \, \omega^{\ren}_{\Gamma_i}(s)$.
\end{defn}
The domain of integration is given by the standard coordinate simplex: 
\begin{equation} \label{defnDeltaN}
\DD_{\Gamma}= \{(\alpha_1:\ldots :\alpha_N ) \subset \Pro(\R)^{N-1}: \alpha_i\geq 0\} \ .
\end{equation}
Let $\Gamma$ be a log-divergent single-scale graph with at most log-divergent subgraphs. Then by  $(\ref{degomega})$,
$\omega^{\ren}_{\Gamma}(s)$ is of degree $0$ and we can consider the projective integral
\begin{equation}\label{defnrenint}
f_{\Gamma}(s)= \int_{\DD_{\Gamma}} \omega^{\ren}_{\Gamma}(s) \ .\end{equation}
We write $f_{\Gamma}=f_{\Gamma}(1)$.
The following theorem will be proved in \S\ref{sectCancelpoles}.
\begin{thm} \label{thmconv}  The integral  $(\ref{defnrenint})$ converges when $s>0$.
\end{thm} 
The renormalization group equations amount to the following.
\begin{thm} \label{Renormeq}  Write the reduced coproduct $\Delta'\Gamma= \sum_i \gamma_i \otimes \Gamma/\gamma_i$. We have
\begin{equation}\label{firstrenormeq}
s\, {d \over ds} f_{\Gamma}(s) = \sum_{i} f_{\gamma_i} f_{\Gamma/\gamma_i}(s) \ \hbox{ forÊ}\, \,  s>0 \ .
\end{equation} 
\end{thm}

In the case when $\Gamma\in H$ is a primitive graph, i.e., $\Delta \Gamma = 1\otimes \Gamma + \Gamma \otimes 1$, then 
$\omega^{\ren}_{\Gamma}(s)= \omega_{\Gamma}$
and  does not depend on $s$. In this case,  $(\ref{defnrenint})$ reduces to the  residue:
$$f_{\Gamma}=\int_{\DD_{\Gamma}} {\Omega_{\Gamma}\over \psi_{\Gamma}^2} \ .$$

\subsection{Relation with BPHZ} Let $\Gamma \in H$ be a $1$-scale graph as above, and let $R(\Gamma)$ be given by  $(\ref{RGammaexplicit})$.
Then the  renormalized Feynman rules   are given by Eq. (\ref{renFR}):
\begin{eqnarray} \label{IGammafull} 
I(s): H &\To& \C \nonumber \\
I_{\Gamma}(s)& = & \int_{\DD_{\Gamma}}  \Big( {\log s \over \psi^2_{\Gamma}} + \sum_{i=1}^n a_i {\log\big( {s \psi_{\gamma_i} \phi_{\Gamma/\gamma_i} + \phi_{\gamma_i} \psi_{\Gamma / \gamma_i}  \over \psi_{\gamma_i} \phi_{\Gamma/\gamma_i} + \phi_{\gamma_i} \psi_{\Gamma / \gamma_i}}\big) \over \psi_{\gamma_i}^2 \psi_{\Gamma/\gamma_i}^2}   \Big) \Omega_{\Gamma} 
\end{eqnarray}
where $s=q^2/\mu^2$, assumed to be positive.
It follows immediately that $I_{\Gamma}(1)=0$ and that 
$ s\, I'_{\Gamma}(s) = f_{\Gamma}(s)$
by definitions $(\ref{Upsilonsdef})$ and $(\ref{defnomegaren})$.  In particular, the convergence of  $(\ref{IGammafull})$  follows from theorem \ref{thmconv}.
Again,  we extend the definition of $I_{\bullet}(s)$ by linearity to linear combinations of connected graphs $\xi \in H$.

\subsubsection{Example 1: The case of a single subdivergence} Let $\Gamma$ be a  log-divergent $1$-scale graph with a single log-divergent subgraph $\gamma\subset \Gamma$. Label the edges of $\Gamma$ from $1$ to $N$, and fix  a $1$-scale structure on $\gamma$. Fixing such a structure involves   choosing any two vertices of $\gamma$ which meet $\Gamma\backslash \gamma$. By definition,
$\Delta(\Gamma) = 1\otimes \Gamma + \Gamma \otimes 1 + \gamma \otimes \Gamma/\gamma,$
and therefore the renormalized Feynman rules are:
$$f_{\Gamma}(s)= \int_{\DD_{\Gamma}} \omega^{\ren}_{\Gamma}(s) = \int_{\DD_{\Gamma}} \Big( {1\over \psi_{\Gamma}^2}  - {s \, \phi_{\Gamma/\gamma}  \over  \psi_{\gamma} \psi_{\Gamma/\gamma}^2  \Upsilon_{\gamma;G/\gamma}(s)} \Big) \Omega_{\Gamma}\ ,$$
where $ \Upsilon_{\gamma;G/\gamma}(s) =\phi_{\gamma} \psi_{G/\gamma} + s\, \psi_{\gamma} \phi_{G/\gamma}$. 
The group equation  $(\ref{firstrenormeq})$ gives
$$f_{\Gamma}(s) = f_{\Gamma}  + f_{\gamma}  f_{\Gamma/\gamma} \log s $$
 Thus all new information about $\Gamma$ is encoded in $f_{\Gamma}$. We can rewrite this entirely in terms of graph polynomials using the definition of circular joins:
\begin{equation} \label{f1asgraphs}
f_{\Gamma} = \int_{\DD_{\Gamma}}   \Big( {1\over \psi_{\Gamma}^2}  - {\psi_{\Gamma^{\bullet}/\gamma}  \over \psi_{\Gamma/\gamma}^2 \psi_{\gamma} \psi_{c(\gamma, \Gamma/\gamma)}} \Big)\,  \Omega_{\Gamma}\ .
\end{equation} 
The smallest non-trivial example of such a graph $\Gamma$ which is simple (i.e., does not reduce to a smaller graph by series-parallel operations) is the wheel with 3 spokes inserted into itself. The five graphs which occur in $(\ref{f1asgraphs})$ are illustrated below: 
\begin{center}
\fcolorbox{white}{white}{
  \begin{picture}(310,143) (464,-57)
    \SetWidth{1.0}
    \SetColor{Black}
    \Arc[clock](716.378,-27.5)(25.627,126.875,-126.875)
    \Arc(685,-27.5)(26.005,52.028,307.972)
    \Text(620,-47)[lb]{\small{\Black{$c(\gamma,\Gamma}/\gamma)$}}
        \Text(730,-2)[lb]{\tiny{\Black{$2$}}}
       \Text(730,-25)[lb]{\tiny{\Black{$4$}}}
      \Text(730,-47)[lb]{\tiny{\Black{$3$}}}
      \Text(713,-17)[lb]{\tiny{\Black{$6$}}}
      \Text(713,-40)[lb]{\tiny{\Black{$5$}}}
      \Text(708,-28)[lb]{\tiny{\Black{$1$}}}
      \Text(695,-28)[lb]{\tiny{\Black{$8$}}}
      \Text(686,-17)[lb]{\tiny{\Black{$12$}}}
      \Text(686,-40)[lb]{\tiny{\Black{$10$}}}
          \Text(672,-2)[lb]{\tiny{\Black{$7$}}}
       \Text(670,-25)[lb]{\tiny{\Black{$11$}}}
      \Text(672,-47)[lb]{\tiny{\Black{$9$}}}
    \Arc(753,49)(20,270,630)
    %center of Gamma/gamma is 753,49
    \Line(753,71)(753,49)
    \Line(753,49)(740,35)
    \Line(753,49)(767,35)
    \Text(709,29)[lb]{{\Black{$\Gamma/\gamma$}}}
    \Vertex(753,69){2.5}
    \Vertex(740,35){2.5}
    \Line(659,-27)(687,-27)
    \Line(687,-27)(700,-7)
    \Line(687,-27)(700,-47)
    \Line(701,-7)(714,-27)
    \Line(714,-27)(701,-47)
    \Line(714,-27)(742,-27)
    \Arc(751.5,-27.5)(54.502,157.906,202.094)
    \Arc[clock](650.5,-27.5)(54.502,22.094,-22.094)
    \Vertex(701,-7){2.5}
    \Vertex(701,-48){2.5}
    \Arc(516,49)(35,270,630)
    \Arc(516,49)(13,270,630)
    %center of Gamma is 516,49
    \Line(516,82)(516,49)
    \Line(516,49)(488,29)
    \Line(516,49)(545,29)
    \Text(461,29)[lb]{{\Black{$\Gamma$}}}
    \Text(480,70)[lb]{\tiny{\Black{$1$}}}
    \Text(520,70)[lb]{\tiny{\Black{$6$}}}
    \Text(550,70)[lb]{\tiny{\Black{$2$}}} 
    \Text(503,58)[lb]{\tiny{\Black{$7$}}}
    \Text(530,58)[lb]{\tiny{\Black{$8$}}}
    \Text(519,52)[lb]{\tiny{\Black{$12$}}}
    \Text(508,45)[lb]{\tiny{\Black{$11$}}}
    \Text(522,45)[lb]{\tiny{\Black{$10$}}}
    \Text(517,38)[lb]{\tiny{\Black{$9$}}}
        \Text(538,37)[lb]{\tiny{\Black{$4$}}}
        \Text(494,37)[lb]{\tiny{\Black{$5$}}}
         \Text(517,17)[lb]{\tiny{\Black{$3$}}}
 %   \Text(515,27)[lb]{{\Black{$\gamma$}}}
    \Vertex(516,84){2.5}
     \Vertex(487.5,28.5){2.5}
      \Vertex(516,62){1.5}
      \Vertex(526.5,41.5){1.5}
    \Arc[clock](563,-46.025)(39.837,151.473,28.527)
    \Arc(515,-27)(13,270,630)
    \Arc[clock](549,-60.75)(39.75,121.891,58.109)
    \Arc(549,-7)(29,-136.397,-43.603)
    \Arc(563,-8.929)(39.39,-152.691,-27.309)
    \Line(570,-27)(598,-27)
    \Text(476,-47)[lb]{\small{\Black{$\Gamma^{\bullet}/\gamma$}}}
      \Text(507,-30)[lb]{\tiny{\Black{$1$}}}
      \Text(550,-20)[lb]{\tiny{\Black{$6$}}}
      \Text(550,-33)[lb]{\tiny{\Black{$5$}}}
      \Text(580,-7)[lb]{\tiny{\Black{$2$}}}
      \Text(580,-25)[lb]{\tiny{\Black{$4$}}}
        \Text(580,-41)[lb]{\tiny{\Black{$3$}}} 
    \Vertex(528,-27){2.5}
    \Arc(640,49)(20,270,630)
    %center of gamma is 640,49
    \Line(640,70)(640,49)
    \Line(640,49)(626,35)
    \Line(640,49)(653,35)
    \Text(604,29)[lb]{{\Black{$\gamma$}}}
    \Text(640,22)[lb]{\tiny{\Black{$9$}}}
    \Text(626,43)[lb]{\tiny{\Black{$11$}}}
    \Text(650,43)[lb]{\tiny{\Black{$10$}}}
    \Text(633,55)[lb]{\tiny{\Black{$12$}}}
    \Text(620,60)[lb]{\tiny{\Black{$7$}}}
    \Text(661,60)[lb]{\tiny{\Black{$8$}}}
   \Text(755,22)[lb]{\tiny{\Black{$3$}}}
    \Text(743,43)[lb]{\tiny{\Black{$5$}}}
    \Text(764,43)[lb]{\tiny{\Black{$4$}}}
    \Text(749,55)[lb]{\tiny{\Black{$6$}}}
    \Text(734,60)[lb]{\tiny{\Black{$1$}}}
    \Text(776,60)[lb]{\tiny{\Black{$2$}}}
      \Vertex(640,69){2.5}
    \Vertex(654,34){2.5}
  \end{picture}
}
\end{center}
It is surprising that the quantity $(\ref{f1asgraphs})$ is not known for this graph,  
as it is the simplest possible example of a graph which requires renormalization and does not reduce trivially to smaller graphs by series-parallel operations.
 However, one easily checks that the parametric integration method of \cite{BrCMP}, \cite{BrFeyn} applies to this case.

\subsubsection{Example 2: overlapping subdivergences} \label{sectexoverlap} Consider the graph $\Gamma$  below,
see also Eq.(\ref{firstexa}). 
\begin{center}
\fcolorbox{white}{white}{
  \begin{picture}(120,76) (257,-27)
    \SetWidth{1.0}
    \SetColor{Black}
    \Line(259,7)(313,45)
    \Line(259,7)(313,-26)
    \Line(313,45)(371,7)
    \Line(371,7)(313,-26)
    \Arc(341.382,9.5)(45.451,128.642,231.358)
    \Arc[clock](284.618,9.5)(45.451,51.358,-51.358)
   % \Vertex(259,7){2}   
    \CBox(262,9)(258,5){Black}{Black}%%JaxoDrawID:Vertex
    % \Vertex(371,7){2}
    \CBox(372,9)(368,5){Black}{Black}%%JaxoDrawID:Vertex
      \Text(235,18)[lb]{\Large{\Black{$\Gamma$}}}
    \Text(276,28)[lb]{{\Black{$1$}}}
    \Text(279,-18)[lb]{{\Black{$2$}}}
    \Text(300,7)[lb]{{\Black{$3$}}}
    \Text(322,7)[lb]{{\Black{$4$}}}
    \Text(344,28)[lb]{{\Black{$5$}}}
    \Text(344,-18)[lb]{{\Black{$6$}}}
  \end{picture}
}
\end{center}
 It has the three overlapping log-divergent subgraphs shown  below. We give them single-scale structures:
\begin{center}
\fcolorbox{white}{white}{
  \begin{picture}(-100,76) (302,-27)
    \SetWidth{1.0}
    \SetColor{Black}
    \Line(159,7)(213,45)
    \Line(159,7)(213,-26)
    \Arc(241.382,9.5)(45.451,128.642,231.358)
    \Arc[clock](184.618,9.5)(45.451,51.358,-51.358)
   % \Vertex(259,7){2}   
    \CBox(162,9)(158,5){Black}{Black}%%JaxoDrawID:Vertex
    % \Vertex(371,7){2}
    \CBox(215,46)(211,42){Black}{Black}%%JaxoDrawID:Vertex
    \Text(140,-10)[lb]{\Large{\Black{$\gamma_l$}}}
     \Text(176,28)[lb]{{\Black{$1$}}}
    \Text(179,-18)[lb]{{\Black{$2$}}}
    \Text(200,7)[lb]{{\Black{$3$}}}
    \Text(222,7)[lb]{{\Black{$4$}}}
    %middle
       \Text(300,7)[lb]{{\Black{$3$}}}
    \Text(322,7)[lb]{{\Black{$4$}}}
  \Arc(341.382,9.5)(45.451,128.642,231.358)
    \Arc[clock](284.618,9.5)(45.451,51.358,-51.358)
     \CBox(315,46)(311,42){Black}{Black}%%JaxoDrawID:Vertex
     \CBox(315,-24)(311,-28){Black}{Black}%%JaxoDrawID:Vertex
    \Text(270,-10)[lb]{\Large{\Black{$\gamma_{34}$}}}
    %right
    \Line(413,45)(471,7)
    \Line(471,7)(413,-26)
    \Arc(441.382,9.5)(45.451,128.642,231.358)
    \Arc[clock](384.618,9.5)(45.451,51.358,-51.358)
   % \Vertex(259,7){2}   
   \CBox(415,46)(411,42){Black}{Black}%%JaxoDrawID:Vertex
    % \Vertex(371,7){2}
    \CBox(472,9)(468,5){Black}{Black}%%JaxoDrawID:Vertex  
    \Text(380,-10)[lb]{\Large{\Black{$\gamma_{r}$}}}
    \Text(400,7)[lb]{{\Black{$3$}}}
    \Text(422,7)[lb]{{\Black{$4$}}}
    \Text(444,28)[lb]{{\Black{$5$}}}
    \Text(444,-18)[lb]{{\Black{$6$}}}
      \end{picture}
}
\end{center}
They generate a Hopf algebra of single-scale graphs $H_{\Gamma}$. 
In  general, let $\gamma_{ij}$ denote a 2-edge banana graph on edges $i,j$ with its unique  single-scale structure.
The coproduct is $\Delta(\Gamma) = 1\otimes \Gamma + \Gamma \otimes 1 + \sum_{i\in \{l,34,r\}} \gamma_i \otimes \Gamma/\gamma_i$, hence
$$\Delta^{(1)}(\Gamma) = \gamma_l \otimes \gamma_{56} + \gamma_r \otimes \gamma_{12} +\gamma_{34} \otimes \gamma_{12}.\gamma_{56}$$
where $\gamma_{12}.\gamma_{56}=\Gamma/\gamma_{34}$ is the  1-vertex join of $\gamma_{12}$ and $\gamma_{56}$ whose single scale structure is given by the two outer  (non-joined) vertices. It follows that 
$$\Delta^{(2)}(\Gamma) = \gamma_{34} \otimes \gamma_{12} \otimes \gamma_{56}  +    \gamma_{34} \otimes \gamma_{56} \otimes \gamma_{12}\ .$$
The preparation map applied to $\Gamma$  therefore gives $R(\Gamma)=$
$$ 1\otimes \Gamma - \big(  \gamma_l \otimes \gamma_{56} + \gamma_r \otimes \gamma_{12} +\gamma_{34} \otimes \gamma_{12}.\gamma_{56} \big) 
+ \big( \gamma_{34} \gamma_{12} \otimes \gamma_{56}  +    \gamma_{34} \gamma_{56} \otimes \gamma_{12}\big) \ , $$
 where $\gamma_{ij}\gamma_{kl}$ denotes the disjoint union $\gamma_{ij}\cup \gamma_{kl}$. Thus 
 \begin{multline} 
 \omega^{\ren}_{\Gamma}(s) = \Big[ {1 \over \psi_{\Gamma}^2} -  {s\,  \phi_{\gamma_{56}} \over  \psi_{\gamma_{l}}  \psi_{\gamma_{56}}^2 \Upsilon_{\gamma_l; \gamma_{56}}(s)}  - {s\,  \phi_{\gamma_{12}}   \over  \psi_{\gamma_{r}}  \psi_{\gamma_{12}}^2 \Upsilon_{\gamma_r; \gamma_{12}}(s)}  - {s\,   \phi_{\gamma_{12}.\gamma_{56}} \over  \psi_{\gamma_{34}}  \psi_{\gamma_{12}.\gamma_{56}}^2  \Upsilon_{\gamma_{34}; \gamma_{12}.\gamma_{56}}(s)}  \nonumber \\
  \quad \qquad + \quad   { s \, \phi_{\gamma_{56}} \over  \psi_{\gamma_{34}} \psi_{\gamma_{12}} \psi_{\gamma_{56}}^2 \Upsilon_{\gamma_{34}\gamma_{12}; \gamma_{56}}(s)} 
 +{ s \, \phi_{\gamma_{12}} \over  \psi_{\gamma_{34}} \psi_{\gamma_{56}} \psi_{\gamma_{12}}^2 \Upsilon_{\gamma_{34}\gamma_{56}; \gamma_{12}}(s)}   \Big]\, \Omega_{\Gamma} \nonumber
 \end{multline}
 where 
 $\psi_{\gamma_{ij}}= \alpha_i +\alpha_j$,   $\phi_{\gamma_{ij}}= \alpha_i\alpha_j$, and 
 \begin{eqnarray}
  \psi_{\gamma_{l}}= (\alpha_1+\alpha_2)(\alpha_3+\alpha_4)+ \alpha_3\alpha_4  & &  \phi_{\gamma_{l}}= \alpha_1 (\alpha_2\alpha_3+\alpha_2\alpha_4+\alpha_3\alpha_4)  \nonumber \\
  \psi_{\gamma_{r}}= (\alpha_5+\alpha_6)(\alpha_3+\alpha_4)+ \alpha_3\alpha_4 & & \phi_{\gamma_{r}}= \alpha_5( \alpha_3\alpha_4+ \alpha_3\alpha_6+\alpha_4 \alpha_6)  \nonumber \\
  \psi_{\gamma_{12}.\gamma_{56}}= (\alpha_1+\alpha_2)(\alpha_5+\alpha_6)   & & \phi_{\gamma_{12}.\gamma_{56} }= 
  \alpha_1\alpha_2\alpha_5+    \alpha_1\alpha_2\alpha_6+ \alpha_2\alpha_5\alpha_6+ \alpha_1\alpha_5\alpha_6 \nonumber
 \end{eqnarray}
 and all remaining polynomials $\Upsilon$  are deduced from these by  $(\ref{Upsilonsdef})$. In particular,
 \begin{eqnarray}
 \Upsilon_{\gamma_l; \gamma_{56}}(s)&  = & s \, \psi_{\gamma_l} \phi_{\gamma_{56}} + \phi_{\gamma_l} \psi_{\gamma_{56}} \nonumber \\
 \Upsilon_{\gamma_{34}; \gamma_{12}.\gamma_{56}}(s)&  = & s\,    \psi_{\gamma_{34}} \phi_{\gamma_{12}.\gamma_{56}}  +\phi_{\gamma_{34}} \psi_{\gamma_{12}.\gamma_{56}}  \nonumber \\
 \Upsilon_{\gamma_{34}\gamma_{56};\gamma_{12}}(s) & = & s\, \psi_{\gamma_{34}} \psi_{\gamma_{56}} \phi_{\gamma_{12}} + \psi_{\gamma_{34}}   \phi_{\gamma_{56}}  \psi_{\gamma_{12}} +      
 \phi_{\gamma_{34}}   \psi_{\gamma_{56}}  \psi_{\gamma_{12}}  \nonumber 
 \end{eqnarray}
In the case $s=1$,  the terms $\Upsilon_{\gamma_l;\gamma_{56}}(1), \Upsilon_{\gamma_r;\gamma_{12}}(1) , \Upsilon_{\gamma_{34}; \gamma_{12}.\gamma_{56}}(1)$
 and $ \Upsilon_{\gamma_{34} \gamma_{12}; \gamma_{56}} (1)=  \Upsilon_{\gamma_{34} \gamma_{56}; \gamma_{12}}(1)$ which occur  in the denominator of $\omega^{\ren}_{\Gamma}(1)$ are the graph polynomials of the following four  circular join graphs, from left to right: 
\begin{center}
\fcolorbox{white}{white}{
  \begin{picture}(354,68) (79,-68)
    \SetWidth{1.0}
    \SetColor{Black}
    \Line(80,-31)(132,-57)
    \Line(80,-31)(132,-5)
    \Arc(197.1,-31)(70.1,158.229,201.771)
    \Arc(66.9,-31)(70.1,-21.771,21.771)
    \Arc[clock](132,-70)(65,143.13,90)
    \Arc[clock](117.611,-41.222)(38.975,164.795,68.335)
    \Line(366,-46)(399,-3)
    \Line(399,-3)(430,-46)
    \Line(430,-46)(366,-46)
    % Singular FArc, ignored!
    \Text(88,-68)[lb]{{\Black{$c(\gamma_l,\gamma_{56})$}}}
    \Text(96,-5)[lb]{\tiny{\Black{$5$}}}  
    \Text(98,-12)[lb]{\tiny{\Black{$6$}}}  
    \Text(106,-24)[lb]{\tiny{\Black{$1$}}}  
    \Text(106,-42)[lb]{\tiny{\Black{$2$}}}  
    \Text(122,-32)[lb]{\tiny{\Black{$3$}}}  
    \Text(132,-32)[lb]{\tiny{\Black{$4$}}}  
   \Text(161,-32)[lb]{\tiny{\Black{$3$}}}  
    \Text(171,-32)[lb]{\tiny{\Black{$4$}}}  
     \Text(194,-24)[lb]{\tiny{\Black{$5$}}}  
    \Text(194,-42)[lb]{\tiny{\Black{$6$}}}  
      \Text(204,-5)[lb]{\tiny{\Black{$1$}}}  
    \Text(202,-12)[lb]{\tiny{\Black{$2$}}}  
      \Text(272,-34)[lb]{\tiny{\Black{$1$}}}  
    \Text(272,-48)[lb]{\tiny{\Black{$2$}}}  
      \Text(294,-10)[lb]{\tiny{\Black{$3$}}}  
    \Text(294,-19)[lb]{\tiny{\Black{$4$}}}  
      \Text(318,-34)[lb]{\tiny{\Black{$5$}}}  
    \Text(318,-48)[lb]{\tiny{\Black{$6$}}}  
      \Text(384,-10)[lb]{\tiny{\Black{$1$}}}  
    \Text(384,-19)[lb]{\tiny{\Black{$2$}}}   
      \Text(412,-10)[lb]{\tiny{\Black{$3$}}}  
    \Text(412,-19)[lb]{\tiny{\Black{$4$}}}   
      \Text(398,-45)[lb]{\tiny{\Black{$5$}}}  
    \Text(398,-52)[lb]{\tiny{\Black{$6$}}}   
    \Text(180,-68)[lb]{{\Black{$c(\gamma_r,\gamma_{12})$}}}
    \Text(265,-68)[lb]{{\Black{$c(\gamma_{34}, \gamma_{12}.\gamma_{56})$}}}
    \Text(370,-68)[lb]{{\Black{$c(\gamma_{12},\gamma_{34},\gamma_{56})$}}}
    \Arc(398,23.571)(76.578,-114.7,-65.3)
    \Arc(368.995,-57.306)(62.044,10.499,61.079)
    \Arc(421.716,-54.596)(56.375,113.762,171.229)
    \Line(171,-5)(222,-31)
    \Line(171,-56)(222,-31)
    \Arc(233.5,-30.5)(67.502,157.805,202.195)
    \Arc(108.5,-30.5)(67.502,-22.195,22.195)
    \Arc[clock](185.724,-39.137)(37.177,113.331,12.643)
    \Arc[clock](166.644,-76.564)(71.696,86.517,39.458)
    \Arc[clock](273,-54.5)(27.5,143.13,36.87)
    \Arc(273,-23.833)(26.167,-147.221,-32.779)
    \Arc[clock](318,-56.545)(29.545,141.12,38.88)
    \Arc[clock](318,-21.958)(28.042,-34.894,-145.106)
    \Arc[clock](296,-65.98)(52.989,148.128,31.872)
    \Arc(296,-52.167)(47.177,17.475,162.525)
  \end{picture}
}
\end{center}
The renormalization group equation in this case reduces to 
$$s\, f_{\Gamma}'(s) =f_{\gamma_{\ell}}  f_{\gamma_{56}}(s) + f_{\gamma_{r}}  f_{\gamma_{12}}(s)+ f_{\gamma_{34}} f_{\gamma_{12}.\gamma_{56}}(s)$$
Since the $\gamma_{ij}$ are all isomorphic and primitive, and since $\gamma_r\cong \gamma_{\ell}$ this reduces to
$$s\, f_{\Gamma}'(s) =2 \, f_{\gamma_{\ell}}  f_{\gamma_{56}}  + f_{\gamma_{34}}  f_{\gamma_{12}.\gamma_{56}}(s) $$
Applying the group equation to $\gamma_{12}.\gamma_{56}$, we see that $ f_{\gamma_{12}.\gamma_{56}}(s)= f_{\gamma_{12}.\gamma_{56}}+ 2 f_{\gamma_{12}}^2 \log s $
by a similar calculation. 
All in all, we obtain
$$  f_{\Gamma}(s) =f_{\Gamma}+ (2 \, f_{\gamma_{\ell}} + f_{\gamma_{12}.\gamma_{56}} ) f_{\gamma_{12}}  \log s +  f_{\gamma_{12}}^3 \log^2 s . $$

\section{Residues} \label{sectResidues}

 Let $G$ be a   graph with edges labelled from $1$ to $N$. Then the Schwinger parameters $\alpha_1,\ldots, \alpha_N$ define projective
coordinates $(\alpha_1:\ldots :\alpha_N)$ on $\Pro^{N-1}$. For any  strict subset of edges $I\subset E(G)$, let  
\begin{equation}\label{LIdef} 
L_I = \bigcap_{i\in I} \{\alpha_i=0\} \subset \Pro^{N-1}
\end{equation}
be the corresponding coordinate  linear subvariety. If $\omega$ is a regular $k$-form  on  a Zariski-open subset of $\Pro^{N-1}$,   let 
$v_I(\omega)$ denote the order of vanishing of $\omega$ along $L_I$.

\begin{lem} \label{lemdivparts} Let $I \subsetneq G$ be a subgraph which is not necessarily connected. 
Then 
\begin{equation} 
\psi_G   = \psi_{I} \psi_{G/I} + R \ , \label{PsiRemainder} 
\end{equation}
where $R$ is some polynomial of degree strictly  greater than $h_I= \deg \psi_{I} $ in the Schwinger parameters corresponding
to $I$.  It follows that $v_I(\psi_G)= h_I$. 
\end{lem}
Equation $(\ref{PsiRemainder})$ is well-known, and  its generalization to products of graphs is given by the following corollary.
In the sequel we use the following notation: if $I \subset E(G)$, and $\gamma $ is a subgraph of $G$, let $I_{\gamma}=I \cap \gamma$, and let 
$I_{G/\gamma} = I_G / I_{\gamma}$.

\begin{cor} \label{corremainders}
 Let  $G_1,\ldots, G_m$ be connected $1$-scale graphs, let $G=\cup_i G_i$, and let $I_i\subset G_i$  be subgraphs for $i=1,\ldots ,m$,  where at least one $I_i\neq G_i$.  Then  
   \begin{equation} \label{PhiRemainder} v_I(\phi_{G}) =\sum_{i=1}^m  h_{I_i}  + \prod_{i=1}^m \varepsilon_{G_i}(I_i)
   \end{equation}
   where $I= \cup_{i=1}^m I_i$. 
Suppose that  for every $i=1,\ldots, m$,  either $\varepsilon_{G_i}(I_i)=0$ or $I_i=G_i$.  Then the non-zero quotients $G_i/I_i$ are   single-scale graphs, and 
  \begin{equation} \label{CircRemainder}
\phi_{G}  = 
                            \psi_{I} \phi_{G/I} + R       \end{equation}
where  $R$ is  of degree strictly greater than  
$ \deg \psi_{I}$  in the variables of $\psi_I$. Likewise, for any $1$-scale graphs $\gamma=\cup_i \gamma_i, \Gamma$ with $\gamma_i, \Gamma$ connected, and  $I\subset \gamma \cup \Gamma$, we
have 
\begin{equation}\label{114star}
v_I(\Upsilon_{\gamma; \Gamma}(s)) = \sum_{i=1}^m h_{I_{\gamma_i}} +   \varepsilon_{\gamma}(I _{\gamma}) \varepsilon_{\Gamma}(I_{\Gamma})\ .
\end{equation}
If $\varepsilon_{\gamma_i} (I_{\gamma_i})=0$ (resp. $\varepsilon_{\Gamma} (I_{\Gamma})=0$) whenever $\gamma_i \subsetneq I$ (resp. $\Gamma \subsetneq I$), then 
\begin{equation}\label{UpsReminder}
\Upsilon_{\gamma;  \Gamma}(s)  = 
                            \psi_{I} \Upsilon_{\gamma/I_{\gamma}; \Gamma/I_{\Gamma}} (s)+ R  \  .
                            \end{equation}
                             \end{cor} 

\begin{proof}  Use the interpretation of $\phi_{G} $ as the graph polynomial of a circular join  $\psi_{c(G_1,\ldots, G_m)}$.  
 We know from $(\ref{PsiRemainder})$ that for any subgraph $I$ in  a  graph $G$, the order of vanishing of $\psi_G$ along $\alpha_I=0$ is  equal to the number of loops in  $I$.
Equations $(\ref{PhiRemainder})$ and $(\ref{CircRemainder})$   follow from computing the number of loops  of $I$ in   $G= c(G_1,\ldots, G_m)$, and 
 applying   $(\ref{PsiRemainder})$  to $G$.  The corresponding calculation for $\Upsilon_{\gamma;\Gamma}(s)$ follows from this case, by  definition 
 $(\ref{Upsilonsdef})$. \end{proof}

\subsection{Orders of poles} \label{sectOrdersofPoles} 
We  compute the orders of the poles of  $\omega_{\gamma \otimes \Gamma}(s)$ along  $L_I$.

\begin{lem}  Let $\gamma, \Gamma$ denote single-scale graphs, with $\Gamma$ connected, and let $I_\gamma \subset \gamma$,  $I_{\Gamma} \subset \Gamma$ and $I = I_{\gamma} \cup I_{\Gamma}$. Then  \begin{equation}\label{valuationformula}
-v_I(\omega_{\gamma\otimes \Gamma}(s)) =  2 h_{I_\gamma} + 2 h_{I_{\Gamma}} - \varepsilon_{\gamma\otimes \Gamma}(I)\ ,\end{equation}
where $\varepsilon_{\gamma\otimes \Gamma} \in \{0,1\}$ is defined by 
\begin{equation} \label{epsilontensdefn}
 \varepsilon_{\gamma\otimes \Gamma}= \varepsilon_{\Gamma} (1-    \varepsilon_{\gamma})\ . 
 \end{equation}
\end{lem} 

\begin{proof}
By the remarks preceding $(\ref{PhiRemainder})$, 
$v_I(\phi_{\Gamma}) = h_{I_{\Gamma}} + \varepsilon_{\Gamma} (I).$
Likewise it follows from corollary \ref{corremainders}  that  
 $v_I( \Upsilon_{\gamma;  \Gamma}(s)  ) = h_{I_\gamma}+ h_{I_\Gamma} +   \varepsilon_{\gamma} \varepsilon_{\Gamma}(I)$. Summing the contributions of each term in $(\ref{omega2def})$ gives the formula.
\end{proof}

From now on let us assume that $G$ is log-divergent and has at most logarithmically divergent subgraphs. A \emph{flag} in $G$ is a nested sequence of
divergent subgraphs 
\begin{equation}\label{flag}
F\ : \quad \gamma_1 \subset \gamma_2 \ldots \subset \gamma_n \subset G \qquad \qquad 
\end{equation}
where all inclusions are strict. Given such a flag $F$, let  us write:
\begin{equation}\label{gammaflag} 
\gamma_F = \gamma_1 \cup \gamma_2/\gamma_1 \cup \ldots \cup \gamma_n / \gamma_{n-1} \quad  \hbox{  and }  \quad 
 G/\gamma_F = G/ \gamma_n \ . 
 \end{equation} 
 If $I$ is a subset of edges of $G$,  write 
\begin{equation} \label{Igammanotation}
I_{\gamma_F} = I_{\gamma_1} \cup  I_{\gamma_2/\gamma_1} \cup \ldots \cup I_{\gamma_n/\gamma_{n-1}} \quad \hbox{ and } \quad   I_{G/\gamma_F} = I_{G/\gamma_n}\ .
\end{equation}
We say that $I_{\gamma_F}$ is divergent if $ I_{\gamma_1}$,  and  all $I_{\gamma_{i+1}/\gamma_i}$  are either empty or divergent.
\begin{cor} \label{corpoleorderalongLI}
 Let   $I\subsetneq E(G)$, and  let $F$  be a flag  in $G$.   The form $\omega_{\gamma_F\otimes G/\gamma_F}(s)$ has a pole along $L_I$ of order $\leq |I|$, with equality if and only if 
 $I_{\gamma_F}$, $I_{G/\gamma_F}$ are divergent, and $I_{G/\gamma_F}\subsetneq G_{\gamma_F}$ is a strict subgraph.
\end{cor}
\begin{proof} From  equation $(\ref{valuationformula})$, we have
\begin{eqnarray}
-v_I(\omega_{\gamma_F\otimes G/\gamma_F}(s))  & = & 2h_{I_{\gamma_F}} + 2h_{I_{G/\gamma_F}} - \varepsilon_{\gamma_F \otimes G/\gamma_F}(I) \ .\nonumber 
\end{eqnarray}
It follows that   $-v_I(\omega_{\gamma_F\otimes G/\gamma_F}(s))  \leq |I| - \varepsilon_{\gamma_F \otimes G/\gamma_F}(I)\leq |I|$ with equality if and only if 
$I_{\gamma_F}, I_{G/\gamma_F}$ are divergent and $\varepsilon_{\gamma_F \otimes G/\gamma_F}(I)=0$. But if $I_{G/\gamma_F}$ is divergent and a strict subgraph of $G/\gamma_F$, $\varepsilon_{G/\gamma_F} (I_{G/\gamma_F} )=0$  by  assumption $(\ref{epsilonHopf})$. By   $(\ref{epsilontensdefn})$, this gives  $\varepsilon_{\gamma_F \otimes G/\gamma_F}(I)=0$. In the case when $I_{G/\gamma_F}=G/\gamma_F$, we have $\varepsilon_{G/\gamma_F} (I_{G/\gamma_F} )=1$, and since $I\subset E(G)$ is strict,
we must have $I_{\gamma_F} \subsetneq \gamma_F$, so $\varepsilon_{\gamma_F} (I_{\gamma_F} )=0$  by  assumption $(\ref{epsilonHopf})$. Thus the case
when $I_{G/\gamma_F}=G/\gamma_F$  gives rise to a pole of order $\leq |I| -1$.
  \end{proof}

\subsection{Blow-ups and differential forms}
Let $\emptyset\neq I\subset E(G)$ be a  strict   subset of edges of $G$, whose edges are labelled $1,\ldots, N$.  Consider the blow-up 
$$\pi_I: P_I \To \Pro^{N-1}$$
of $\Pro^{N-1}$ along $L_I$. We denote the exceptional divisor by $\mathcal{E}_I$, which is isomorphic to $\Pro^{|I|-1} \times \Pro^{|I^c|-1}$, where $I^c$ is the complement of $I$ in $E(G)$.

\begin{lem} \label{lemgeneralresidues} Let $\omega$ be a differential form of the shape $(\prod_{k=1}^m P_k^{r_k}) \Omega_G$ where $r_k \in \Z$, and each $P_k$ is a homogeneous polynomial in the $\alpha_i$ of the form
$$P_k = A_k(\{\alpha_i\}_{ i\in I}) B_k(\{\alpha_j\}_{j\in I^c}) + R_k$$
where $A_k,B_k,R_k$ are homogeneous and the degree of $R_k$ in the variables $\alpha_i, i\in I$ is strictly greater than $\deg A_k$.  Then the order of the pole of $\pi_I^* \omega$ along $\mathcal{E}_I$ is 
\begin{equation}\label{generalpoleorder} 1- \sum_{k=1}^m r_k \deg(A_k) - |I| = 1 - v_I(\omega) - |I|
\end{equation}
where $v_I(\omega)$ is the order of vanishing of $\omega$ along $L_I$. The pole is simple if and only if $-v_I(\omega)= |I|$, in which case  the residue of $\pi_I^* \omega$ along $\mathcal{E}_I$ is given by
\begin{equation} \label{Residueformula}
\mathrm{Res}_{\mathcal{E}_I} \pi_I^* \omega = \prod_{k=1}^m A_k(\alpha_i, i\in I) \,\Omega_I \otimes \prod_{k=1}^m B_k(\alpha_j, j\in I^c) \, \Omega_{I^c} \end{equation}
\end{lem}
\begin{proof} Let  $I=\{i_1,\ldots, i_a\}$, and $I^c = \{ j_1,\ldots, j_b\}$.  On  the affine chart  of $P$ defined by $\alpha_{j_b}=1$,  take  local affine coordinates  $\alpha_{j_1},\ldots, \alpha_{j_{b-1}}$,  $z=\alpha_{i_1}$, and
$$\beta_{i_2} = {\alpha_{i_2} \over \alpha_{i_1}}\ ,\ \ldots\ ,\  \beta_{i_a} = {\alpha_{i_a} \over \alpha_{i_1}}\ .$$
 In these coordinates,  $\mathcal{E}_I$ is given by $z=0$.  On this chart, 
$$\Omega_N =  \pm z^{a-1} dz \wedge  d\beta_{i_2}\wedge\ldots \wedge d\beta_{i_a}\wedge d \alpha_{j_1}  \wedge \ldots \wedge d \alpha_{j_{b-1}} \ .$$
Equation $(\ref{generalpoleorder})$ follows by substituting the new variables into $\omega$ and computing 
the order of the pole in $z$. Equation $(\ref{Residueformula})$ follows on taking the residue at $z=0$.
\end{proof}

\subsection{Calculation of the residues} Let $G, \pi_I$ be as above.

\begin{prop} \label{propResIformula} Let $F$  be a flag of  divergent subgraphs  in  $G$, and let $I$ be a strict subset of edges in $G$. The  form $\pi_I^* \omega_{\gamma_F\otimes G/\gamma_F}(s)$ has a simple pole along the exceptional divisor $\mathcal{E}_I$ if and only if $I_{\gamma_F}, I_{G/\gamma_F}$ are divergent, and $I_{G/\gamma_F} \subsetneq G/\gamma_F$. The residue is  
\begin{equation}\label{ResIformula}
\mathrm{Res }_{\mathcal{E}_I} \, \pi^*_I  \omega_{\gamma_F\otimes G/\gamma_F}(s) =
                                \omega_{I_{\gamma_F} \cup I_{G/\gamma_F}} \otimes \omega_{\gamma_F/I\otimes G/(\gamma_F\cup I)}  (s) \ .
                         \end{equation}
\end{prop}

\begin{proof}   Apply lemma $\ref{lemgeneralresidues}$ to the definition $(\ref{omega2def})$ of $\omega_{\gamma_F\otimes G/\gamma_F}(s)$. By    $(\ref{CircRemainder})$ 
and $(\ref{UpsReminder})$, to lowest order terms in the $I$ parameters: 
\begin{eqnarray}\phi_{G/\gamma_F} & =  & \psi_{I_{G/\gamma_F}} \phi_{G/(\gamma_F\cup I)} \ , \nonumber \\
\Upsilon_{\gamma_F ; G/\gamma_F} (s) & =& \psi_{I_{\gamma_F}} \psi_{I_{G/\gamma_F}} \Upsilon_{\gamma_F/I ; G/(\gamma_F \cup I)} (s) \ .\nonumber 
\end{eqnarray} 
By lemma $\ref{lemgeneralresidues}$ and separating out terms we deduce that $ \mathrm{Res }_{\mathcal{E}_I}  \pi_I^* \omega_{\gamma\otimes G/\gamma_F}(s) $ is 
\begin{eqnarray} && {\psi_{I_{G/\gamma_F}}   \over \psi^2_{I_{G/\gamma_F}}   \psi^2_{I_{\gamma_F}} \psi_{I_{G/{\gamma_F}}}}  \,  \Omega_I\otimes    {   \phi_{G/(\gamma_F\cup I)}  \over \psi_{\gamma_F/I}   \psi^2_{G/(\gamma_F\cup I)} \Upsilon_{\gamma_F/I; G/(\gamma_F\cup I)}(s)  } \, \Omega_{G/I}  \nonumber 
  \end{eqnarray} 
  which is exactly  $ \omega_{I_{\gamma_F} \cup I_{G/\gamma_F} } \otimes \omega_{\gamma_F/I \otimes G/(\gamma_F\cup I)}(s)$.
\end{proof}
For any meromorphic algebraic form $\omega$  on $\Pro^{N-1}$, define the total residue to be 
\begin{equation} \label{deftotalresforomega}
\Res \, \omega = \bigoplus_I \Res_{{\mathcal{E}_I}  } \, \pi^*_I \omega\ ,
\end{equation} 
where the sum is over all strict subsets  $I \subsetneq   \{1,\ldots, N\} $.
\begin{prop}  \label{propresomega2} Let $G$ be a connected log-divergent graph with at most log-divergent subgraphs. For any flag $F$ of divergent subgraphs of $G$, 
 the total residue is
$$ \Res \, \omega_{\gamma_F \otimes G/\gamma_F } (s)= (\omega^1 \otimes \omega^{23} ) \circ \mu_{13} (\Delta \otimes \Delta) (\gamma_F\otimes G/\gamma_F) \ ,  $$
where $(\omega^1\otimes \omega^{23}) (x\otimes y \otimes z) = \omega_x \otimes \omega_{y\otimes z}(s)\ . $
\end{prop}

\begin{proof}  By the definition of the total residue, $\Res \, \omega_{\gamma_F\otimes G/\gamma_F} (s) $ is equal to
$$ \sum_I \Res_{\mathcal{E}_I} \, \pi^*_I  \omega_{\gamma_F\otimes G/\gamma_F}(s) =   \sum_{I_{\gamma_F}, I_{G/\gamma_F} \diverg}  \omega_{I_{\gamma_F} \cup I_{G/\gamma_F} } \otimes \omega_{\gamma_F/I_{\gamma_F} \otimes G/(\gamma_F\cup I) }(s) \ ,$$
by proposition \ref{propResIformula}, 
since the  right-hand side vanishes when  $I_{G/\gamma_F}= G/\gamma_F$.
 This can be rewritten 
$$  \sum_{ a\subset \gamma_F \diverg} \sum_{b\subset G/\gamma_F \diverg}  ( \omega^1 \otimes  \omega^{23}) \circ  \mu_{13} ( a \otimes \gamma_F/a \otimes b \otimes G/\gamma_F/ b )$$
which is  $( \omega^1 \otimes  \omega^{23} ) \circ  \mu_{13} (\Delta(\gamma_F) \otimes \Delta(G/\gamma_F))$. \end{proof}

\section{Cancellation of  simple poles} \label{sectCancelpoles}
The proof of convergence of the renormalized Feynman integral follows from a few abstract Hopf-algebra theoretic properties of the residues.
\subsection{The set-up}  Let $H$ be a commutative, graded Hopf algebra over $\Q$, and let $\D$ be a vector space over $k$, a field of characteristic $0$. We think of $\D$ as some  space of differential forms.
Suppose we are given  a  linear map 
$\omega: H \rightarrow \D$
such that $\omega(1)=0$, 
and a map (the  `total  residue')
$$\Res: \D \rightarrow \D \otimes_k \D$$
which is related  to the coproduct $\Delta:H\rightarrow H\otimes_{\Q} H$ by:
\begin{equation} \label{Resomegaisdelta}
\Res \, \omega = (\omega \otimes \omega) \circ \Delta \ .
\end{equation}
In particular, this means that the map $\Res$ is coassociative.
 Now suppose that there exists another linear  map
$$\omegat: H\otimes_{\Q} H \To \D$$
such that, for all $\xi\in H$, 
\begin{eqnarray}\label{axioms}
\mathbf{(1)}: & & \omegat(\xi\otimes 1) = 0   \\
\mathbf{(2)}:& &  \omegat(1\otimes \xi) = \omega(\xi) \nonumber  \\
 \mathbf{(3)}: & &  \Res\, \omegat = (\omega \otimes \omegat) \circ \mu_{13} (\Delta \otimes \Delta) \ . \nonumber
\end{eqnarray}
\begin{defn}ÊIn this situation, we define the renormalized map 
\begin{eqnarray}
\omega^{\ren}:H & \To  & \D  \\
\omega^{\ren} & =  & \omegat \circ R  \nonumber 
\end{eqnarray}
\end{defn}
Since we have by definition
$R(x) = 1 \otimes x $ +  terms of lower order in the coradical filtration, 
then by assumption $\axii$, we can view
$\omega^{\ren} = \omega $ +  lower order terms. The following proposition is the main result for the renormalization of logarithmic singularities.
\begin{prop}  \label{proprenorm} With the assumptions $(\ref{axioms})$, 
$\Res \, \omega^{\ren} =0\ .$
\end{prop}
\begin{proof}
Using properties $\axi-\axiii$ above, and theorem \ref{thmRmainproperty}, we have
\begin{eqnarray}
\Res\, \omega^{\ren}  =  \Res\, \omegat \circ R \nonumber    & \overset{\axiii}{=} &    (\omega \otimes \omegat) \circ \mu_{13} (\Delta \otimes \Delta)\circ R \nonumber \\
&    \overset{\axi}{=} &  (\omega \otimes \omegat) \circ \mu_{13} (\Delta \otimes (\Delta-id\otimes 1))\circ R \nonumber \\
&    \overset{}{=} &  (\omega \otimes \omegat) \circ (1\otimes R) \ .\nonumber 
\end{eqnarray}
The  equality on the third line follows from equation $(\ref{eqnRmainproperty})$. The final  expression vanishes by the property that $\omega(1)=0$, and this  follows
from $\axi$ and $\axii$.
\end{proof}
\begin{rem} Observe that  $\axiii$ implies $(\ref{Resomegaisdelta})$. For $x\in H$, we have 
$$\Res \, \omega(x) \overset{\axii}{=} \Res \, \omegat(1\otimes x) \overset{\axiii}{=} (\omega \otimes \omegat)\circ \mu_{13} (\Delta \otimes \Delta) (1\otimes x)$$
If we write $\Delta \,x= \sum_{(x)} x^{(1)} \otimes x^{(2)}$ using Sweedler's  notation $(\ref{Sweedler})$,  this reduces to
$$\sum_{(x)} (\omega \otimes \omegat)\circ (x^{(1)}\otimes 1 \otimes x^{(2)})  \overset{\axii}{=} \sum_{(x)} \omega(x^{(1)}) \otimes \omega(x^{(2)}) =(\omega\otimes \omega) \circ \Delta(x) \ .$$
 \end{rem}

\section{Blow-ups and  mixed Hodge structures}
\subsection{Hypersurfaces} Fix  a  $1$-scale graph $G$ with $N$ edges.  For any subgraph or quotient graph $\gamma$ of $G$, the graph hypersurface 
 $X_{\gamma}= V(\psi_{\gamma})\subset \Pro^{N-1}$ is defined to be   the zero locus of the graph polynomial
$\psi_{\gamma}$ (using  convention $(\ref{productconvention})$). 
If $F$ is a flag of  divergent  subgraphs in   $G$, define a family of hypersurfaces
$$X^s_{\gamma_F\otimes G/\gamma_F} = V(\Upsilon_{\gamma_F\otimes G/\gamma_F}(s))\subset \Pro^{N-1}\ ,$$
over $\Pro^1$, with coordinate $s$.
If $G$ has at most logarithmically divergent subgraphs,  define  $X^{tot}_{G,s}\subset \Pro^{N-1}$ to be  the union of the graph  hypersurfaces $ X_{\gamma_F}$,  $X_{G/\gamma_F}$,  and   $X^s_{\gamma_F\otimes G/\gamma_F}$, as $F$ ranges over the set of  flags of divergent subgraphs of $G$. By definition $(\ref{omega2def})$, $\omega^{\ren}_{G}(s) \in \Omega^{N-1}(\Pro^{N-1} \backslash X^{tot}_{G,s})$. 
Recall from $(\ref{defnDeltaN})$  that $\DD_{G}$
is the standard coordinate simplex, and write $D_I= L_I\cap \DD_G$ for all $I\subset E(G)$. 
It is known by \cite{BEK}, Proposition 3.1 and Lemma 7.1, that
\begin{equation} \label{XmeetsDI}
X_G \cap D_{I} \neq \emptyset \quad   \Longleftrightarrow \quad   D_I \subset X_G  \quad \Longleftrightarrow  \quad h_I >0 \ .\end{equation}
There is an obvious generalization for the hypersurfaces $X^s_{\gamma_F\otimes G/\gamma_F}$.

\begin{lem} Let $G,F $ be as above, and let $s>0$. The following are equivalent:
\begin{enumerate}
\item $X^s_{\gamma_F\otimes G/\gamma_F} \cap D_{I} \neq \emptyset $.
\item $D_I \subset X^s_{\gamma_F\otimes G/\gamma_F}  .$
\item $ h_{I_{\gamma_F}} +h_{I_{G/\gamma_F}} >0 \hbox{ or } \varepsilon_{\gamma_F}(I) = \varepsilon_{G/\gamma_F}(I) =1$.
\item The subgraph of $c(\gamma_F,G/\gamma_F)$ defined by the edges $I$ contains a loop.
\end{enumerate}
In particular,  $X_{\gamma_F}\cap D_I \subset X^s_{\gamma_F\otimes G/\gamma_F}\cap D_I$  and  $X_{G/\gamma_F} \cap D_I \subset X^s_{\gamma_F\otimes G/\gamma_F} \cap D_I$.
\end{lem}
\begin{proof}  Since $\Upsilon_{\gamma_F;G/\gamma_F}(s)$ has positive coefficients, $X^s_{\gamma_F\otimes G/\gamma_F}$ meets $D_I$ if and only if
$L_I \subset X^s_{\gamma_F\otimes G/\gamma_F}$. By $(\ref{114star})$, this occurs if and only if $  h_{I_{\gamma_F}} +h_{I_{G/\gamma_F}} + \varepsilon_{\gamma_F}(I) \varepsilon_{G/\gamma_F}(I)$ is positive, which is in turn equivalent to $(4)$.  So $(1)$-$(4)$ are equivalent. The last part is obvious by $(3)$ and $(\ref{XmeetsDI})$.
\end{proof} 

Thus the intersections of $X_{G,s}^{tot}$ with $\DD_G$ are contained in the union of the intersections of each hypersurface $X^s_{\gamma\otimes G/\gamma}$ with $\DD_G$. Hereafter, let $s>0$.

\subsection{The blowup}  \label{sectblowup} Define the following  set of  subsets of edges of $G$:
$$G^{div} = \{ I \subsetneq E(G): I \hbox{ minimal  such that  for some flag } F  \hbox{ of divergent}$$
$$ \hbox{  subgraphs}, 
 I\subset c(\gamma_F, {G/\gamma_F}) \hbox{ contains a loop} \}$$
The set  $\{L_I: I \in G^{div}\}$ is therefore the set of maximal linear coordinate spaces whose real points are contained in $X^{tot}_{G,s}\cap \DD_G$. Define
$\mathcal{F}_G$ to be the set of all intersections $L_{I_1\cup \ldots \cup I_k}$, with $I_j\in G^{div}$, of such coordinate spaces. Following standard practice (\cite{BEK}, \S7), one blows up the elements in $\mathcal{F}_G$ in strictly increasing order of codimension.  One knows  that the space $P_G$ obtained in this way is well-defined.  If  $ \pi:P_G \rightarrow \Pro^{N-1}$ denotes the blow-up, 
let $Y^{tot}_G$ be the strict transform of $X^{tot}_G$ and let $B$ be the total inverse image of the linear spaces $L_I$ under $\pi$. Then $B$ is a union of the strict transforms of the  coordinate hyperplanes $\alpha_i=0$  with exceptional  components
$\mathcal{E}_I$ for each $I \in G^{div},$
where $\mathcal{E}_I$ is the strict transform of $L_I$. 

\begin{prop} \label{propnc} The divisor $B$ is normal crossing, and  the strict transform of $\DD_G$ in $P_G$ does not meet $Y^{tot}_G$. No non-empty intersection of the components of $B$ is contained in $Y^{tot}_G$.
 \end{prop}
\begin{proof} This can be proved in an identical manner to \cite{BEK}, Proposition 7.3 (ii), (iii) (only part (i) of loc. cit. requires the assumption that the graph be primitive). \end{proof}

\begin{example} Let $G$ be the graph $\gamma_l$ in \S\ref{sectexoverlap}. It has one divergent subgraph $\gamma=\{3,4\}$. 
 The circular join of $G/\gamma$ and $\gamma$ is the banana graph with four edges, with graph polynomial
$\alpha_1\alpha_2\alpha_3+\alpha_1\alpha_2\alpha_4 + \alpha_1\alpha_3\alpha_4+ \alpha_2\alpha_3\alpha_4\ ,$
 which meets $\DD_G$ along all $D_I$ for any $I\subset \{1,2,3,4\}$, with $|I|\geq 2$.  It follows that  $G^{div} = \{ \{i,j\}, 1\leq i<j\leq 4\}$, and $\mathcal{F}_G$ consists of all divisors
 $L_I$ where $|I|\geq 2$.  
 \end{example}
 
\begin{rem} In general, one could simply blow-up all  the faces of codimension $\geq 2$  for good measure. In this case the
space $P_G$ admits an action of the symmetric group on $N$ letters which permutes the coordinates,  and $B$ has the structure of a permutohedron with $N!$ vertices. In particular, there are canonical coordinates on $P_G$ in the neighbourhood of each vertex indexed by  a $\tau \in \Sigma_N$, given by 
$$\alpha^{\tau}_1 = \alpha_{\tau(1)}\ ,\  \alpha^{\tau}_2 = {\alpha_{\tau(2)} \over \alpha_{\tau(1)}}\ ,\  \ldots\ ,\  \alpha^{\tau}_N = {\alpha_{\tau(N)} \over \alpha_{\tau(N-1)}} \ .$$
These are precisely the local coordinates in the  sector decomposition  of K. Hepp.
\end{rem}

\subsection{The renormalized differential form} We  now apply the formalism of \S\ref{sectCancelpoles}. Let $G$ be a labelled single-scale overall log-divergent graph with at most log-divergent subgraphs and let $H=H_G$ be the Hopf algebra obtained by fixing single-scale
structures on its sub and co-graphs  (\S\ref{sectHopf1scale}). For every  non-empty set of edges $I\subset E(G)$,  let  $\Pro^I$ denote the projective space whose coordinates
are the Schwinger parameters of $I$, and let  $\D^I\subset \Omega^{|I|-1}_{\Q(\alpha_i: i\in I)/\Q}$ denote the  subspace of  regular forms of  homogeneous degree $0$ on a Zariski open subset of $\Pro^{I}$ which are defined over $\Q$. Let
\begin{equation}\label{Dformsdef}
\D=\bigoplus_{\emptyset \neq I\subset E(G)} \D^I
\end{equation}
 As in definition $(\ref{omega2def})$, set
\begin{eqnarray}
\omegat : H_G \otimes_{\Q} H_G &\rightarrow & \D^{E(G)}\subset \D  \\
{\gamma\otimes \Gamma} & \mapsto & \omega_{\gamma\otimes \Gamma}(s)  \ . \nonumber 
\end{eqnarray}
and define $\omega: H_G\rightarrow \D$ by $\omega(G)=\omegat(1\otimes G)=\Omega_G \psi_G^{-2} $. 
The  total residue   $\Res: \, \D \rightarrow \D\otimes_{\Q} \D$ is given by $(\ref{deftotalresforomega})$, i.e.,  the sum of the   residues along  all exceptional divisors $\mathcal{E}_I$ in the blow-ups of $\Pro^{E(G)}$ along $L_I$.
Let 
$$\pi: P_G \rightarrow \Pro^{E(G)}$$
denote the blow-up constructed in \S\ref{sectblowup}.
\begin{thm}  \label{thmconv2} If $s>0$, the form 
$\pi^{*} \omega^{\ren}_G(s)$ has no poles along $B$. 
\end{thm}
\begin{proof}
It is enough (by, for example, the proof of proposition 7.3 in \cite{BEK}) to show that  $\pi_I^*\omega^{\ren}_G(s)$ has no poles along $\mathcal{E}_I$, where $\pi_I : P^I \rightarrow \Pro^{E(G)}$ is the blow-up of a single
coordinate hyperplane $L_I$.  By proposition \ref{propResIformula}, it has at most simple poles. The residue is zero by    proposition  \ref{proprenorm},   since properties $\axi$ and $\axii$ hold by $(\ref{omega1tens})$, and property $\axiii$ holds by proposition \ref{propresomega2}.  
\end{proof}

\begin{cor} \label{corconvergence} If $s>0$, the renormalized Feynman integral   converges:
$$f_{G}(s) = \int_{\DD_G} \omega^{\ren}_G(s)<\infty \ .$$ 
\end{cor} 
\begin{proof}  Pull back the integral  to $P_G$. Thus $f_G(s)$ is the integral of $\pi^{*} \omega^{\ren}(s)$ over the strict transform of $\DD_G$, a compact polytope whose boundary is contained in $B$.  By the previous theorem, the poles of $\pi^{*} \omega^{\ren}(s)$  are contained in $Y^{tot}_{G,s}$, which does not meet the new domain of integration by proposition \ref{propnc}. Thus $f_G(s)$ is the  integral of a continuous function on a compact domain, and is therefore bounded.
\end{proof}

\subsection{The tangent mixed Hodge structure} 
We can therefore define the \emph{tangent mixed Hodge structure of G} as follows. Writing $Y^{tot}_G$ for $Y^{tot}_{G,1}$, let 
\begin{equation}\label{motdef}
 \mathrm{Mot}(G)=H^{N-1}(P_G \backslash Y^{tot}_G, B \backslash B\cap Y^{tot}_G)\ .\end{equation}
The renormalized form  $\omega_G^{\ren}(1)$   defines a cohomology class
$$[\pi^*(\omega_G^{\ren}(1))] \in H_{DR}^{N-1}(P_G\backslash Y^{tot}_G,  B\backslash B\cap Y^{tot}_G)\ ,$$
and the strict transform of the simplex $\Delta_N$ defines a relative homology class
$$[\pi^{-1}(\Delta_N)] \in  W^0 H_{B,N-1}(P_G\backslash Y^{tot}_G, B\backslash B\cap Y^{tot}_G)\ .$$

\begin{cor} The lowest log term $f_G$ of the  renormalized Feynman  integral of $G$  
 defines a period of the tangent mixed Hodge structure $\mathrm{Mot}(G)$. 
\end{cor}
It follows from the renormalization group equations that $f_G(s)$ is a polynomial in $\log s$, whose coefficients
are (products of) periods of mixed Hodge structures $(\ref{motdef})$.

\section{The Renormalization group}
We prove that the  Feynman rules respect  the renormalization group structure.

\subsection{Renormalization group equations} \label{sectrenormgroup}
It is convenient to define a new set of differential forms as follows. Let $\lambda$ denote any parameter, and set 
\begin{eqnarray}
\pomega_G(\lambda) &  = & {\phi_G \over \psi_G}  {\lambda \over (\psi_{G} \, \lambda + \phi_{G}  )^2 } \,  \prod_{e\in E(G)}  d\alpha_e \ ,   \\
\pomega_{\gamma\otimes \Gamma}(\lambda)&  = & {\phi_{\Gamma} \over \psi_{\Gamma}}  {\lambda \over ( \psi_{\gamma  \cup \Gamma} \, \lambda +  \phi_{\gamma\cup \Gamma} )^2 } \,  \prod_{e\in E(\gamma) \cup E(\Gamma)}  d\alpha_e \ ,
\end{eqnarray}
where $G$ and $\gamma, \Gamma$ are labelled single-scale graphs with disjoint labels.   Denote the corresponding renormalized form by
$\pomega^{\ren}_G (\lambda) = \pomega(\lambda) \circ R(G)$.
 The following lemma enables us to lift the domain of integration to one dimension higher.
\begin{lem} \label{lemminormonotone} Let $G$ be single-scale as above. For all $\lambda >0$ we have :
\begin{equation} \label{liftedintegral} 
\int_{\DD_{G}} \omega^{\ren}_{G}(1)=  \int_{[0,\infty]^{E(G)}} \pomega^{\ren}_G ( \lambda)\ .
\end{equation}
\end{lem}

\begin{proof} We first show that the right-hand side is convergent. For this,  it is a simple matter to check that    $\pomega_{1\otimes G}(\lambda)=\pomega_G(\lambda)$, and
 $\pomega_{G\otimes 1}(\lambda)=0$. By a similar calculation to proposition \ref{propresomega2} (with the small difference that we work in affine rather than projective space) we find that the total residue satisfies 
 $$ \Res  \, \pomega(\lambda) =(\pomega(\lambda) \otimes \pomega(\lambda) ) \circ ( \Delta \otimes \Delta) \ ,$$
and therefore by the general set-up of \S \ref{sectCancelpoles}, the integral is convergent. 
 The proof of the  lemma uses the fact that for any  $A,B>0$, we have  
\begin{equation} \label{mmid} {1\over A B } = \int_0^\infty { d\lambda  \over (A\, \lambda+B)^2  } \ .
\end{equation}
Now let $\DD^{\varepsilon}_G = \{ (\alpha_1:\ldots : \alpha_{|E(G)|})\in \Pro^{|E(G)|-1}: \alpha_i \geq  |\varepsilon|\}$,  and write  the left hand side of   $(\ref{liftedintegral})$ as a limit
 $$\int_{\DD_{G}} \omega^{\ren}_{G}(1)= \lim_{\varepsilon\rightarrow 0^+} \int_{\DD_G^{\varepsilon}}  \omega^{\ren}_{G}(1)\ .$$
 Writing   the preparation map  
 $R (G)=1 \otimes G + \sum_{\gamma} a_{\gamma} \, \gamma \otimes G / \gamma, $
 we obtain 
 $$\int_{\DD^{\varepsilon}_{G}} \omega^{\ren}_{G}(1)= \int_{\DD^{\varepsilon}_{G}}   {\Omega_G\over \psi_{G}^2} + \sum_{\gamma} a_{\gamma} \int_{\DD^{\varepsilon}_{G}}  {\phi_{G/ \gamma}\, \Omega_{G} \over \psi_{\gamma} \psi_{G/\gamma}^2 \phi_{\gamma \cup G/\gamma } } \ .  $$ 
 Since  $\DD^{\varepsilon}_G$ does not meet the coordinate hyperplanes,  the graph polynomials are strictly positive on $\DD^{\varepsilon}_G$ and  we can apply   $(\ref{mmid})$ with $(A,B)=(\psi_G,\phi_G)$ to the first factor  and with 
$(A,B)= ( \psi_{\gamma\cup G/\gamma},  \phi_{\gamma\cup G/\gamma})$ to the factors in the sum.  This gives 
$$ \int_{\DD^{\varepsilon}_{G}}  \int_0^{\infty}  
 {\phi_G \over \psi_G}  { \Omega_G  \, d\lambda \over (\psi_{G} \, \lambda + \phi_{G}  )^2 } + \sum_{\gamma} a_{\gamma} \int_{\DD^{\varepsilon}_G} \int_0^{\infty} 
 {\phi_{G/\gamma} \over \psi_{G/\gamma}}  {\Omega_{G}  \, d\lambda \over ( \psi_{\gamma  \cup G/\gamma} \, \lambda +\phi_{\gamma \cup G/\gamma} )^2 } \ . $$
 We can take the limit as $\varepsilon \rightarrow 0$:
 $$ \int_{\DD_{G}}  \int_0^{\infty}  \Big(
 {\phi_G \over \psi_G}   {1 \over (\psi_{G} \, \lambda + \phi_{G}  )^2 } + \sum_{\gamma} a_{\gamma} 
 {\phi_{G/\gamma} \over \psi_{G/\gamma}}  {1 \over (\psi_{\gamma  \cup  G/\gamma} \, \lambda +  \phi_{\gamma\cup G/\gamma} )^2 }\Big) \,  \Omega_G  \, d\lambda$$
 since the right-hand side of the integrand  (viewed as an integral on $\Pro^{|E(G)|-1}\times \A^1$) is convergent, by the usual arguments. 
 The integrand  defines a projective integral on $\Pro^{|E(G)|}$ with projective coordinates $(\alpha_1: \ldots: \alpha_{E(G)} : \lambda)$ if we replace $\Omega_G \, d\lambda$   with 
$\Omega_G\, d\lambda + \lambda \prod_i d\alpha_i$. 
 Restricting to 
any  hyperplane  $\lambda=$ constant  gives  an affine integral which is precisely the right-hand side of $(\ref{liftedintegral})$.
\end{proof}

 \begin{prop}  \label{prop1} Let $G_1,G_2$ denote single-scale graphs with at most logarithmic subdivergences, and denote their  images under the preparation map $R$  by
$$R(G_i) = \sum_{\gamma_i} a_{\gamma_i} \gamma_i \otimes G_i /\gamma_i\ ,  $$
where  $a_{\gamma_i} = \pm 1$, the sum is over all flags of (possibly empty) divergent subgraphs  $\gamma_i\subset G_i$ $(\ref{gammaflag})$  and  $i=1, 2$.  If $G_1,G_2$ have distinct labels and $s>0$ then we have
\begin{eqnarray}  \label{renomeqmain}
\int_{\DD_{G_1}} \omega^{\ren}_{G_1}(1) \times \int_{\DD_{G_2}} \omega^{\ren}_{G_2}(s)\quad  =   \qquad\qquad\qquad\quad&& \\ 
\qquad \qquad \qquad  \int_{\DD_{G_1\cup G_2}}  \sum_{\gamma_i\subset G_i} a_{\gamma_1} a_{\gamma_2} {\phi_{G_1/ \gamma_1} \phi_{G_2/ \gamma_2} \over \psi_{G_1/\gamma_1}
 \psi_{G_2/\gamma_2}}
 {s \, \Omega_{G_1\cup G_2} \over \big( \Upsilon_{\gamma_1\cup \gamma_2;G_2/\gamma_2}(s)\big)^2}\ ,   \label{renomeqmain2}
 \end{eqnarray}
where all integrals are convergent.
\end{prop}
\begin{proof} The two integrals in the first line of  $(\ref{renomeqmain})$  are convergent by theorem \ref{thmconv}. To check the convergence of $(\ref{renomeqmain2})$,
write
\begin{equation}\label{omegafour}
\omega^{(4)}_{\gamma_1\otimes \gamma_2 \otimes \gamma_3 \otimes \gamma_4}(s)   = {\phi_{\gamma_2}\over \psi_{\gamma_2}} {\phi_{\gamma_4} \over \psi_{\gamma_4}} 
{s\,  \Omega_{\gamma_1\cup \ldots \cup \gamma_4} \over \Upsilon^2_{\gamma_1\cup \gamma_2 \cup \gamma_3; \gamma_4}(s) } \end{equation}
and verify, by a similar computation to lemma \ref{lemgeneralresidues} that $\omega^{(4)}$ has at most  simple poles along exceptional divisors $\mathcal{E}_I$ 
and that its total residue is
$$\Res\, \omega^{(4)}(s)  = (\omega \otimes \omega^{(4)}) \circ \mu_{1357}  \circ (\Delta \otimes \Delta \otimes \Delta \otimes \Delta)  \ .
$$
where $\mu_{1357}(\gamma_1\otimes \ldots \otimes \gamma_8)= \gamma_1\gamma_3\gamma_5\gamma_7\otimes \gamma_2\otimes \gamma_4\otimes \gamma_6 \otimes  \gamma_8$.
The integrand   in $(\ref{renomeqmain2})$ is given by $\omega^{(4)}(s) \circ (R \otimes R) (G_1 \otimes G_2)$.
We have
\begin{eqnarray}
 \Res\, \omega^{(4)} \circ (R \otimes R)  & = &  (\omega \otimes \omega^{(4)}) \circ \mu_{1357}  \circ (\Delta \otimes \Delta \otimes \Delta \otimes \Delta) \circ (R\otimes R) \nonumber \\
 &=&  (\omega \otimes \omega^{(4)}) \circ  \mu_{14} \big( ( \mu_{13} \circ (\Delta \otimes \Delta)\circ R) \otimes  ( \mu_{13} \circ (\Delta \otimes \Delta)\circ R)\big) \nonumber \\
& = & (\omega \otimes \omega^{(4)}) \circ  \mu_{14} (1\otimes R \otimes 1 \otimes R) \nonumber \\
 & =  &   \omega(1) \otimes \omega^{(4)} (R  \otimes R) \quad =  \quad 0 \nonumber 
\end{eqnarray}
 The third equality follows from the fact that $\omega^{(4)}(\gamma_1\otimes \gamma_2\otimes \gamma_3 \otimes 1 ) = 0$ and theorem \ref{thmRmainproperty}, and the  vanishing of the last line follows from  $\omega(1)=0$.  Using the positivity of graph polynomials (since $s>0$),
 we prove as in corollary \ref{corconvergence}  that $(\ref{renomeqmain2})$ converges.
   For the proof of the main identity,   
write
$$\int_{\DD_{G_2}} \omega^{\ren}_{G_2}(s)= \lim_{\varepsilon\rightarrow 0} \Big( \int_{\DD^{\varepsilon}_{G_2}}  \omega_{G_2} + \sum_{\emptyset \neq\gamma_2 \subset G_2} a_{\gamma_2} \int_{\DD^{\varepsilon}_{G_2}}  \omega_{\gamma_2\otimes G_2 /\gamma_2} (s) \Big) \ ,$$
where $\DD^{\varepsilon}_{G_2} = \{ (\alpha_1:\ldots : \alpha_{|E(G_2)|})\in \Pro^{|E(G_2)|-1}: \alpha_i \geq  |\varepsilon|\}$.  
Multiply   this expression through by  
$$\int_{\DD_{G_1}} \omega^{\ren}_{G_1} (1) = \int_{[0,\infty]^{E(G_1)}} \pomega^{\ren}_{G_1}(\lambda)\ ,$$
which holds by lemma \ref{lemminormonotone} for any $\lambda >0$. Thus, for any choice of $\lambda_{\gamma_2}$'s, we have
$$ (\ref{renomeqmain})= \lim_{\varepsilon\rightarrow 0} \Big( \int_{X_{\varepsilon}}  \pomega^{\ren}_{G_1}(\lambda_{G_2}) \wedge\omega_{G_2} + \sum_{\gamma_2 \subset G_2} a_{\gamma_2} \int_{X_{\varepsilon}}  \pomega^{\ren}_{G_1}(\lambda_{\gamma_2}) \wedge  \omega_{\gamma_2\otimes G_2 /\gamma_2} (s) \Big) \ ,$$
where the domain of integration $X_{\varepsilon}=[0,\infty]^{E(G_1)}\times \DD^{\varepsilon}_{G_2}$. Now substitute
$$\lambda_{G_2}  = { s \, \phi_{G_2} \over  \psi_{G_2}}\quad \hbox{ and } \quad \lambda_{\gamma_2}  = {\Upsilon_{\gamma_2; G_2/\gamma_2 } (s)\over  \psi_{\gamma_2\cup G_2/\gamma_2}} \ .$$
Directly from the definitions, one verifies that
$$ \pomega_{\gamma_1 \otimes G_1/\gamma_1}(\lambda_{\gamma_2}) \wedge \omega_{\gamma_2\otimes G_2/\gamma_2} (s) = 
  {\phi_{G_1/ \gamma_1} \phi_{G_2/ \gamma_2} \over \psi_{G_1/\gamma_1}
 \psi_{G_2/\gamma_2}}
 {s \,  \prod_{e\in E(G_1)} d\alpha_e \wedge \Omega_{G_2}\over \big( \Upsilon_{\gamma_1\cup \gamma_2;G_2/\gamma_2}(s)\big)^2}\ .$$
By a similar calculation involving $\pomega_{G_1}(\lambda)$ and so on, we conclude that 
$$ (\ref{renomeqmain})= \lim_{\varepsilon\rightarrow 0} \Big( \int_{X_{\varepsilon}}    \sum_{\gamma_i\subset G_i} a_{\gamma_1} a_{\gamma_2}\, 
 {\phi_{G_1/ \gamma_1} \phi_{G_2/ \gamma_2} \over \psi_{G_1/\gamma_1}
 \psi_{G_2/\gamma_2}}
 {s \,  \prod_{e\in E(G_1)} d\alpha_e \wedge \Omega_{G_2}\over \big( \Upsilon_{\gamma_1\cup \gamma_2;G_2/\gamma_2}(s)\big)^2}  \Big)  \ .  $$
Projectivising  the integral on the right-hand side gives $(\ref{renomeqmain2})$.  Therefore the limit  $\varepsilon=0$ is finite and  equal to $(\ref{renomeqmain2})$. 
\end{proof}

\begin{defn} Let $\gamma, \Gamma$ be  labelled single-scale graphs with disjoint labels. Let
\begin{equation}  \label{omegabardef}
\overline{\omega}_{\gamma\otimes  \Gamma} (s) =  s {\partial \over \partial s}\, \omega_{\gamma \otimes \Gamma}(s) ={\phi_{\gamma}  \phi_{\Gamma}\over \psi_{\gamma} \psi_{\Gamma}} 
{s \over \big( s  \psi_{\gamma}\phi_{\Gamma} + \phi_{\gamma} \psi_{\Gamma}\big)^2}\ .
\end{equation}
\end{defn}
 In particular,  $\overline{\omega}_{1\otimes  \Gamma} (s) =0$.
 The following lemma is  immediate.
\begin{lem} \label{lemcjoinproperty} Let $\gamma_1,\ldots, \gamma_n$ be connected 1-scale graphs, and $\gamma=\cup_i \gamma_i$. Then 
\begin{equation} \label{Phiadditive}
{\phi_{\gamma}  \over \psi_{\gamma}} = \sum_{i=1}^n {\phi_{\gamma_i} \over \psi_{\gamma_i}} \ . 
\end{equation}
\end{lem}

\begin{prop} \label{propprerenormgroup}  Let $\Gamma$ be a connected single-scale graph with at most logarithmic subdivergences as above, and write $\Delta' \Gamma =\sum_{\gamma} \gamma \otimes \Gamma/\gamma$.  Then 
\begin{equation}  \label{propprenormgroupeq}
\int_{\DD_{\Gamma}} \overline{\omega}^{\ren}_{\Gamma}(s) =  \sum_{\gamma\subset \Gamma} \, \int_{\DD_{\gamma}} \omega_{\gamma}^{\ren} (1)  \times \int_{\DD_{\Gamma/\gamma}}  \omega_{\Gamma/\gamma}^{\ren} (s) \ . 
\end{equation} 
\end{prop}
\begin{proof}  If $a_1,\ldots, a_n\in H$ are graphs, then lemma  \ref{lemcjoinproperty}  implies that
\begin{equation} \label{Jsum} \sum_{1\leq  i \leq n}  \omega^{(4)}_{a_1\ldots a_{i-1} \otimes a_i \otimes a_{i+1}\ldots a_n \otimes b}(s) = \overline{\omega}_{a_1\ldots a_n \otimes b}(s)\  
\end{equation} 
where empty products  are defined to be $1$, e.g., $\omega^{(4)}_{1\otimes a_1 \otimes 1\otimes b} = \overline{\omega}_{a_1\otimes b}(s)$.
Viewed as an equality of elements in $\mathrm{Hom}(H\otimes_{\Q} H, \R)$, proposition \ref{prop1} states that:
$$\mu \circ (f_{G_1}(1)\otimes f_{G_2}(s)) = \int_{\DD_{G_1\cup G_2}} \omega^{(4)}(s)  \circ (R\otimes R) \circ  (G_1\otimes G_2) \ . $$ 
Apply both sides  to $\Delta'(\Gamma)\in H\otimes_{\Q} H$. The left-hand side equates to the right-hand side of 
$(\ref{propprenormgroupeq})$. 
By $(\ref{Jsum})$ and  lemma \ref{proponR}, 
the integrand on the right-hand side is
$$\omega^{(4)}(s) \circ (R\otimes R) \circ \Delta' (\Gamma)=\overline{\omega}(s) \circ R(\Gamma)=\overline{\omega}^{\ren}_{\Gamma}(s)\ ,$$
which therefore gives equation  $(\ref{propprenormgroupeq})$.
\end{proof}
\subsection{Group equations with angular dependence}
 The proof of the group equations essentially only uses $(\ref{Phiadditive})$ together with some formal properties of the map $R$. As a result the proof of the renormalization group equations also goes through
in the case when $\phi$ has non-trivial angular dependencies.
 
\section{Quadratic subdivergences}   \label{sectQuadratic}
We treat the general case of  graphs in massless $\phi^4$ with arbitrary subdivergences.

 \subsection{Single-scale graphs with quadratic subdivergences} 
 Let $G$ be a connected graph in $\phi^4$ theory, and let $q\subsetneq G$ be a connected, 1PI  subgraph with $sd(q)=1$.  
 The crucial feature of quadratic subdivergences is that, by remark \ref{remvertexnumbers},  they have exactly two $3$-valent vertices, 
 and   two external edges, or \emph{connectors}, which we denote by  $\{e_q,f_q\}\in G\backslash q$.  
 Thus $q$ inherits a unique single-scale structure.

   We define  a \emph{squashing} of $G$ 
  to be the choice, for every   $q$ a 1PI connected quadratic subgraph, of an ordered pair of connectors $(e_q,f_q)$ such that $f_{q_1}\neq f_{q_2}$ for all $q_1\neq  q_2$. It is easy to verify (e.g. as a consequence of corollary \ref{qdontoverlap} below) that  such a structure always exists.  We define the \emph{squashed} graph  by:
   \begin{equation}
  \label{squashdef}
 \overline{G} = G/ \cup_{q} f_q
  \end{equation}
   Note that $\overline{G}$ has vertices of arbitrary degree. The image $\overline{q}$ in $\overline{G}$ of a 
 quadratic subgraph $q$ in $G$ can have $0$ or $1$ connectors.   
 
 \begin{defn} A  \emph{single-scale graph with quadratic subdivergences}  is a labelled graph $G$ with a choice of squashing. A \emph{tadpole-free subdivergence} of $G$
 is a subgraph $\gamma\subsetneq G$ such that $\gamma$ is 1PI,  divergent and 
 \begin{eqnarray} \label{tadpolefree}
 \varepsilon_q(\gamma) = 0 \hbox{ for all  }  q\supsetneq \gamma \ ,
 \end{eqnarray}
where $q$ ranges over the set of  quadratic ($sd(q)=1$) connected 1PI subgraphs of $G$. 
  \end{defn} 
    The condition $(\ref{tadpolefree})$ ensures that in the cograph $G/\gamma$ the connectors of $q$ remain unjoined (no tadpoles are spawned). 
 To construct a Hopf algebra from this, we proceed as before. Let $G$ be a $1$-scale graph with quadratic subdivergences and 
 suppose that for every divergent tadpole-free subgraph $\Gamma$ of $G$ there is a choice of two distinguished connectors giving it a $1$-scale 
structure, in a such a way that for all $\gamma \subset \Gamma$ divergent and tadpole-free, we have
\begin{equation}\label{noepsilonsquadcase}
\varepsilon_{\Gamma}(\gamma)=0 \ .
\end{equation}
This guarantees that the cographs $\Gamma/\gamma$ have a well-defined single-scale structure.
 
  For any such  labelled  single-scale graph $G$, let   $H_G$ denote the (loop-number graded) $\Q$-vector space spanned by its tadpole-free divergent  subgraphs and  their cographs. 
We therefore have a well-defined map   
   \begin{eqnarray} 
 \Delta  : H_G &  \To &  H_G\otimes_{\Q} H_G \nonumber \\
 \Gamma & \mapsto & \sum_{\gamma} \gamma \otimes \Gamma/\gamma
 \end{eqnarray} 
 where the sum is over all tadpole-free subdivergences $(\ref{tadpolefree})$.  It is well known that this map is  coassociative and hence defines a Hopf algebra.
  
 \subsection{General massless Feynman rules} 
  With $G$ as above, and  $\gamma$,    $\Gamma$ in $H_G$, where $\gamma$ is not necessarily connected, define differential forms as follows:
 \begin{eqnarray} \label{quadforms}
 \omega_{G} & = & \Big( \prod_{q\subseteq G} {\phi_{\overline{q}}  \over \psi_{\overline{q}}}\Big) \, {\Omega_{\overline{G}}  \over \psi_{\overline{G}}^2}  \ ,  \\ 
 \omega_{\gamma \otimes \Gamma }(s) & = & \Big( \prod_{q\subseteq \gamma} {\phi_{\overline{q}}  \over \psi_{\overline{q}}}\Big)  \Big( \prod_{q\subseteq \Gamma} {\phi_{\overline{q}}  \over \psi_{\overline{q}}}\Big) { s\, \phi_{\overline{\Gamma}}  \over  \psi_{\overline{\gamma}} \psi_{\overline{\Gamma}}^2 \Upsilon_{\overline{\gamma};\overline{\Gamma}}(s)   } \,  \Omega_{\overline{\gamma} \cup \overline{\Gamma}}   \ ,    \nonumber 
    \end{eqnarray}
    where the products are over all connected 1PI quadratically-divergent subgraphs $q$.
    Note that  the forms are  in the  edge variables of  $\overline{\gamma}, \overline{\Gamma}$, and not $\gamma, \Gamma$. 
  Obviously,
  \begin{equation}\label{1tensgammaforq}
  \omega_{\gamma\otimes 1}(s) = 0 \quad  \hbox{ and } \quad  \omega_{1\otimes \Gamma}(s) = \omega_{\Gamma}\ .
  \end{equation}
 Define  renormalized Feynman rules as follows.
 Given $\Gamma \in H_G$, write as usual
 $$\Delta(\Gamma) = 1 \otimes \Gamma + \sum_{i=1}^n a_i \, \gamma_i \otimes \Gamma/\gamma_i \in H_G \otimes_{\Q} H_G\ ,$$
 and let $\omega^{\ren}=\omega \circ R$, i.e.,  define 
  $$\omega^{\ren}_{\Gamma} = \omega_{\Gamma} + \sum_{i=1}^n a_i  \, \omega_{\gamma_i \otimes \Gamma/\gamma_i}(s)\ .$$
 The renormalized Feynman integral is then 
 \begin{equation}\label{frenq}
 f_{\Gamma}(s)=\int_{\DD_{\overline{\Gamma}}} \omega^{\ren}_{\Gamma}(s)\ .\end{equation}
 We show below that $(\ref{frenq})$ converges and satisfies renormalization group equations (theorem \ref{thmquadconv} and \S\ref{sectrenormgroupquad}).
 As previously, we write $f_{\Gamma}= f_{\Gamma}(1)$ for the lowest log term. 
 \begin{example} Consider the graphs with nested quadratic subdivergences below.
 
\begin{center}
\fcolorbox{white}{white}{
  \begin{picture}(342,100) (367,1)
    \SetWidth{0.5}
    \SetColor{Black}
    \Text(444,99)[lb]{{\Black{$5$}}}
    \SetWidth{1.0}
    \Arc(444,41)(39,270,630)
    \Line(405,41)(483,41)
    \Arc(444,80)(16,270,630)
    \Arc(570,41)(24,270,630)
    \Arc(658,41)(24,270,630)
    \Line(546,41)(594,41)
    \Line(634,41)(683,41)
    \Vertex(658,65){2}
    \Text(405,65)[lb]{\small{\Black{$1$}}}
    \Text(484,65)[lb]{\small{\Black{$2$}}}
    \Text(444,43)[lb]{\small{\Black{$3$}}}
    \Text(444,8)[lb]{\small{\Black{$4$}}}
    \Text(444,82)[lb]{\small{\Black{$6$}}}
    \Text(444,67)[lb]{\small{\Black{$7$}}}
    \Text(380,33)[lb]{\Large{\Black{$\Gamma$}}}
    \Text(532,33)[lb]{{\Black{$\gamma$}}}
    \Text(570,67)[lb]{\small{\Black{$5$}}}
    \Text(570,44)[lb]{\small{\Black{$6$}}}
    \Text(570,22)[lb]{\small{\Black{$7$}}}
    \Text(636,60)[lb]{\small{\Black{$1$}}}
    \Text(680,60)[lb]{\small{\Black{$2$}}}
    \Text(658,43)[lb]{\small{\Black{$3$}}}
    \Text(658,22)[lb]{\small{\Black{$4$}}}
    \Text(611,33)[lb]{{\Black{$\Gamma/\gamma$}}}
  \end{picture}
}
\end{center}

 A squashing of this graph is determined by choosing one of the connectors of the subgraph $\gamma$. Let it be edge $1$. The squashed 
 graphs $\overline{\Gamma}, \overline{\gamma}, \overline{\Gamma}/\overline{\gamma}$ are obtained from the graphs above by contracting edge $1$.
 The renormalized integral is:
 $$f_{\Gamma}(s) = \int_{\DD_{\overline{\Gamma}}}\omega^{\ren}_{\Gamma}(s) = \int_{\DD_{\overline{\Gamma}}}\Big( {\phi_{\overline{\gamma}} \over \psi_{\overline{\gamma}} }  {\phi_{\overline{\Gamma}} \over \psi^3_{\overline{\Gamma}} }   -   
  {\phi_{\overline{\gamma}} \over \psi^2_{\overline{\gamma}} }   {\phi^2_{\overline{\Gamma}/\overline{\gamma}} \over \psi^3_{\overline{\Gamma}/\overline{\gamma}} } {s \over \Upsilon_{\overline{\gamma};\overline{\Gamma}/\overline{\gamma}}(s)}\Big) \Omega_{\Gamma} $$
 where $\Upsilon_{\overline{\gamma};\overline{\Gamma}/\overline{\gamma}}(s)= \psi_{\overline{\gamma}}  \phi_{\overline{\Gamma}/ \overline{\gamma}}\, s+ \psi_{\overline{\Gamma}/ \overline{\gamma}}  \phi_{\overline{\gamma}}$, and all other polynomials are:
 \begin{eqnarray}
  \psi_{\overline{\gamma}} = \alpha_5\alpha_6+\alpha_5\alpha_7+\alpha_6\alpha_7 \ , & & \phi_{\overline{\gamma}} =\alpha_5\alpha_6\alpha_7  \ ,\nonumber \\
  \psi_{\overline{\Gamma}/ \overline{\gamma}} = \alpha_2\alpha_3+\alpha_2\alpha_4+\alpha_3\alpha_4 \ , & &  \phi_{\overline{\Gamma}/ \overline{\gamma}} = \alpha_2\alpha_3\alpha_4  \ ,\nonumber  \\
  \psi_{\overline{\Gamma}} = \psi_{\overline{\gamma}}\psi_{\overline{\Gamma}/\overline{\gamma}} + (\alpha_3+\alpha_4)\alpha_5\alpha_6\alpha_7 \ , &&
  \phi_{\overline{\Gamma}} = \psi_{\overline{\gamma}}\phi_{\overline{\Gamma}/\overline{\gamma}} + \alpha_3\alpha_4\alpha_5\alpha_6\alpha_7 \ . \nonumber
 \end{eqnarray}
Note that  in general, the graph polynomial of $\overline{\Gamma}$ is obtained from that of $\Gamma$ by setting the Schwinger parameters of the  squashed edges to $0$.
  \end{example}

\subsection{The single-scale Hopf algebra in the general case}
The structure of the Hopf algebra is particularly simple:  quadratic subdivergences  have no non-trivial overlaps with any other subdivergences and can be separated off.

\begin{lem} \label{lemqhasnoepsilon} Let $q_1 \subsetneq q_2$ be nested quadratic subdivergences. Then $\varepsilon_{q_2}(q_1)=0$.
\end{lem} 
\begin{proof}
Let  $v_1,v_2$ denote the two 3-valent vertices of $q_2$ and denote its connectors by $e_1,e_2$.  If $\varepsilon_{q_2}(q_1)=1$,  there exists a path from $v_1$ to $v_2$ which is contained in $q_1$ and so  $q_1$ contains $v_1,v_2$. Since $q_1$ has exactly two 3 valent-vertices (which must be $v_1,v_2$) and  all  other vertices of $q_1$ are 4-valent   there can be no edge $e$ of $q_2$ which does not already lie in $q_1$, and therefore $q_1=q_2$.   So $q_1\subsetneq q_2$ implies 
that $\varepsilon_{q_2}(q_1)=0$.
\end{proof}

Quadratic subdivergences are always tadpole-free.

 \begin{lem}\label{lemIandqnesting} Let $I$ be  a connected,  divergent 1PI subgraph of $G$. Then the  following  two statements  are equivalent:
 \begin{enumerate}
\item  $I$ is tadpole-free $(\ref{tadpolefree})$ , 
\item For all  connected 1PI quadratic  subgraphs $q$  of  $G$, one of the following holds:
\begin{equation} \label{Inestedq}
\left\{
\begin{array}{lll}
  (i)&  I \cap q = \emptyset \ ,   \\
 (ii) &   I \supseteq q  \ , \\
 (iii) &   I \subsetneq q \hbox{ and } \varepsilon_q(I)=0    \ .
\end{array}
\right.
 \end{equation}
\end{enumerate} 
 \end{lem}
 \begin{proof} Clearly $(2)$ implies $(1)$.  In the other direction, let $q$ be a connected quadratic  1PI subgraph of $G$.  Denote its connectors by $\{e_1,e_2\}$. Since $I$ is 1PI, the intersection
 $I\cap \{e_1,e_2\}$  consists of  $0$ or $2$ elements. Suppose that  $I\cap \{e_1,e_2\}=\emptyset$.  Since $I$ is connected,  either $I\cap q=\emptyset$ or  $I$ is contained in $q$, and  $(2)$ holds. 
Now suppose that  $e_1,e_2\in I$.   In this case, write $I$ as a two-edge join 
 $I= I_q \cup \{e_1,e_2\} \cup A$, where $I_q=I\cap q$ and $A$ is defined to be the complement of $q\cup \{e_1,e_2\}$ in $I$. Suppose that    $I_q$ is strictly contained in $q$, 
 for otherwise, $I\supseteq q$ and $(2)$ holds. By power counting,  
 \begin{equation}\label{sdinproof}
 sd(A)+sd(I_q)=sd(I)\geq 0\ .
 \end{equation}
   If $I_q$ is quadratic, then since $I$ is connected and  contains the connectors of $e_1,e_2$,  $\varepsilon_q(I_q)=1$, which contradicts lemma \ref{lemqhasnoepsilon}. Therefore $I_q$ is at most logarithmically divergent  which implies  by $(\ref{sdinproof})$ that $sd(A)\geq 0$. Therefore the graph  $Q=I \cup q$   is   quadratically divergent, since $sd(Q)=sd(A)+sd(q)\geq 1$, and therefore $sd(Q)=1$.     But $Q$ strictly contains $I$ and satisfies $\varepsilon_Q(I)=1$, which contradicts $(1)$. 
 Therefore we must have had $I_q=q$, which completes the proof.   \end{proof}

\begin{cor}  \label{qdontoverlap} Let $q_1,q_2$ be connected $1PI$ quadratic subdivergences in $G$. Then they cannot overlap non-trivially: either $q_1 \cap q_2=\emptyset$, $q_1\subseteq q_2$,  or  $q_2 \subseteq q_1$. 
\end{cor}
\begin{proof}
This follows immediately from  lemmas \ref{lemqhasnoepsilon} and \ref{lemIandqnesting}.
  \end{proof}

One can deduce from this the fact that 
 $\Delta$ is coassociative.
 
\subsection{Inflating  a squashed graph}
The Hopf algebra is phrased in terms of the original graph $G$, but the Feynman rules are expressed in terms of the squashed graph  $\overline{G}$.  To compute the poles of the forms $(\ref{quadforms})$ in terms of $H_G$, we therefore require a correspondence between subgraphs of $\overline{G}$
 and certain subgraphs of $G$. 
 
 \begin{defn} Let $G$ be as above. 
 If  $I\subset \overline{G}$, let $Q_I$ denote the set of connected 1PI quadratic subgraphs $q\subset G$ such that $\varepsilon_{\overline{q}}(I)=1$.  Define $\Iup$ to be the smallest subgraph of $G$ which contains $I$ and satisfies
  \begin{equation} \label{Iupdef}
 | \Iup \cap \{e_q,f_q\} | \in \{ 0 ,   2\} \quad  \hbox{ for all } q \in Q_I \ .
 \end{equation} 
 It is clear  that  $\overline{\Iup}=I$,  but note that  $\Iup$ is not necessarily  connected, even if $I$ is. We shall call $\Iup$ the \emph{inflation} of the graph $I \subset \overline{G}$.
  \end{defn}

 For divergent graphs, the condition $(\ref{Iupdef})$ is equivalent to being $1PI$.
 \begin{lem} \label{inflationis1PI} Let $J\subset G$ be divergent. Then 
 $$J=(\overline{J})^{\ell} \quad \Leftrightarrow  \quad  J  \hbox{ is  1PI} \ .$$
  \end{lem}
 \begin{proof}  Suppose that $J \subset G$ is divergent but not 1PI. Then there is an edge $e\in J$ such that $J\backslash e$ has two components, $A$ and $B$. Since $sd(J)\geq 0$, it follows  that at least one component, say $A$, is quadratic. We can assume that $A$ is 1PI connected.  Necessarily $\varepsilon_A(J)=1$, but $J$ only contains one connector of $A$, namely the edge $e$, which violates $(\ref{Iupdef})$.  Therefore if $J$ is of the form  $I^{\ell}$ for some $I$, then it is $1PI$. Conversely, suppose that $J$ is 1PI. Then 
 $J$ contains $\overline{J}$ and is the minimal such graph satisfying condition $(\ref{Iupdef})$, since if there is a quadratic subgraph  $q$ such that the intersection $J\cap \{e_q,f_q\}$ has exactly one element $e$,  then $J\backslash e$  has one more component than $J$.
 \end{proof}

\begin{rem} \label{remdecomp} Suppose that $Q\subset G$ is a connected 1PI quadratic subgraph such that $\varepsilon_Q(\Iup)=1$ but $  \Iup \cap \{e_Q,f_Q\} =\emptyset$. If $\Iup$ is not contained in $Q$ then $\Iup$ has at least two  components: $\Iup= \Iup_Q \cup \Iup_{Q^c}$ where $\Iup_Q=\Iup\cap Q$ and $\Iup_{Q^c}\subset G\backslash (Q\cup \{e_Q,f_Q\})$.  By repeating this for all such subgraphs $Q$, we obtain a decomposition of $\Iup$ into  disjoint  (but not necessarily connected) pieces $\Iup= \cup_i I^{\ell,i}$. 
\end{rem}

\begin{lem} \label{lemliftepsilon}
For any $I\subset \overline{G}$, and $A\subset G$ a single-scale subgraph,  we have $\varepsilon_{\overline{A}}(I)= \varepsilon_{\overline{A}^{\ell}}(\Iup)$. 
Let $Q, \Iup_Q, \Iup_{Q^c}$ be as in remark \ref{remdecomp}. Then for any   connected quadratic 1PI subgraph   $q$ of $G$ we have
\begin{equation} \label{Qepsilonadditive}
\varepsilon_q( \Iup) = \varepsilon_q(\Iup_Q) + \varepsilon_q(\Iup_{Q^c})\ .
\end{equation} 
\end{lem} 
\begin{proof} For the first part, it is obvious that $\varepsilon_{A}(\Iup)=1$ implies that $\varepsilon_{\overline{A}}(I)=1$. In the opposite direction, suppose that
$\varepsilon_{\overline{A}}(I)=1$.  Then there exists a path $\gamma$ connecting the two distinguished vertices of $\overline{A}$. If $q$  is quadratic  such that $e_q\in \overline I$ and such that
$\varepsilon_{\overline{q}}(\gamma\cap \overline{q})=1$  then we must have $e_q \in \gamma$. Hence $\gamma \cup \{e_q,f_q\}$ is a path in $I\cup \{f_q\}$ which connects the distinguished vertices of $A\cup \{f_q\}$. Continuing by induction it follows that $\gamma^{\ell}\subset \Iup$ connects the distinguished  vertices of $A^{\ell}$.

For the second part,  we know by corollary \ref{qdontoverlap} that either $q\cap Q=\emptyset$, $Q\subsetneq q$, or $q \subseteq Q$. In the first two cases, clearly $\varepsilon_q(\Iup_{Q})=0$ and we have
$\varepsilon_q(\Iup)=\varepsilon_q(\Iup_{Q^c})$. In the third case we have $\varepsilon_q(\Iup_{Q^c})=0$ and $\varepsilon_q(\Iup)=\varepsilon_q(\Iup_Q)$. 
\end{proof}

 \subsection{Location of the poles}
 We compute the orders of the poles of the forms $(\ref{quadforms})$ along the exceptional divisors $\mathcal{E}_I$ where $I\subset E(\overline{G})$, and relate this to 
 $\Iup$.
 \begin{lem} Let $G$ be as above, and let $I$ be a connected 1PI subset of edges of $\overline{G}$. Let $F$ be a flag of (tadpole-free) divergent subgraphs of $G$, and let $\gamma_F$ be defined by 
 $(\ref{gammaflag})$. Then $\omega_{\gamma_F \otimes G/\gamma_F}(s)$ has a pole along $\mathcal{E}_I$ of order $1+p(I)$, where  
  \begin{equation} \label{viformulaq}
p(I)= 2\, h_{I_{\overline{\gamma}_F}}+2\, h_{I_{\overline{G}/\overline{\gamma}_F}} - |I|  -\varepsilon_{\overline{\gamma}_F\otimes \overline{G}/\overline{\gamma}_F}(I)-  \sum_{q} \varepsilon_{\overline{q}}(I) \ ,
\end{equation}  
and where the sum is over all connected 1PI quadratic subgraphs $q$ of $G$.
\end{lem}

  \begin{proof} It follows from $(\ref{PhiRemainder})$ and the definition of $\varepsilon$  that
\begin{equation}\label{valq}
v_I \Big( \prod_q { \phi_{\overline{q}} \over \psi_{\overline{q}} } \Big)=\sum_q  \varepsilon_{\overline{q}}(I) \ .
\end{equation}
By the calculations in \S\ref{sectOrdersofPoles},  the valuation along  $L_I$ is  given (cf $(\ref{valuationformula})$) by
 $$ - v_I\Big({ s\, \phi_{\overline{\Gamma}/\overline{\gamma}_F}  \over  \psi_{\overline{\gamma}_F} \psi_{\overline{\Gamma}/\overline{\gamma}_F}^2 \Upsilon_{\overline{\gamma}_F;\overline{\Gamma}/\overline{\gamma}_F}(s)   } \,  \Omega_{\overline{\Gamma}} \Big)  = 2\, h_{I_{\overline{\gamma}_F}}+2\, h_{I_{\overline{G}/\overline{\gamma}_F}} -\varepsilon_{\overline{\gamma}_F\otimes \overline{\Gamma}/\overline{\gamma}_F} (I)\ .   $$
 Using the fact (corollary \ref{qdontoverlap}) that there is a bijection between quadratic connected 1PI subgraphs of $G$ and those of 
  $\gamma_F, G/\gamma_F$, the result follows from the two previous expressions and  $(\ref{generalpoleorder})$.
 \end{proof}
   
 \begin{lem} Let $G, I, F,\gamma$ be as in the previous lemma.
 Maximally decompose $\Iup$ into components $I^{\ell,1}\cup \ldots \cup I^{\ell,m}$ according to remark \ref{remdecomp}. Then
  \begin{equation} \label{viformula} 
  p(I)= \sum_{i=1}^m   \Big( sd(I^{\ell,i}_{\gamma_F})+sd(I^{\ell,i}_{G/\gamma_F}) -\sum_{I^{\ell,i}\subseteq q}  \varepsilon_q(I^{\ell,i})\Big)\ .
\end{equation}  
 In particular,    $\omega_{\gamma_F\otimes G/\gamma_F}(s)$ has at most simple poles along   exceptional divisors $\mathcal{E}_I$ indexed by $I$  such that
 $\Iup_{\gamma_F}, \Iup_{G/\gamma_F} $ are divergent,   tadpole-free, and $\Iup_{G/\gamma_F}\subsetneq G/\gamma_F$.
 \end{lem}
 
 \begin{proof} By  lemmas   \ref{inflationis1PI} and \ref{lemliftepsilon},   we can write
 $$p(I) = 2h_{I_{\overline{\gamma}_F}}+2h_{I_{\overline{G}/\overline{\gamma}_F}}  - |I|  - \varepsilon_{\gamma_F\otimes G/\gamma_F}(\Iup)  - \sum_q \varepsilon_{q}(\Iup)\ .$$  
   By definition of the inflation map, we have
  $$|I^{\ell,i}|- |\overline{I^{\ell,i}}| =  \sum_{I^{\ell,i}\not\subseteq q} \varepsilon_q(I^{\ell,i})\ .$$ 
   Since inflation does not change the number of loops,  we have
  by $(\ref{Qepsilonadditive})$:
 \begin{eqnarray}
 p(I) &=  &\sum_i \Big(2h_{I^{\ell,i}_{\gamma_F}} +2h_{I^{\ell,i}_{G/\gamma_F}}    - |I^{\ell,i}|  - \sum_{I^{\ell,i} \subseteq q} \varepsilon_{q}(I^{\ell,i})\Big)-  \varepsilon_{\gamma_F\otimes G/\gamma_F}(\Iup) \nonumber \\
  &=  &\sum_i \Big( sd(I^{\ell,i}_{\gamma_F}) +  sd(I^{\ell,i}_{G/\gamma_F})  - \sum_{I^{\ell,i} \subseteq q} \varepsilon_{q}(I^{\ell,i})\Big)-  \varepsilon_{\gamma_F\otimes G/\gamma_F}(\Iup) \ . \nonumber 
  \end{eqnarray}
   Now observe that every term 
  $$  sd(I^{\ell,i}_{\gamma_F}) +  sd(I^{\ell,i}_{G/\gamma_F})  -\sum_{I^{\ell,i}\subseteq q}  \varepsilon_q(I^{\ell,i})$$
  is $\leq 0$. If   $sd(I^{\ell,i}_{\gamma_F})= sd(I^{\ell,i}_{G/\gamma_F}) =0$ this is obvious, and one verifies that if $\{sd(I^{\ell,i}_{\gamma_F}), sd(I^{\ell,i}_{G/\gamma_F})\} =\{0,1\}$ then in each case, $I^{\ell,i}$ is a quadratic 1PI graph, and  $\varepsilon_q(I^{\ell,i})=1$ for $q=I^{\ell,i}$.   One cannot have 
  $sd(I^{\ell,i}_{\gamma_F})=sd(I^{\ell,i}_{G/\gamma_F})=1$.
     Therefore the pole of    $\omega_{\gamma_F\otimes G/\gamma_F}(s)$ along  $\mathcal{E}_I$ 
   is at most simple, and this only happens when every $I^{\ell,i}_{\gamma_F}, I^{\ell,i}_{G/\gamma_F}$ are divergent, $\sum_{I^{\ell,i}\subsetneq q}  \varepsilon_q(I^{\ell,i})=0$,    and $ \varepsilon_{\gamma_F\otimes G/\gamma_F}(\Iup)=0$. In other words, $\Iup$ is a union of tadpole-free graphs, and by assumption $(\ref{noepsilonsquadcase})$, $ \varepsilon_{\gamma_F\otimes G/\gamma_F}(\Iup)$ vanishes in this case provided that $I_{G/\gamma_F}\subsetneq G/\gamma_F$. 
  As in the proof of corollary \ref{corpoleorderalongLI}, there is no pole when  $I_{G/\gamma_F}= G/\gamma_F$. 
    Therefore we have shown that  the poles are simple, and are in one-to-one correspondence with  the set of pairs of  tadpole-free  divergent  subgraphs $\Iup_{\gamma_F}\subset \gamma_F$ and  $\Iup_{G/\gamma_F}\subsetneq G/\gamma_F$.
 \end{proof} 
 Since by lemma \ref{inflationis1PI}, divergent  inflated subgraphs of $G$ are the same as 1PI divergent subgraphs of $G$, the poles of $\omega_{G}$ are indeed indexed by 1PI divergent tadpole-free subgraphs of $G$, and hence by the terms in the definition of the coproduct.

 \subsection{Residues and proof of convergence}
 
  \begin{prop}    Let $F$  be a flag of  divergent subgraphs  in  $G$, and let $I$ be a strict subset of edges in $\overline{G}$. The  form $\omega_{\gamma_F\otimes G/\gamma_F}(s)$ has a simple pole along the exceptional divisor $\mathcal{E}_I$ if and only if $\Iup_{\gamma_F}, \Iup_{G/\gamma_F}$ are divergent 1PI and tadpole-free, and 
  $\Iup_{G/\gamma_F} \subsetneq G/{\gamma_F}$. Then the residue is 
\begin{equation}\label{ResIformula}
\mathrm{Res }_{\mathcal{E}_I} \, \omega_{\gamma_F\otimes G/\gamma_F}(s) =
                                \omega_{\Iup_{\gamma_F} \cup \Iup_{G/\gamma_F}} \otimes \omega_{\gamma_F/\Iup\otimes G/(\gamma_F\cup \Iup)}  (s) \ .
                         \end{equation}
   \end{prop}
 \begin{proof} 
 Let $q$ be any 1PI connected quadratic subgraph of $G$. It follows from $(\ref{PsiRemainder})$ and $(\ref{CircRemainder})$ that to leading order  in the $I$-variables:
 $$ { \phi_{\overline{q}} \over \psi_{\overline{q}} } \To
\left\{
\begin{array}{ll}
  { \phi_{\overline{q}} \over \psi_{\overline{q}} }  \otimes 1  &   \hbox{ if } q \subseteq \Iup   \\
 1\otimes  { \phi_{\overline{q}} \over \psi_{\overline{q}} }    &     \hbox{ if } q \cap \Iup =\emptyset   \\
 1\otimes  { \phi_{\overline{q}/I} \over \psi_{\overline{q}/I} }    &       \hbox{ if }  \Iup \subsetneq  q  \hbox{ and } \varepsilon_q(\Iup) = 0 
\end{array}
\right.
 $$
 where, following lemma \ref{lemgeneralresidues}, the left hand side of the tensor corresponds to the $I$ variables, and the right-hand side to the $I^c$ variables. 
  It follows that  \begin{equation} \label{Rtendstoprod} \prod_{q\subseteq G}   {\phi_{\overline{q}}  \over \psi_{\overline{q}}}\To  \Big(\prod_{q\subseteq \Iup} {\phi_{\overline{q}}  \over \psi_{\overline{q}}}  \Big) \otimes \Big( \prod_{q\subseteq G/\Iup}   {\phi_{\overline{q}}  \over \psi_{\overline{q}}}\Big)
  \end{equation} 
  after applying lemma  \ref{lemIandqnesting} to every connected component of $\Iup$, since every quadratic subdivergence $q$ corresponds to one of the three cases of $(\ref{Inestedq})$, and there is a bijection between $1PI$ connected  quadratic subdivergences of $G$ and those of $\Iup$ and $G/ \Iup$. Now apply $(\ref{Rtendstoprod})$ to the definition $(\ref{quadforms})$.  The result follows by an identical computation to proposition \ref{propResIformula}. 
   \end{proof} 
 By the standard properties of the coproduct, and the argument given in proposition \ref{propresomega2}, for  any flag $F$ of divergent subgraphs of $\Gamma$, 
 the total residue is
\begin{equation}\label{TotalResforq}
 \Res \, \omega_{\gamma_F \otimes \Gamma/\gamma_F } (s)= (\omega^1 \otimes \omega^{23} ) \circ \mu_{13} (\Delta \otimes \Delta) (\gamma_F\otimes \Gamma/\gamma_F) \ ,  \end{equation} 
where $(\omega^1\otimes \omega^{23}) (x\otimes y \otimes z) = \omega_x \otimes \omega_{y\otimes z}(s)\ . $ Note that in this formula, the residue  corresponding to a term $I$ is the residue along the divisor indexed by the squashed graph $\overline{I}$.
 By  $(\ref{1tensgammaforq})$ and 
$(\ref{TotalResforq})$, the forms $(\ref{quadforms})$ satisfy our hypotheses  $(\ref{axioms})$ for renormalization. Since the denominators $\psi_{\overline{q}}$ of the  quadratic correction factors are polynomials with positive coefficients, we conclude that: 
  
  \begin{thm} \label{thmquadconv} The form $\omega^{\ren}_{\Gamma}(s)$ has no poles along any exceptional divisors $\mathcal{E}_I$, and therefore
 the renormalized integral $(\ref{frenq})$  is convergent.
  \end{thm} 
  The proof is similar to the proof of theorem \ref{thmconv2}.
  \subsection{Renormalization group equations: quadratic case} \label{sectrenormgroupquad}
  The proof of the group equations is essentially identical to \S\ref{sectrenormgroup}. We summarize the main steps. For any graph $\gamma \in H_G$ define the quadratic correction factor by
  \begin{equation}\label{Qcorrectionfactors}
  Q_{\gamma} = \prod_{q\subseteq \gamma} {\phi_{\overline{q}} \over \psi_{\overline{q}}}
  \end{equation} 
  where the product is over all 1PI connected quadratic subgraphs $q\subseteq \gamma$. In addition to the data of the forms $(\ref{quadforms})$, define for any $\lambda>0$, 
  \begin{eqnarray}
  \pomega_{\gamma\otimes \Gamma}(s) & = & Q_{\gamma} Q_{\Gamma}  \, {\phi_{\Gamma} \over \psi_{\Gamma}}  {\lambda \over ( \psi_{\gamma  \cup \Gamma} \, \lambda +  \phi_{\gamma \cup \Gamma} )^2 } \,  \prod_{e\in E(\gamma) \cup E(\Gamma)}  d\alpha_e \ , \\
  \omega^{(4)}_{\gamma_1\otimes \gamma_2 \otimes \gamma_3 \otimes \gamma_4}(s) &   = & Q_{\gamma_1}Q_{\gamma_2}Q_{\gamma_3} Q_{\gamma_4}\,   {\phi_{\gamma_2}\over \psi_{\gamma_2}} {\phi_{\gamma_4} \over \psi_{\gamma_4}} 
{ \Omega_{\gamma_1\cup \ldots \cup \gamma_4} \over \Upsilon^2_{\gamma_1\cup \gamma_2 \cup \gamma_3; \gamma_4}(s) }\ .
   \end{eqnarray}
   and extend by linearity to $H\otimes_{\Q} H$ and $H^{\otimes 4}$ respectively.
  The proof is now identical to the log-divergent case since the correction terms $(\ref{Qcorrectionfactors})$ completely factor out of every equation.   In brief: the forms $\pomega^{\ren}(\lambda)$ are pole-free along the $\mathcal{E}_I$,  and we have
$$\int_{\DD_{\overline{G}}} \omega^{\ren}_{G}(1)=  \int_{[0,\infty]^{E(\overline{G})}} \pomega^{\ren}_G ( \lambda)\ .$$
as before. The proof is the same as lemma \ref{lemminormonotone}. The analogue of  proposition \ref{prop1} is 
$$
\int_{\DD_{\overline{G}_1}} \omega^{\ren}_{G_1}(1) \times \int_{\DD_{\overline{G}_2}} \omega^{\ren}_{G_2}(s)  = \int_{\DD_{\overline{G}_1\cup \overline{G}_2}}  \omega^{(4)}_{R(G_1)\otimes R(G_2)}(s) $$
Finally, it is clear from the definition and the property $Q_{\gamma_1\cup \gamma_2}= Q_{\gamma_1}Q_{\gamma_2}$ that  
$$ \sum_{1\leq  i \leq  n}  \omega^{(4)}_{a_1\ldots a_{i-1} \otimes a_i \otimes a_{i+1}\ldots a_n \otimes b}(s) = s {\partial \over \partial s} {\omega}_{a_1\ldots a_n \otimes b}(s)\ ,$$
as in $(\ref{Jsum}).$
Thus the proof of proposition \ref{propprerenormgroup} goes through as before and we obtain
\begin{thm} The Feynman rules $(\ref{quadforms})$  define a cocharacter on $H_G$. Equivalently, the renormalization group equations hold for graphs with quadratic subdivergences:
$$
\int_{\DD_{\overline{\Gamma}}} s {\partial \over \partial s} {\omega}^{\ren}_{\Gamma}(s) =  \sum_{\gamma\subset \Gamma} \, \int_{\DD_{\overline{\gamma}}} \omega_{\gamma}^{\ren} (1)  \times \int_{\DD_{\overline{\Gamma}/\overline{\gamma}}}  \omega_{\Gamma/\gamma}^{\ren} (s) \ . 
$$
for all $\Gamma\in H_G$. \end{thm}   
The first part of the theorem follows from the second, by proposition \ref{propHomifdiff}.
\begin{rem}
These are still the renormalization group equations for massless diagrams. In the presence of masses the result 
Thm. \ref{finalr} is a polynomial in squared masses. Repeating the analysis term by term for the coefficient functions verifies
the full Callan-Symanzik equations, as expected.
\end{rem}
 \subsection{The tangent motive in the general case} The definition of the mixed Hodge structure corresponding to $G$ in the general case is similar to the case when $G$ has at most logarithmic subdivergences,  if we include  extra hypersurfaces $X_{q}$ for each quadratic 1PI connected subgraph $q$ in the definition of $X^{tot}_{G,s}$. Then only the differential form $\omega^{\ren}_G(s)$ changes. Thus we can define $\mathrm{Mot}(G)$ via $(\ref{motdef})$, define framings
 in the same manner,  and conclude that $f_G$ is its period, in all cases.

\section{Reminders on  Hopf algebras} 
\subsection{Basic definitions}
Consider any commutative, graded Hopf algebra 
$$H=\bigoplus_{n\geq 0} H_n$$
over a field $k$ of characteristic zero, where  $H_0=k$.  Denote the  multiplication by $\mu_2:H\otimes_k H \rightarrow H$,  and more generally let
$\mu_n : H^{\otimes n} \rightarrow H$ denote the multiplication of $n$ elements.  We frequently write $y_1\ldots y_n$ as a shorthand for 
$\mu_n(y_1\otimes \ldots \otimes y_n)$ for simplicity.
If the coproduct is denoted  $\Delta: H \rightarrow H\otimes_k H$, the reduced coproduct is defined by  $\Delta'$, where $\Delta'=\Delta-1\otimes id - id \otimes 1$. 
It satisfies
\begin{equation} \label{Deltaprimegraded} \Delta' (H_n ) \subseteq \bigoplus_{p+q=n,p\geq 1, q\geq 1} H_p \otimes_k H_q \ .
\end{equation}
For $n\geq 1$, consider the maps
$$\Delta^{(n)} : H\rightarrow H^{\otimes n+1}  $$
obtained  by setting $\Delta^{(1)}=\Delta'$ and  iterating the reduced coproduct:
$$\Delta^{(n)}= (id \otimes \Delta^{(n-1)})\circ \Delta' = (\Delta^{(n-1)}\otimes id) \circ \Delta' \quad  \hbox{ for   } n\geq 2 \ ,$$
The maps $\Delta^{(n)}$ are    well-defined by the coassociativity of $\Delta$. 
For any element $x\in H$, we shall sometimes use the following version of   Sweedler's notation and write
\begin{equation}\label{Sweedler}
\Delta^{(n)} (x) = \sum_{(x)} x^{(1)} \otimes \ldots \otimes x^{(n+1)}\ ,
\end{equation}
where by $(\ref{Deltaprimegraded})$, all elements $x^{(i)}$ have  degree $\geq 1$. 

\begin{defn} The \emph{coradical filtration} is the increasing filtration defined by 
$$H^{(i)} = \{ x: \Delta^{(i)} x= 0\}$$
where $i\geq 1$, and $H^{(0)}=k$. The set of \emph{primitive} elements in $H$ are the elements in $ H^{(1)}$, i.e.,  which satisfy $\Delta x= 1\otimes x+ x\otimes1$. \end{defn} 
 The Hopf algebras we consider in this paper (Hopf algebras of graphs or trees) are graded with respect to the coradical filtration, i.e.,  $\Delta$ is homogeneous 
 with respect to the grading associated to the coradical filtration.

\subsection{The preparation map}
\begin{defn} \label{defnR}  The \emph{preparation map} $R:H \rightarrow H \otimes_k H$ is defined  by 
\begin{equation} \label{defneqnR} R=1\otimes id + \sum_{n\geq 1} (-1)^n (\mu_n \otimes id ) \Delta^{(n)}\  ,
\end{equation}
where $\mu_1=id$.
It is well-defined since  by $(\ref{Deltaprimegraded})$ the sum on the right-hand side terminates when applied to 
any element in $H$.
\end{defn}
  It is convenient to set  $R^0=1\otimes id$, and  $R^n=   (-1)^n (\mu_n \otimes id ) \Delta^{(n)}$ for all $n\geq 1$. By Sweedler's notation $(\ref{Sweedler})$, we can write for $n\geq 1$ and $x\in H$, 
\begin{equation} \label{RinSweedler} R^n (x) = (-1)^n\sum_{(x)} x^{(1)}\ldots x^{(n)}\otimes x^{(n+1)}  \ .
\end{equation} 
One can also define $R$ recursively in the following way.

\begin{lem} Using Sweedler's notation $(\ref{Sweedler})$, the map $R$ satisfies
\begin{equation} \label{recursiveR} R(x) = 1 \otimes x - \sum_{(x)} \mu( R(x^{(1)}))  \otimes x^{(2)}   \  . 
\end{equation}  
\end{lem}
\begin{proof} This follows immediately from $R^n=-(\mu \circ R^{n-1}\otimes id)\circ \Delta'$.
\end{proof}
\noindent
Since $H$ is commutative and graded, the antipode $S:H\rightarrow H$ is 
$S= -\mu \circ R.$

\subsection{Renormalization property of $R$}
Consider the map defined by:
\begin{eqnarray}  \mu_{13}: H^{\otimes 4} & \To  & H^{\otimes 3} \label{mu13def} \\
x_1\otimes x_2\otimes x_3 \otimes x_4 & \mapsto & x_1 x_3 \otimes x_2 \otimes x_4 \nonumber 
\end{eqnarray}
The following theorem is the main mechansim for the cancellation of poles in  renormalization. For want of a suitable reference,
we give a complete proof here.

\begin{thm}  \label{thmRmainproperty} The preparation map satisfies the following equation:
\begin{equation} \label{eqnRmainproperty} 
\mu_{13} \circ ( \Delta \otimes  (\Delta -id \otimes 1)) \circ R = 1 \otimes R \ .
\end{equation}
  \end{thm} 
  \begin{proof} By definition $(\ref{defneqnR})$, the left-hand side of $(\ref{eqnRmainproperty})$ is given by
   $$   \sum_{n \geq 0} (-1)^n  \mu_{13} ( \Delta \mu_n  \otimes  (\Delta -id \otimes 1)) \circ   \Delta^{(n)}    \ ,$$
   where $\mu_0=id$ and $\Delta^{(0)}=1\otimes id$. 
 Since $\Delta$ is a homomorphism, this can be rewritten: 
   \begin{equation}Ê \label{inpfscndeqforR}
     \sum_{n \geq 0} (-1)^n  \widetilde{\mu}_n ( \Delta^{\otimes n} \otimes  (\Delta -id \otimes 1)) \circ   \Delta^{(n)}   \ , 
     \end{equation} 
   where $\widetilde{\mu}_n:H^{\otimes 2n} \rightarrow H^{\otimes 3}$ is the map defined  by
   $$\widetilde{\mu}_n (y_1 \otimes  \ldots \otimes y_{2n} ) = y_1y_3\ldots y_{2n-1} \otimes y_2y_4 \ldots y_{2n-2} \otimes y_{2n}\ .$$
 Now let $x\in H$, and  write out the following  terms 
   $$  ( \Delta^{\otimes n} \otimes  (\Delta -id \otimes 1)) \circ   \Delta^{(n)}  (x) \qquad = \qquad \qquad \qquad$$
   \begin{equation}  \label{expandedDelta}
  \qquad \qquad \qquad  (( \Delta'+ id \otimes 1 + 1 \otimes id ) ^{\otimes n} \otimes  (\Delta' + 1\otimes id)) \circ   \Delta^{(n)}  (x)  \end{equation}
 using Sweedler's notation $(\ref{Sweedler})$.  When  $n=0$, $(\ref{expandedDelta})$  gives 
  \begin{equation}\label{nis0twoterms} 
  1\otimes 1\otimes 1\otimes x + \sum_{(x)} 1 \otimes 1\otimes x^{(1)} \otimes x^{(2)}\ ,
  \end{equation}
 and   when $n=1$ $(\ref{expandedDelta})$  has the following six terms:
  \begin{eqnarray} \label{nis1sixterms} && \sum_{(x)} \,\, (1\otimes x^{(1)}  \otimes 1\otimes x^{(2)}  + 1 \otimes x^{(1)} \otimes x^{(2)} \otimes x^{(3)}  +
    x^{(1)} \otimes 1  \otimes 1\otimes x^{(2)} +    \\ 
     && x^{(1)}  \otimes 1\otimes x^{(2)} \otimes x^{(3)} +  x^{(1)}  \otimes x^{(2)}\otimes 1  \otimes x^{(3)} + x^{(1)}  \otimes x^{(2)}\otimes x^{(3)} \otimes x^{(4)} )  \nonumber 
       \end{eqnarray}
  The general expression can be encoded by an alphabet with three letters $a,b,c$, which correspond respectively to 
   the maps $1\otimes id$, $id \otimes 1$, and $\Delta'$. For any word in the letters $\{a,b,c\}$ of length $n$ not ending in $b$,  we associate a term in  $(\ref{expandedDelta})$. For this, consider the unique morphism (for the concatenation product) of monoids 
   $$\phi_1: \{a,b,c\}^{\times} \rightarrow \{1,x\}^{\times}$$
   such that $\phi_1(a)=1x$, $\phi_1(b)=x1$, and $\phi_1(c)=xx$. Now  let $\phi_2$ be the map which inserts a tensor in between all letters of a word
   in the letters $\{1,x\}$, and adds superscripts to all letters $x$ in strictly increasing order.  Finally, if $w\in \{a,b,c\}^{\times}$ is a word with at least 2 letters, set 
     $\phi(w)=\phi_2\circ \phi_1(w)$. We set $\phi(a)=1\otimes 1 \otimes 1 \otimes x_1$ and $\phi(c)=1\otimes 1 \otimes x^{(1)} \otimes x^{(2)}$ to agree with  $(\ref{nis0twoterms})$.  Hence  for $n\geq 2$  we have
   \begin{equation} \label{phian}
   \phi(a^n) =1\otimes x^{(1)}\otimes 1 \otimes x^{(2)} \otimes \ldots \otimes 1 \otimes x^{(n)}
   \end{equation} 
and    the six terms of $(\ref{nis1sixterms})$ are $\phi(aa), \phi(ac),\phi(ba),\phi(bc),\phi(ca), \phi(cc)$ in order.
    With these definitions,  the full expansion  of  $(\ref{expandedDelta})$  in Sweedler notation is just
   $$  \phi ((a+b+c)^n (a+c )) \ . $$
   Therefore consider the      non-commutative formal  power series
   $$T= \sum_{n\geq 0} (-1)^n (a+b+c)^n (a+c )\in \Q\langle\langle  a,b,c\rangle\rangle \ . $$
We have shown that $(\ref{inpfscndeqforR})$ is $\widetilde{\mu} \circ\phi(T)$. 
   Now it is easy to verify from the definitions that  for all (possibly empty) words $w, w' \in \{a,b,c\}^{\times}$, 
   $$\widetilde{\mu}\circ \phi( w\, c \, w') = \widetilde{\mu}\circ \phi( w\, b\,a \,w')\ . $$
  It follows that $\widetilde{\mu}\circ \phi(T)=\widetilde{\mu} \circ \phi(\overline{T})$, where 
   $$\overline{T}= \sum_{n\geq 0} (-1)^n (a+b+ba)^n (a+ba )\in \Q\langle\langle  a,b\rangle\rangle \ . $$  
  This satisfies $ (1+a+b+ba) \overline{T} = (a+ba) $, which is $ (1+b) (1+a) \overline{T} = (1+b) a.$
 Since $1+b$  (resp. $1+a$) is invertible as a  non-commutative formal power series in  $\Q\langle \langle a,b\rangle \rangle$, this implies that  $(1+a)\overline{T} = a$, which has
 the unique solution
 $$\overline{T}= \sum_{n\geq 1}(-1)^{n+1} a^n \ . $$
  Therefore, by  $(\ref{phian})$,   $\widetilde{\mu} \circ \phi( \overline{T})$   reduces to 
   $$\sum_{n \geq 0} (-1)^{n} \widetilde{\mu} \circ\phi( a^{n+1})= 1\otimes 1\otimes x+  \sum_{n \geq 1}\sum_{(x)} (-1)^{n}\,   1\otimes x^{(1)} \ldots x^{(n)} \otimes x^{(n+1)}$$
     which is exactly $1\otimes R(x)$ by $(\ref{RinSweedler})$. Since $H$ is graded commutative,  all the previous formal power series arguments are in fact finite series when evaluated on a specific element $x$ of $H$, and so this 
    proves that  $(\ref{inpfscndeqforR})$ is equal to $1\otimes R$.
    \end{proof}

%\begin{rem} One can also prove this inductively   with respect to the coradical filtration using  $(\ref{recursiveR})$ and the fact that $S$ is an anti-cohomomorphism for $\Delta$.
%\end{rem}

\subsection{Properties for the renormalization group equations}

\begin{lem} \label{proponR} For any $m,n\geq 0$,  the following diagram commutes:
\begin{equation} \label{commdiagforR}
\xymatrix{ H   \ar[d]^{R^{m+n+1}} \ar[r]^{\Delta'} & H^{\otimes 2} \ar[d]^{R^{m}\otimes R^{n}} \\
H^{\otimes 2} &   \ar[l]_{-\mu_3\otimes id}  H^{\otimes 4}} 
\end{equation}
\end{lem}
\begin{proof} In Sweedler's notation $(\ref{RinSweedler})$, we have:
$$R^{m+n+1}(x) = \sum_{(x)} (-1)^{m+n+1} x^{(1)}\ldots x^{(m+n+1)}\otimes x^{(m+n+2)}\ .$$
On the other hand, $(R^{m}\otimes R^{n})\circ \Delta' (x)$ can be written
$$\sum_{(x)} (-1)^{m+n} x^{(1)}\ldots x^{(m)}\otimes x^{(m+1)}\otimes  x^{(m+2)}\ldots x^{(m+n+1)}\otimes x^{(m+n+2)}\ ,$$
which is mapped to $- R^{m+n+1}(x) $ on applying $\mu_3\otimes id$.
\end{proof}

We also need the following characterization of additive cocharacters on $H$.
Let $K$ be a field containing $k$,  and let $K[L]$ be the Hopf algebra of $K$-valued functions on  the additive group $\mathbb{G}_a$, where  $L$ is primitive. If $g\in K[L]$, write
$g'(L)$ (resp. $g''(L)$), 
for the usual  derivative (resp. second derivative) with respect to $L$.  
Since $H$ and $K[L]$ are connected, they have counits (or augmentations) which  we denote by 
$\varepsilon: H \rightarrow H_0 \cong k$ and $\varepsilon : K[L] \rightarrow K$.

\begin{prop} \label{propHomifdiff}  Suppose that  $H$ is graded for its  coradical filtration. Then a homomorphism  $f: H \rightarrow K[L]$  of  augmented algebras is a homomorphism of  Hopf algebras 
 if and only if 
  \begin{equation} \label{f''equation}  f''_x(L)  = \sum_{(x)}   f'_{x^{(1)}}(0)  f'_{x^{(2)}}(L) \ , 
\end{equation}
where $\Delta'(x)=\sum_{(x)} x^{(1)}  \otimes x^{(2)} $, and the dependence on $x$ is written in subscript.
\end{prop}
\begin{proof} Let us denote the Lie coalgebras of $H$ and $K[L]$ by 
$$\cH={H_{>0} \over H_{>0}H_{>0}} \quad  \hbox{ and } \quad KL \cong {K[L]_{>0} \over K[L]_{>0}K[L]_{>0}}   $$ 
respectively, and let $\overline{f}: \cH \rightarrow KL$ be the map induced by $f$. One verifies using the fact that $H$ is graded for its coradical filtration, that 
the map $f$ is a homomorphism of Hopf algebras if and only if the following diagram commutes:
$$
\xymatrix{ H   \ar[d]^{f} \ar[r]^{} &  \cH\otimes_k H \ar[d]^{\overline{f}\otimes f } \\
  K[L]  \ar[r]^{}  &   KL \otimes_K K[L] } 
$$
where both horizontal arrows are given by the infinitesimal coaction $(\pi\otimes id) \circ \Delta'$, and
$\pi:H_{>0} \rightarrow \cH$ (resp. $\pi:K[L]_{>0} \rightarrow KL$) are the natural quotient maps. The commutativity of the diagram is therefore equivalent to the equation
\begin{equation} \label{coprodtodiffpf}
(\pi \otimes id) \circ \Delta' f_{x} = (\overline{f} \otimes f) \circ (\pi \otimes id) \circ \Delta'(x)\end{equation}
for all $x\in H$. Using the fact that $\Delta'(L^n) =\sum_{1\leq i<n} \binom{n}{i} L^i\otimes L^{n-i}$, we see that
$$ ( \pi  \otimes id)\circ \Delta' ( g(L) ) = L \otimes \big( g'(L) -g'(0)\big) \quad \hbox{ for all } g(L) \in K[L]\ , $$
since the formula is linear and evidently true for $g(L)=L^n$.  Now let  $x\in H$ satisfy $\varepsilon(x)=0$, where $\varepsilon$ is the counit. Using the Sweedler notation $(\ref{Sweedler})$,  and applying the previous equation to the left-hand side of 
$(\ref{coprodtodiffpf})$, we obtain 
\begin{equation} \label{coprodtodiffpf2}
L\otimes \big( f_x'(L) -f_x'(0)\big) = \sum_{(x)} \overline{f}_{ \pi(x^{(1)})}(L) \otimes f_{x^{(2)}} (L) \ .
\end{equation}
As  $\overline{f}_{ \pi(x^{(1)})}=\pi(f_{x^{(1)}}(L))$ is  the linear term  $f'_{x^{(1)}}(0)L$ of $f_{x^{(1)}}(L)$,  $(\ref{coprodtodiffpf2})$ becomes
\begin{equation} \label{coprodtodiffpf3} 
   f_x'(L) -f_x'(0)= \sum_{(x)} f_{x^{(1)}}'(0)  f_{x^{(2)}} (L) \ .
   \end{equation}
   Since $f$ respects the counit, $f_y(L)$ has no constant term in $L$ for any $y$ of degree $\geq 1$, and so $f_y(0)=0$. Thus the constant terms on both sides of 
 of $(\ref{coprodtodiffpf3})$ vanish, and $(\ref{coprodtodiffpf3})$ is equivalent to its derivative with respect to $L$, which is  simply $(\ref{f''equation})$.
 \end{proof}
In the applications, $H$ is a Hopf algebra of  graphs over $k=\Q$, the coefficient field $K$ is $\C$ or $\R$, and $L=\log s$, where $s$ is the chosen 
renormalization scale. 

\bibliographystyle{plain}
\renewcommand\refname{References}

\end{document}